\documentclass[11pt,letterpaper]{article}
\usepackage[T1]{fontenc}
\usepackage{lmodern}
\usepackage[margin=1in]{geometry}
\usepackage{amsmath,amssymb,amsthm,mathtools}
\usepackage{microtype}
\usepackage{xcolor}
\usepackage{graphicx}
\usepackage{tikz}
\usetikzlibrary{arrows.meta,positioning,calc,fit,backgrounds}
\usepackage{booktabs,array}
\usepackage[numbers,sort&compress]{natbib}
\usepackage[colorlinks=true,linkcolor=blue!45!black,citecolor=blue!45!black,urlcolor=blue!45!black]{hyperref}
\hypersetup{
 pdftitle={Geometric R\'enyi divergence and RLD Fisher information under quantum processes without definite causal order},
 pdfauthor={Zixin Huang and Mark M. Wilde}
}
\usepackage{doi}
\usepackage{bookmark}
\usepackage{placeins}
\usepackage{needspace}
\numberwithin{equation}{section}
\newtheorem{theorem}{Theorem}[section]
\newtheorem{lemma}[theorem]{Lemma}
\newtheorem{proposition}[theorem]{Proposition}
\newtheorem{corollary}[theorem]{Corollary}
\theoremstyle{definition}
\newtheorem{definition}[theorem]{Definition}
\newtheorem{example}[theorem]{Example}

\DeclareMathOperator{\Tr}{Tr}
\DeclareMathOperator{\id}{id}
\DeclareMathOperator{\supp}{supp}

\newcommand{\cN}{\mathcal N}
\newcommand{\cM}{\mathcal M}
\newcommand{\cK}{\mathcal K}
\newcommand{\cC}{\mathcal C}
\newcommand{\cW}{\mathcal W}
\newcommand{\cT}{\mathcal T}
\newcommand{\cF}{\mathfrak F}
\newcommand{\cL}{\mathcal L}
\newcommand{\cD}{\mathcal D}
\newcommand{\Dh}{\widehat D}
\newcommand{\Qh}{\widehat Q}
\newcommand{\IR}{\widehat I_F}

\newcommand{\ket}[1]{|#1\rangle}
\newcommand{\bra}[1]{\langle#1|}

\newcommand{\norm}[1]{\left\|#1\right\|}
\newcommand{\normket}[1]{\left\|#1\right\|}

\newcommand{\Ch}{\operatorname{Ch}}
\newcommand{\ind}{\mathrm{gen}}
\newcommand{\err}{\mathrm{err}}
\newcommand{\transpose}{\mathsf T}
\newcommand{\Rains}{\operatorname{Rains}}
\allowdisplaybreaks[1]
\title{Geometric R\'enyi divergence and RLD Fisher information under quantum processes
without definite causal order}
\author{Zixin Huang$^{1}$ \and Mark M. Wilde$^{2}$}
\date{}
\begin{document}
\maketitle
\vspace{-1.6em}
\begin{center}
\small
\textit{$^{1}$School of Science, STEM College, RMIT University,\\
Melbourne, VIC 3000, Australia}\\[0.4em]
\textit{$^{2}$School of Electrical and Computer Engineering, Cornell University,\\
Ithaca, New York 14850, USA}
\end{center}
\vspace{0.3em}
\begin{abstract}
Quantum experiments without a definite causal order offer new ways to
distinguish devices, estimate physical parameters, and communicate. We
prove a common limit on these tasks: the geometric R\'enyi divergence
at the output of a deterministic completely positive process cannot
exceed the sum of the divergences of its individual channel uses.
The inequality holds for every R\'enyi order in $(0,1)\cup(1,2]$ and
any number of independently replaceable channels. Parallel experiments
with separate reference-assisted inputs approach the bound, giving
an exact equality after optimization over processes. We also prove
an additive bound on the right logarithmic derivative Fisher
information of parameter-dependent channels under every fixed
process in this class. A resulting Cram\'er--Rao bound limits locally unbiased estimation
precision without assuming a definite causal order, and an
erasure-probability example attains the bound. The proofs for the geometric
R\'enyi divergence and RLD Fisher information follow
from the requirement that every independent choice of channels
produce a normalized output. We show how this requirement controls
positive local operations, and how small channel perturbations
convert the divergence inequality into a Fisher-information bound,
including for singular channels. Further applications give finite-use
bounds on channel discrimination and exclusion, together with
communication and private-state-generation converses when the
processes preserve the relevant comparison conditions.
\end{abstract}

\begingroup
\small
\setlength{\baselineskip}{11.4pt}
\tableofcontents
\endgroup

\section{Introduction}\label{sec:introduction}

\subsection{Background and motivation}
How much distinguishability can a quantum experiment extract from several
channel uses when their causal order is not specified in advance?
This question appears in channel discrimination, where the experimenter
tries to identify a device, and in communication converses, where a
successful code must distinguish an actual channel from a suitably chosen
comparison channel. A bound that holds for a large class of experiments
can therefore constrain several operational tasks at once. A related
local question asks how accurately the experiment can estimate a
continuous channel parameter. Because Fisher information describes
the distinguishability of nearby parameter values, a suitable
divergence inequality can also bound estimation precision.

For experiments with a definite causal order, a useful approach is to
bound the increase of a divergence after each channel use. Fang and
Fawzi proved a chain rule for the geometric R\'enyi divergence for
$1<\alpha\leq2$ and used it to obtain channel-capacity
converses~\cite[Lemmas 7--8]{FangFawzi2021}. Katariya and Wilde established the
corresponding chain rule for $0<\alpha<1$ and developed applications to
channel estimation and discrimination
\cite[Propositions 45 and 47]{KatariyaWilde2021}. To state the
sequential bound, let $\cN_1,\cM_1\colon A\to B$ and
$\cN_2,\cM_2\colon B\to C$ be quantum channels. Under the first
hypothesis, the input is sent through $\cN_1$ and then $\cN_2$;
under the second, it is sent through $\cM_1$ and then $\cM_2$.
The symbol $\circ$ denotes this sequential composition, with the
rightmost channel acting first. For every
$\alpha\in(0,1)\cup(1,2]$, the bound is
\begin{equation}\label{eq:intro-sequential}
\Dh_\alpha(\cN_2\circ\cN_1\|\cM_2\circ\cM_1)
\leq
\Dh_\alpha(\cN_1\|\cM_1)+\Dh_\alpha(\cN_2\|\cM_2).
\end{equation}
Here $\Dh_\alpha(\cN\|\cM)$ compares the two channels using the
same input state, which may be entangled with a retained reference
system, and optimizes over that input; see
\eqref{eq:channel-divergence}. This is the geometric R\'enyi
instance of the generalized channel divergence introduced by
Leditzky, Kaur, Datta, and Wilde
\cite[Definition II.2]{LeditzkyKaurDattaWilde2018}.
The divergence is especially
useful because it has an explicit expression in terms of Choi operators
and is additive under tensor products.

A sequence of intermediate states is unavailable in a general process
without definite causal order. Supermaps on product channels and the
process-matrix framework provide mathematical descriptions of such
processes~\cite{ChiribellaDArianoPerinotti2008,ChiribellaEtAl2013,OreshkovCostaBrukner2012}. They include
the quantum switch, in which a control system coherently determines the
order of two channel uses. General strategies can improve finite-use
channel-discrimination performance~\cite{BavarescoMuraoQuintino2021}. Therefore a
sequential chain rule cannot simply be applied after choosing an order
that the process does not possess.

Our approach replaces a sequence of intermediate states by a single
normalization requirement: every independent choice of channels for
the slots must produce a quantum channel, or a normalized state when
the experiment has no external input. We represent the process by a
completely positive linear action on the slots' Choi operators and
call it deterministic when this normalization holds.
Section~\ref{sec:processes} gives the precise definition and examples.
The results below show that these two properties already control the
total distinguishability and the RLD Fisher information available
from the channel uses.

\subsection{Summary of results}
We prove in Theorem~\ref{thm:main} that every deterministic completely
positive $n$-slot process $\cW$ satisfies
\begin{equation}\label{eq:intro-main}
\Dh_\alpha\!\left(
 \cW(\cN_1,\ldots,\cN_n)\middle\|
 \cW(\cM_1,\ldots,\cM_n)\right)
\leq\sum_{i=1}^n\Dh_\alpha(\cN_i\|\cM_i),
\qquad
\alpha\in(0,1)\cup(1,2].
\end{equation}
The notation $\cW(\cN_1,\ldots,\cN_n)$ means that each local channel
is inserted into its designated slot. Equivalently, its Choi operator
is obtained by applying one completely positive linear transformation
to $\bigotimes_iJ_{\cN_i}$. The output of $\cW$ can be a channel or a state.
There is no assumption about the order in which the slots are used.

The proof of \eqref{eq:intro-main} rests on a normalization statement that does not mention
divergences. Let $f$ be a positive linear functional normalized to one
on every product of channel Choi operators. For arbitrary positive
semidefinite $X_i$ on $A_iB_i$, Theorem~\ref{thm:normalization} states that
\begin{equation}\label{eq:intro-normalization}
\prod_i\lambda_{\min}(\Tr_{B_i}X_i)
\leq f\!\left(\bigotimes_iX_i\right)
\leq\prod_i\norm{\Tr_{B_i}X_i}.
\end{equation}
The upper bound follows by completing a trace-nonincreasing operation
to a channel. The lower bound requires an additional idea. We normalize
each $X_i$ by its input effect, construct an operation with the inverse
input effect, and use Cauchy--Schwarz to relate their process weights.
This converts the lower bound into the already proved upper bound.
Section~\ref{sec:normalization} gives the construction and its motivation
in detail.

Equation~\eqref{eq:intro-main} implies that the supremum of the output
geometric R\'enyi divergence over all such processes is exactly the sum
of the individual channel divergences
(Corollary~\ref{cor:collapse}). Parallel product probes approach this
value. This equality concerns a divergence, rather than the complete
finite-use discrimination region. We derive binary error constraints,
a strong-converse bound, a symmetric-error bound, and a channel-exclusion
converse in Section~\ref{sec:discrimination}.

For communication tasks, an additional preservation condition is
essential: replacing the resource channels by comparison channels must
produce an output on which the target test has small acceptance
probability. Section~\ref{sec:capacities} makes this condition explicit
and derives entanglement-generation, private-state-generation, and
classical-communication converses. The result does not assume that
every process preserves a selected free class.

A second principal result of our paper concerns quantum estimation.
For channel families $\cN_{i,\theta}$ depending on a single real
parameter and a process $\cW$ independent of that parameter,
Theorem~\ref{thm:rld-process} establishes
\begin{equation}\label{eq:intro-rld}
\IR\!\left(\theta_0;
 \cW(\cN_{1,\theta},\ldots,\cN_{n,\theta})\right)
\leq\sum_{i=1}^n\IR(\theta_0;\cN_{i,\theta}),
\end{equation}
where $\IR$ is the channel right logarithmic derivative (RLD)
Fisher information defined in
\eqref{eq:rld-channel}. The channel families need only be
differentiable at the parameter value $\theta_0$.
The proof applies the order-two divergence inequality to affine
channel paths that preserve each channel's value and derivative.
This gives a pointwise result even for singular Choi operators,
without exchanging an optimization and a local limit.
Corollary~\ref{cor:rld-cramer-rao} then bounds the mean-square error
of every locally unbiased estimator from below by the reciprocal of the
sum in \eqref{eq:intro-rld}. This extends the sequential RLD
estimation converse of Katariya and Wilde
\cite[Conclusion 41]{KatariyaWilde2021} to the process class
considered here. The erasure-probability example in
\eqref{eq:rld-erasure-error}--\eqref{eq:rld-erasure-attainment}
attains the bound using independent parallel measurements.
Our finite-use RLD result permits nonidentical slots and channel
outputs. Section~\ref{sec:estimation} compares it with existing
optimizations and symmetric-logarithmic-derivative bounds for
estimation with general processes~\cite{LiuHuYuanYang2023,WeiLiPang2026}.

The scope of these contributions differs from recent asymptotic
channel-discrimination results. Zhu and Wang give a channel Stein
theorem and a strong-converse-exponent characterization for general
positive testers, with support hypotheses for the exponent
characterization~\cite[Definition 2.3, Theorem 2.4, and
Corollary 6.1]{ZhuWang2026}. Their finite-use comparison uses
regularized sandwiched R\'enyi divergences and a polynomial factor.
Our result is an exact finite-use inequality for geometric divergences,
allows nonidentical slots and channel outputs, and contains no
additional factor. It also supplies the process inequality needed to
extend the geometric channel-exclusion converse of
Ji, Mishra, Mosonyi, and Wilde~\cite{JiEtAl2026} to the positive
process-tester class.

\subsection{Paper organization}
Section~\ref{sec:preliminaries} fixes the notation and presents the
geometric perspective and channel formula. Section~\ref{sec:processes}
defines the process model and explains its normalization.
Section~\ref{sec:normalization} proves the normalization bound in
Theorem~\ref{thm:normalization}.
Section~\ref{sec:main} establishes the main theorem and its extensions.
Sections~\ref{sec:discrimination} and~\ref{sec:capacities} develop
applications. Section~\ref{sec:examples} gives exact examples and
explains limitations of the assumptions. Section~\ref{sec:estimation}
proves the RLD Fisher-information inequality and its Cram\'er--Rao
application. Section~\ref{sec:conclusion} concludes. Appendix~\ref{app:choi} provides details of the Choi formula
and the approximation of optimizing probes. Readers interested
in estimation can proceed from Section~\ref{sec:main} directly
to Section~\ref{sec:estimation}, which introduces its own
estimation preliminaries.

\section{Notation and geometric R\'enyi divergence}\label{sec:preliminaries}
The proof of \eqref{eq:intro-main} needs two descriptions of a channel: its Choi operator and
the distinguishability of its output states. We first fix the operator
notation and identify the input effect that controls the weight of a
positive operation. We then introduce the geometric perspective and
its behavior under completely positive transformations. These facts
will connect the normalization bound in
Theorem~\ref{thm:normalization} to channel distinguishability.

\subsection{States, channels, and input effects}
Throughout we take all Hilbert spaces to be finite-dimensional, and $d_A\coloneqq\dim A$.
For a system~$A$,
let $\cL(A)$ denote the vector space of linear operators on its Hilbert space,
and let $\cD(A)$ denote its set of density operators. The identity is denoted
by $I$, with a system subscript when useful. We write
$\norm{X}$ for the operator norm and $\norm{X}_1$ for the trace norm.
For a vector, $\normket{\ket{v}}$ denotes its Hilbert-space norm.
A projector is a self-adjoint idempotent.
All information quantities and rates use base-two logarithms.
The notation $X\ll Y$ means $\supp X\subseteq\supp Y$.

A quantum channel $\cN\colon A\to B$ is a completely positive,
trace-preserving linear transformation. More generally, a completely
positive transformation $\mathcal X\colon A\to B$ need not preserve
trace. Its adjoint is specified by
$\Tr[Y\mathcal X(X)]=\Tr[\mathcal X^\dagger(Y)X]$.
The positive operator $\mathcal X^\dagger(I_B)$ is its \emph{input
effect}: for an input state $\rho$,
\begin{equation}\label{eq:input-effect}
\Tr[\mathcal X(\rho)]
=\Tr[\rho\,\mathcal X^\dagger(I_B)].
\end{equation}
For a trace-nonincreasing operation, this is its occurrence probability.
For a general completely positive transformation, it is a nonnegative
weight.

Fix an orthonormal basis $\{\ket{j}_A\}_{j=1}^{d_A}$ and a copy
$A'\simeq A$ with the matching basis $\{\ket{j}_{A'}\}_{j=1}^{d_A}$.
We identify the input of $\mathcal X$ with $A'$ when defining the
unnormalized Choi operator
\begin{equation}\label{eq:choi}
J_{\mathcal X}
\coloneqq
(\id_A\otimes\mathcal X_{A'\to B})
(\ket{\Gamma}\bra{\Gamma}_{AA'}),
\qquad
\ket{\Gamma}_{AA'}\coloneqq\sum_{j=1}^{d_A}\ket{j}_A\ket{j}_{A'}.
\end{equation}
Thus $J_{\mathcal X}\geq0$ exactly when $\mathcal X$ is completely
positive, and $\mathcal X$ is a channel exactly when
$\Tr_BJ_{\mathcal X}=I_A$ also holds. We denote the channel Choi set by
\begin{equation}\label{eq:channel-choi-set}
\Ch(A\to B)
\coloneqq\{J\in\cL(AB):J\geq0,\ \Tr_BJ=I_A\}.
\end{equation}
The input effect satisfies
\begin{equation}\label{eq:effect-transpose}
\mathcal X^\dagger(I_B)=(\Tr_BJ_{\mathcal X})^{\transpose},
\end{equation}
where the transpose is taken in the fixed input basis
$\{\ket{j}_A\}_{j=1}^{d_A}$ used in \eqref{eq:choi}; the input
copy $A'$ is identified with $A$ by matching these basis vectors. Consequently, the minimum
and maximum eigenvalues of $\Tr_BJ_{\mathcal X}$ bound the weight in
\eqref{eq:input-effect} for every input state.

\subsection{The operator perspective}
For classical probability distributions, the quantity inside a
R\'enyi divergence is equal to a moment of a likelihood ratio:
$\sum_jy_j(x_j/y_j)^\alpha$ when $y_j>0$.
The following operator expression $P_\alpha(X,Y)$ replaces the scalar ratio $x_j/y_j$ by
$Y^{-1/2}XY^{-1/2}$ and places the result back on the original
space with two factors of $Y^{1/2}$.
For positive definite $Y$ and positive semidefinite $X$, define
\begin{equation}\label{eq:perspective}
P_\alpha(X,Y)
\coloneqq
Y^{1/2}(Y^{-1/2}XY^{-1/2})^\alpha Y^{1/2}.
\end{equation}
If $X$ and $Y$ commute, a common eigenbasis gives
$P_\alpha(X,Y)=X^\alpha Y^{1-\alpha}$ and
$\Tr P_\alpha(X,Y)=\sum_jx_j^\alpha y_j^{1-\alpha}$.
Thus its trace agrees with the classical R\'enyi moment.
Its useful noncommutative property is its behavior under
completely positive transformations in Lemma~\ref{lem:perspective}.
For $0<\alpha<1$, this is the weighted operator geometric mean with
the order of the arguments chosen to match relative entropy. At
singular arguments, we use its continuous extension from above:
\begin{equation}\label{eq:mean-extension}
P_\alpha(X,Y)
\coloneqq
\lim_{\delta\downarrow0}
P_\alpha(X+\delta I,Y+\delta I),
\qquad 0<\alpha<1.
\end{equation}
This extension exists in finite dimensions~\cite{KuboAndo1980}.
Continuity from above refers to these decreasing positive
regularizations; it does not assert continuity along every sequence
approaching a singular pair.
It is important not to
replace the inverse in \eqref{eq:perspective} by a generalized inverse
without checking support inclusion. When $X\not\ll Y$, that
substitution need not agree with \eqref{eq:mean-extension};
see~\cite[Proposition 44]{KatariyaWilde2021}.
For $1<\alpha\leq2$ and $X\ll Y$, Eq.~\eqref{eq:perspective} is evaluated
on $\supp Y$.

For a state $\rho$ and a positive semidefinite operator $\sigma$,
the geometric R\'enyi divergence is defined as
\begin{equation}\label{eq:state-divergence}
\Dh_\alpha(\rho\|\sigma)
\coloneqq\frac{1}{\alpha-1}
\log\Tr[P_\alpha(\rho,\sigma)].
\end{equation}
For $\alpha>1$, its value is $+\infty$ when $\rho\not\ll\sigma$.
For $\alpha<1$, a zero trace in \eqref{eq:state-divergence} also
gives $+\infty$. The same definition applies to nonnormalized
positive comparison operators, without renormalization. This is
the maximal quantum R\'enyi divergence
associated with $t^\alpha$~\cite{Matsumoto2018}.
We call $\Tr P_\alpha(\rho,\sigma)$ the geometric R\'enyi moment.

The following perspective properties connect the preceding definition
to quantum processes.  Their essential feature is that the completely
positive transformation need not preserve either trace or the identity.
We will use the operator Jensen inequality: if $\mathcal U$ is unital
and completely positive, and $g$ is operator convex on an interval
containing the spectrum of a Hermitian operator $T$, then
\begin{equation}\label{eq:operator-jensen}
g(\mathcal U(T))\leq\mathcal U(g(T)).
\end{equation}
For operator-concave $g$, the inequality reverses
\cite[Theorem~2.1]{HansenPedersen2003}.
In finite dimensions, such a transformation has a Kraus representation
$\mathcal U(Z)=\sum_j V_j^\dagger ZV_j$ satisfying
$\sum_jV_j^\dagger V_j=I$. The latter condition is exactly unitality.
Here operator convexity means that the scalar convexity inequality
continues to hold in positive-semidefinite order for matrices of every
finite dimension.  The function $t\mapsto t^\alpha$ is operator concave
on $[0,\infty)$ for $0<\alpha<1$ and operator convex there for
$1<\alpha\leq2$.  These are the two ranges used in the Jensen proof of state data
processing in \cite[Theorem~73]{KatariyaWilde2021}.

\begin{lemma}[Perspective under completely positive transformations]
\label{lem:perspective}
Let $\Psi\colon\cL(A)\to\cL(B)$ be a completely positive linear
transformation between finite-dimensional matrix algebras, and let
$X,Y\geq0$ act on $A$.  Use the perspective $P_\alpha$ in
\eqref{eq:perspective}, with the extension \eqref{eq:mean-extension}
when $0<\alpha<1$.  For $1<\alpha\leq2$ and $X\ll Y$, evaluate
\eqref{eq:perspective} on $\supp Y$ and extend the result by zero on
its orthogonal complement; in particular, $P_\alpha(0,0)=0$.
For $0<\alpha<1$,
\begin{equation}\label{eq:jensen-lower}
P_\alpha(\Psi(X),\Psi(Y))
\geq\Psi(P_\alpha(X,Y)).
\end{equation}
For $1<\alpha\leq2$ and $X\ll Y$, the output pair also satisfies
$\Psi(X)\ll\Psi(Y)$, and
\begin{equation}\label{eq:jensen-upper}
P_\alpha(\Psi(X),\Psi(Y))
\leq\Psi(P_\alpha(X,Y)).
\end{equation}
For every integer $n\geq1$, the perspective also satisfies
\begin{equation}\label{eq:perspective-tensor}
P_\alpha\!\left(\bigotimes_{i=1}^nX_i,
                    \bigotimes_{i=1}^nY_i\right)
=\bigotimes_{i=1}^nP_\alpha(X_i,Y_i).
\end{equation}
Here each $X_i,Y_i$ is positive semidefinite on a finite-dimensional
space.  The range $0<\alpha<1$ permits arbitrary such pairs; the
range $1<\alpha\leq2$ requires $X_i\ll Y_i$ for every $i$.
\end{lemma}
\begin{proof}
Jensen's inequality requires a unital transformation, whereas $\Psi$
need not be unital.  We obtain the required normalization by conjugating
the input with $Y^{1/2}$ and the output with $\Psi(Y)^{-1/2}$.
After applying Jensen, we undo the output conjugation.  Support
restriction handles singular second arguments above one; below one,
we must instead check the prescribed limit from positive operators.

First suppose that $Y>0$.  Every positive operator $Z$ obeys
$Z\leq c_ZY$ for some finite $c_Z$, and positivity of $\Psi$ gives
$\Psi(Z)\leq c_Z\Psi(Y)$.  Thus every such $\Psi(Z)$ is supported
on $\supp\Psi(Y)$.  If $\Psi(Y)=0$, both sides of the claimed
inequalities vanish.  Otherwise, restrict the output space to
$\supp\Psi(Y)$.  This agrees with the full-space perspective in
both ranges: below one, identity regularization splits over this
support and the common zero complement, whose contribution is
$P_\alpha(\delta I,\delta I)=\delta I$, and the positive second
argument on the support converges by functional calculus.  Define
\begin{equation}\label{eq:normalized-jensen}
\mathcal U(Z)
\coloneqq
\Psi(Y)^{-1/2}\Psi(Y^{1/2}ZY^{1/2})\Psi(Y)^{-1/2}.
\end{equation}
The inverse is taken on this output support.  Conjugation by a fixed
operator is completely positive, so $\mathcal U$ is completely
positive as a composition of completely positive transformations.
Inputting the identity operator $I_A$ gives
\begin{equation}\label{eq:normalized-jensen-unital}
\mathcal U(I_A)
=\Psi(Y)^{-1/2}\Psi(Y)\Psi(Y)^{-1/2}
=I_{\supp\Psi(Y)}.
\end{equation}
Consequently, $\mathcal U$ is unital as a transformation into the
matrix algebra on the restricted output space.

Set $T=Y^{-1/2}XY^{-1/2}\geq0$.  The relative operator at the
output is exactly
\begin{equation}\label{eq:relative-operator-image}
\mathcal U(T)=\Psi(Y)^{-1/2}\Psi(X)\Psi(Y)^{-1/2}.
\end{equation}
Jensen's inequality therefore implies that
$\mathcal U(T)^\alpha\geq\mathcal U(T^\alpha)$ for
$0<\alpha<1$ and the reverse inequality for
$1<\alpha\leq2$.  Conjugation by $\Psi(Y)^{1/2}$ preserves
positive-semidefinite order.  The two conjugated expressions are
\begin{align}
\Psi(Y)^{1/2}\mathcal U(T)^\alpha\Psi(Y)^{1/2}
&=P_\alpha(\Psi(X),\Psi(Y)),
\label{eq:jensen-output-identification}\\
\Psi(Y)^{1/2}\mathcal U(T^\alpha)\Psi(Y)^{1/2}
&=\Psi\!\left(Y^{1/2}T^\alpha Y^{1/2}\right)
=\Psi(P_\alpha(X,Y)).
\label{eq:jensen-input-identification}
\end{align}
In \eqref{eq:jensen-input-identification}, the outer support projectors
can be removed because the positive operator on the right is supported
on $\supp\Psi(Y)$.
This proves \eqref{eq:jensen-lower} and \eqref{eq:jensen-upper} when
$Y>0$, without requiring $X$ to be invertible.

Now let $1<\alpha\leq2$ and $X\ll Y$.  If $Y=0$, then $X=0$
and the assertion follows from the stated convention.  Otherwise,
regard $X$ and $Y$ as operators on $\supp Y$, where $Y$ is positive
definite.  Restricting the domain of $\Psi$ to this matrix algebra
means first extending an operator by zero and then applying $\Psi$;
this remains a completely positive transformation.  The argument just
given applies to that restriction.  Moreover, $X\leq cY$ for some
finite $c$, so $\Psi(X)\leq c\Psi(Y)$ and
$\Psi(X)\ll\Psi(Y)$.  Thus the output perspective is precisely
the supported perspective appearing in \eqref{eq:jensen-upper}.

For $0<\alpha<1$, allow arbitrary singular $X,Y\geq0$ and apply
the positive-definite argument to $X+\delta I_A,Y+\delta I_A$.
The right-hand side converges to $\Psi(P_\alpha(X,Y))$, because
\eqref{eq:mean-extension} defines the input limit and a linear
transformation between finite-dimensional spaces is continuous.
The output regularization adds $\delta\Psi(I_A)$, rather than
$\delta I_B$, so its limit needs a separate check.  All positive
images of $\Psi$ are supported on
$S\coloneqq\supp\Psi(I_A)$: indeed,
$0\leq\Psi(Z)\leq\norm{Z}\Psi(I_A)$ for $Z\geq0$.
If $S$ is zero, the assertion follows at once.  Otherwise,
there exist $0<c\leq C<\infty$ such that, on $S$,
\begin{equation}\label{eq:output-regularization-bounds}
cI_S\leq\Psi(I_A)\leq CI_S.
\end{equation}
The weighted geometric mean is increasing in both arguments
\cite{KuboAndo1980}.  Hence, on $S$,
\begin{align}
P_\alpha(\Psi(X)+c\delta I_S,\Psi(Y)+c\delta I_S)
&\leq
P_\alpha(\Psi(X)+\delta\Psi(I_A),
          \Psi(Y)+\delta\Psi(I_A))
\label{eq:output-regularization-squeeze:line-1}\\
&\leq
P_\alpha(\Psi(X)+C\delta I_S,\Psi(Y)+C\delta I_S).
\label{eq:output-regularization-squeeze}
\end{align}
Both outer expressions tend to
$P_\alpha(\Psi(X),\Psi(Y))$ by \eqref{eq:mean-extension}.
Removing the common zero subspace causes no change: regularization
on that subspace contributes $P_\alpha(\delta I,\delta I)=\delta I$,
which tends to zero.  This proves the required output limit and
establishes \eqref{eq:jensen-lower} for all singular pairs.

For the tensor identity in \eqref{eq:perspective-tensor}, first suppose that each $Y_i>0$ and set
$T_i=Y_i^{-1/2}X_iY_i^{-1/2}$.  Functional calculus in product
eigenbases gives
\begin{equation}\label{eq:product-functional-calculus}
\left(\bigotimes_iY_i\right)^{-1/2}
 =\bigotimes_iY_i^{-1/2},
\qquad
\left(\bigotimes_iT_i\right)^\alpha
 =\bigotimes_iT_i^\alpha.
\end{equation}
Substitution into \eqref{eq:perspective} proves
\eqref{eq:perspective-tensor}; zero eigenvalues of the $T_i$ are
allowed because $\alpha>0$.  Above one, the same calculation on
$\bigotimes_i\supp Y_i$ proves the supported singular case.  If
some $Y_i=0$, then $X_i=0$ and both sides vanish.

Below one, write $X=\bigotimes_iX_i$, $Y=\bigotimes_iY_i$ and
\begin{equation}\label{eq:regularized-tensor-pairs}
X_\delta=\bigotimes_i(X_i+\delta I),
\qquad
Y_\delta=\bigotimes_i(Y_i+\delta I).
\end{equation}
For $0<\delta\leq1$, expanding these products gives a finite
constant $C>0$, independent of $\delta$, such that
\begin{equation}\label{eq:tensor-regularization}
X+\delta^nI\leq X_\delta\leq X+C\delta I,
\qquad
Y+\delta^nI\leq Y_\delta\leq Y+C\delta I.
\end{equation}
The lower bounds retain the term in which every factor is
$\delta I$; all remaining terms are positive.  The upper bounds
use $\delta^k\leq\delta$ for $1\leq k\leq n$ and bound each
fixed positive coefficient by its operator norm times the identity.
Monotonicity now squeezes $P_\alpha(X_\delta,Y_\delta)$ between
$P_\alpha(X+\delta^nI,Y+\delta^nI)$ and
$P_\alpha(X+C\delta I,Y+C\delta I)$.  Both converge to
$P_\alpha(X,Y)$.  On the other side of the tensor identity for
the positive-definite local pairs, each local factor converges by
\eqref{eq:mean-extension}.  Taking these finite-dimensional limits
proves \eqref{eq:perspective-tensor} for arbitrary singular pairs.
\end{proof}

The following data-processing inequality for states is known. It follows
from Matsumoto's maximal $f$-divergence framework
\cite[Lemma 3.1 and Theorem 6.1]{Matsumoto2018}, using the
operator-convex generator $f(t)=t^\alpha$ above one and
$f(t)=-t^\alpha$ below one. It is reproved in
\cite[Theorem 73]{KatariyaWilde2021}. We include its short deduction
from Lemma~\ref{lem:perspective} because both data processing and
the resulting normalization bound are used throughout the paper.

\begin{corollary}[Data processing and normalization]
\label{cor:state-dpi}
Let $\rho\in\cD(A)$, let $\sigma\geq0$ act on $A$, and let
$\mathcal T\colon A\to B$ be a quantum channel between
finite-dimensional systems.  For every
$\alpha\in(0,1)\cup(1,2]$, the geometric R\'enyi divergence in
\eqref{eq:state-divergence} satisfies
\begin{equation}\label{eq:state-data-processing}
\Dh_\alpha(\mathcal T(\rho)\|\mathcal T(\sigma))
\leq\Dh_\alpha(\rho\|\sigma).
\end{equation}
The comparison operator $\sigma$ need not be normalized.  In
particular,
\begin{equation}\label{eq:state-normalization-bound}
\Dh_\alpha(\rho\|\sigma)\geq-\log\Tr\sigma,
\end{equation}
where $-\log0=+\infty$.  Thus the divergence is nonnegative
whenever $\Tr\sigma\leq1$.
\end{corollary}
\begin{proof}
If the input divergence is infinite, \eqref{eq:state-data-processing}
holds in the extended-real sense.  Otherwise the input moment is
strictly positive, and above one the required support inclusion holds.
Lemma~\ref{lem:perspective} and trace preservation give
\begin{align}
\Tr P_\alpha(\mathcal T(\rho),\mathcal T(\sigma))
&\geq\Tr P_\alpha(\rho,\sigma),&&0<\alpha<1,
\label{eq:state-moment-data-processing-lower}\\
\Tr P_\alpha(\mathcal T(\rho),\mathcal T(\sigma))
&\leq\Tr P_\alpha(\rho,\sigma),&&1<\alpha\leq2.
\label{eq:state-moment-data-processing-upper}
\end{align}
The output moment is also strictly positive: below one this follows
from the first inequality, and above one it follows because
$\mathcal T(\rho)$ is a nonzero state supported on
$\mathcal T(\sigma)$.  Taking logarithms and dividing by
$\alpha-1$ gives \eqref{eq:state-data-processing}; the negative
denominator below one reverses the moment inequality.

Apply \eqref{eq:state-data-processing} to the trace channel, whose
one-dimensional outputs are $1$ and $s=\Tr\sigma$.
For $s>0$ their divergence is
$(\alpha-1)^{-1}\log s^{1-\alpha}=-\log s$.
For $s=0$, positivity gives $\sigma=0$, and the divergence is
$+\infty$ under the stated conventions.  This proves
\eqref{eq:state-normalization-bound} in both cases.
\end{proof}

\subsection{Channel divergence and its Choi formula}
The generalized channel divergence compares two transformations on
the same input, allowing an arbitrary retained reference system
\cite[Definition~II.2]{LeditzkyKaurDattaWilde2018}.
Let $\cN\colon A\to B$ be a quantum channel and
$\cM\colon A\to B$ a completely positive transformation, with
$A,B$ finite-dimensional.  For
$\alpha\in(0,1)\cup(1,2]$, define
\begin{equation}\label{eq:channel-divergence}
\Dh_\alpha(\cN\|\cM)
\coloneqq
\sup\left\{
\begin{aligned}
&\Dh_\alpha\!\left(
 (\id_R\otimes\cN)(\rho_{RA})
 \middle\|
 (\id_R\otimes\cM)(\rho_{RA})\right):\\
&\rho_{RA}\in\cD(RA),\ R\text{ finite-dimensional}
\end{aligned}
\right\}.
\end{equation}
Most results use channels in both arguments.  If $\cM$ is trace
nonincreasing, the second output is subnormalized; for a general
completely positive $\cM$, it is simply positive.  Both cases are
covered by Corollary~\ref{cor:state-dpi}.

The definition has the equivalent pure-state form
\begin{equation}\label{eq:channel-pure-reduction}
\Dh_\alpha(\cN\|\cM)
=\sup\left\{
\begin{aligned}
&\Dh_\alpha\!\left(
 (\id_R\otimes\cN)(\psi_{RA})
 \middle\|
 (\id_R\otimes\cM)(\psi_{RA})\right):\\
&\psi_{RA}\in\cD(RA)\text{ pure},\ R\simeq A
\end{aligned}
\right\}.
\end{equation}
To prove this reduction, take an arbitrary mixed probe $\rho_{RA}$
from \eqref{eq:channel-divergence} and purify it to
$\ket\varphi_{REA}$.  Across the partition $RE:A$, its Schmidt
rank is at most $d_A$.  Thus there are a system $R_0\simeq A$,
a pure state $\psi_{R_0A}$, and an isometry $V\colon R_0\to RE$
such that
$\ket\varphi_{REA}=(V\otimes I_A)\ket\psi_{R_0A}$.
Here the purifying system $E$ can be enlarged by an unused factor
so that the isometry extends from the Schmidt support to all of
$R_0$.  The channel
$\mathcal T_{R_0\to R}(Z)=\Tr_E[VZV^\dagger]$ consequently satisfies
\begin{equation}\label{eq:mixed-probe-from-pure}
\rho_{RA}=(\mathcal T_{R_0\to R}\otimes\id_A)(\psi_{R_0A}).
\end{equation}
Because this channel acts only on the reference, it commutes with
either transformation on $A$.  In particular, the two outputs
generated by $\rho_{RA}$ arise from the two outputs generated by
$\psi_{R_0A}$ through the same channel
$\mathcal T_{R_0\to R}\otimes\id_B$.  Applying
\eqref{eq:state-data-processing} gives
\begin{align}
\Dh_\alpha\!\left((\id_R\otimes\cN)(\rho_{RA})
 \middle\|(\id_R\otimes\cM)(\rho_{RA})\right)
\leq
\Dh_\alpha\!\left((\id_{R_0}\otimes\cN)(\psi_{R_0A})
 \middle\|(\id_{R_0}\otimes\cM)(\psi_{R_0A})\right).
\label{eq:pure-probe-dominates-mixed}
\end{align}
This proves one inequality in \eqref{eq:channel-pure-reduction}.
The other follows because every pure probe with $R\simeq A$ is
already allowed in \eqref{eq:channel-divergence}.  Neither argument
assumes that the supremum is attained.

Set $H_\alpha\coloneqq P_\alpha(J_{\cN},J_{\cM})$.
The Choi formula for $\Dh_\alpha(\cN\|\cM)$ is
\begin{equation}\label{eq:channel-choi-formula}
\Dh_\alpha(\cN\|\cM)=
\begin{cases}
\displaystyle\frac{1}{\alpha-1}
\log\lambda_{\min}(\Tr_BH_\alpha),&0<\alpha<1,\\[1.1em]
\displaystyle\frac{1}{\alpha-1}
\log\norm{\Tr_BH_\alpha},&1<\alpha\leq2.
\end{cases}
\end{equation}
The second line assumes $J_{\cN}\ll J_{\cM}$; otherwise the divergence
is infinite. Fang and Fawzi prove the formula above
one~\cite[Theorem 3 and Lemma 5]{FangFawzi2021}, and Katariya and Wilde
prove the extension below one~\cite[Proposition 44]{KatariyaWilde2021}.
Appendix~\ref{app:choi} recalls the proof, including the possibility
that the supremum is approached rather than attained.

To see why an eigenvalue appears, parametrize a pure probe $\ket{\psi_Z}_{AA'}$ of full
Schmidt rank by an invertible matrix $Z$:
\begin{equation}\label{eq:probe-parameter-main}
\ket{\psi_Z}_{AA'}=(Z_A\otimes I_{A'})\ket{\Gamma}_{AA'},
\qquad \Tr(Z^\dagger Z)=1,
\qquad A'\simeq A.
\end{equation}
Here $A$ is the retained reference and $A'$ enters the channel,
as in the Choi convention \eqref{eq:choi}. The matrix $Z$ is
a parametrization of the prepared probe, often called an input
filter; no physical filtering operation is assumed.
Applying a channel to $A'$ conjugates its Choi operator by
$Z\otimes I_B$. Covariance of the perspective $P_\alpha$ under this
invertible conjugation gives
\begin{equation}\label{eq:choi-moment-main}
\Tr P_\alpha\!\left(
(\id\otimes\cN)(\psi_Z),(\id\otimes\cM)(\psi_Z)\right)
=\Tr\!\left[(Z^\dagger Z)\Tr_BH_\alpha\right],
\end{equation}
where $\psi_Z=\ket{\psi_Z}\bra{\psi_Z}$.
Thus the moment is an expectation of the positive operator
$\Tr_BH_\alpha$ in the density operator $Z^\dagger Z$.
Above one the logarithmic prefactor is positive, so maximizing
the divergence selects the largest eigenvalue. Below one the
prefactor is negative, so it selects the smallest eigenvalue.
Appendix~\ref{app:choi} proves the conjugation identity and treats
singular probes and approximating eigenvectors. This sign change
is why both sides of \eqref{eq:intro-normalization} are needed.

\paragraph{The order-one limit.}
The limit as the R\'enyi order decreases to one will provide the
strong-converse thresholds used in Section~\ref{sec:discrimination}.
For a state $\rho$ and a positive semidefinite operator $\sigma$ on
the same finite-dimensional space, the Belavkin--Staszewski divergence
is defined as
\begin{equation}\label{eq:bs-state}
\Dh(\rho\|\sigma)
\coloneqq
\Tr\!\left[\rho\log\!\left(
 \rho^{1/2}\sigma^{-1}\rho^{1/2}\right)\right]
\end{equation}
when $\rho\ll\sigma$, and is $+\infty$ otherwise
\cite{BelavkinStaszewski1982}.  The inverse of $\sigma$ is taken
on its support and extended by zero, and the logarithm is evaluated
on $\supp\rho$.  The operator inside that logarithm is positive
definite on $\supp\rho$ under the stated support condition.
For commuting states this is the classical relative entropy; for
noncommuting states it generally differs from the Umegaki relative
entropy $\Tr[\rho(\log\rho-\log\sigma)]$
\cite[Section 4, Definition 1 and Eq.~(4.1)]{Umegaki1962}.

This divergence is the order-one endpoint of the geometric family:
\begin{equation}\label{eq:bs-state-limit}
\lim_{\alpha\downarrow1}\Dh_\alpha(\rho\|\sigma)
=\Dh(\rho\|\sigma).
\end{equation}
This holds also for singular arguments
\cite[Proposition~79]{KatariyaWilde2021}.
For completeness, if $\rho\ll\sigma$, set
$L=\sigma^{-1/2}\rho^{1/2}$.  The singular-value decomposition
gives $(LL^\dagger)^\alpha=L(L^\dagger L)^{\alpha-1}L^\dagger$
on the nonzero supports.  Therefore, for $\alpha>1$,
\begin{equation}\label{eq:state-moment-order-one}
\Tr P_\alpha(\rho,\sigma)
=\Tr\!\left[\rho
 (\rho^{1/2}\sigma^{-1}\rho^{1/2})^{\alpha-1}\right].
\end{equation}
The right-hand side equals one at $\alpha=1$, and its derivative
there is
$\Tr[\rho\ln(\rho^{1/2}\sigma^{-1}\rho^{1/2})]$.
Taking the logarithm and dividing by $\alpha-1$ proves
\eqref{eq:bs-state-limit}, with the factor $1/\ln2$ converting to
base-two logarithms.  If support inclusion fails, both sides are
$+\infty$, since every divergence with $\alpha>1$ is already
infinite.

Taking this fixed-state limit in
Corollary~\ref{cor:state-dpi} gives
\begin{equation}\label{eq:bs-state-data-processing}
\Dh(\mathcal T(\rho)\|\mathcal T(\sigma))
\leq\Dh(\rho\|\sigma)
\end{equation}
for every quantum channel $\mathcal T$ and every positive
semidefinite comparison $\sigma$.  The same limit in
\eqref{eq:state-normalization-bound} yields
$\Dh(\rho\|\sigma)\geq-\log\Tr\sigma$; hence this divergence
is nonnegative for a subnormalized comparison.

For channels $\cN,\cM\colon A\to B$ between finite-dimensional
systems, define its channel version by
\begin{equation}\label{eq:bs-channel}
\Dh(\cN\|\cM)
\coloneqq
\sup\left\{
\begin{aligned}
&\Dh\!\left((\id_R\otimes\cN)(\rho_{RA})
 \middle\|(\id_R\otimes\cM)(\rho_{RA})\right):\\
&\rho_{RA}\in\cD(RA),\ R\text{ finite-dimensional}
\end{aligned}
\right\}.
\end{equation}
The reference-channel construction in
\eqref{eq:mixed-probe-from-pure}, now using
\eqref{eq:bs-state-data-processing}, proves the equivalent reduction
\begin{equation}\label{eq:bs-channel-pure-reduction}
\Dh(\cN\|\cM)
=\sup\left\{
\begin{aligned}
&\Dh\!\left((\id_R\otimes\cN)(\psi_{RA})
 \middle\|(\id_R\otimes\cM)(\psi_{RA})\right):\\
&\psi_{RA}\in\cD(RA)\text{ pure},\ R\simeq A
\end{aligned}
\right\}.
\end{equation}
The channel endpoint is
\begin{equation}\label{eq:bs-limit}
\lim_{\alpha\downarrow1}\Dh_\alpha(\cN\|\cM)
=\Dh(\cN\|\cM).
\end{equation}
This established fixed-pair limit is stated in
\cite[Eq.~(6.30)]{KatariyaWilde2021}; see also the channel formulas
in \cite[Lemma~5]{FangFawzi2021}.
Appendix~\ref{app:choi} also proves this limit by expanding the
fixed-pair Choi expression. Neither that proof nor the application
in Corollary~\ref{cor:bs} moves a probe or process supremum through
a state limit.

\section{Deterministic processes with channel slots}\label{sec:processes}
We now specify how several independently chosen channels can be
inserted into one experiment. The model isolates the two properties
needed for the proof: complete positivity of the action on Choi
operators and normalization whenever every slot contains a channel.
It specifies the resulting channel or state without assuming that
the individual channel uses can be arranged in one fixed temporal
order. Thus the proof cannot rely on a state evolving through a
prescribed sequence of channel uses.

\Needspace{15\baselineskip}
\subsection{Definition and examples}
Each slot has its own input and output systems. Replacing the channel
in one slot leaves the surrounding process unchanged. The following
definition expresses this requirement on Choi operators;
Figure~\ref{fig:slots} depicts three such insertion slots.
\begin{definition}[Deterministic completely positive process]
\label{def:process}
Let $n\geq1$ be an integer. Fix finite-dimensional slot interfaces
$A_i\to B_i$, $i\in\{1,\ldots,n\}$, and a finite-dimensional
external interface $C\to D$. A deterministic completely positive
process is a completely positive linear transformation
\begin{equation}\label{eq:process-domain}
\Phi\colon
\cL\!\left(\bigotimes_{i=1}^nA_iB_i\right)\to\cL(CD)
\end{equation}
such that
\begin{equation}\label{eq:process-normalization}
\Tr_D\Phi\!\left(\bigotimes_iJ_i\right)=I_C
\quad\text{for every }J_i\in\Ch(A_i\to B_i).
\end{equation}
Here $\Ch(A_i\to B_i)$ is the channel Choi set in
\eqref{eq:channel-choi-set}, with Choi convention \eqref{eq:choi}.
For channels $\cC_i\colon A_i\to B_i$, the output channel
$\cW(\cC_1,\ldots,\cC_n)$ is defined by
\begin{equation}\label{eq:process-output}
J_{\cW(\cC_1,\ldots,\cC_n)}
\coloneqq\Phi\!\left(\bigotimes_iJ_{\cC_i}\right).
\end{equation}
When $C$ is one-dimensional,
we call this a channel-to-state process.
\end{definition}

A simple one-slot example is $\Phi(J)=J/d_A$, obtained by
feeding one share of a normalized maximally entangled state
through the slot and retaining the other share.
It produces a normalized state from every channel Choi operator,
whose trace is $d_A$. Yet $\Phi$ scales the trace on arbitrary
matrices, illustrating why the Choi action need not be trace
preserving on its full domain.

The slots are independently replaceable: the channel in each slot
can be changed without changing the process. The same $\Phi$ is
used under different hypotheses. Initial states, memories, control systems,
and final processing can all be included in $\Phi$, provided that
\eqref{eq:process-normalization} holds. We do not assume that every
process in Definition~\ref{def:process} has a laboratory realization;
a converse for the full mathematical class also holds for a
physically realizable subclass.

Fixed-order circuits, circuits with classical control of order,
and the quantum switch are examples. To make the switch concrete,
suppose two channels $\mathcal A,\mathcal B\colon H\to H$ act on
the same target system and have Kraus operators $\{A_j\}_j$ and
$\{B_k\}_k$, respectively. Its Kraus operators on a
control qubit and the target are
\begin{equation}\label{eq:switch-kraus}
L_{jk}
=\ket{0}\bra{0}\otimes B_kA_j
+\ket{1}\bra{1}\otimes A_jB_k.
\end{equation}
The identity $\sum_{j,k}L_{jk}^\dagger L_{jk}=I$ verifies
normalization for every pair of channels. To check complete
positivity of the Choi action, define the vectorization
$\ket{A_j}\!\rangle=(I\otimes A_j)\ket{\Gamma}$, and similarly
for the other Kraus operators, using the appropriate input space
in \eqref{eq:choi}. Bilinearity of \eqref{eq:switch-kraus}
defines a fixed linear operator $V$ satisfying
\begin{equation}\label{eq:switch-vectorization}
V\bigl(\ket{A_j}\!\rangle\otimes\ket{B_k}\!\rangle\bigr)
=\ket{L_{jk}}\!\rangle.
\end{equation}
A channel's Choi operator is the sum of the outer products of its
vectorized Kraus operators with their adjoints. Summing those outer
products in
\eqref{eq:switch-vectorization} therefore gives the switch Choi
operator $V(J_{\mathcal A}\otimes J_{\mathcal B})V^\dagger$.
The action $Z\mapsto VZV^\dagger$ is completely positive and
depends only on the input Choi operators, proving independence
of the chosen Kraus representations as well. This is the usual
quantum switch supermap~\cite{ChiribellaEtAl2013}.

Figure~\ref{fig:slots} distinguishes the insertion slots from their
internal connections. Its three holes can hold three independently
chosen channels. Local inputs enter from the left and local outputs
leave to the right, but no connection between different slots is
drawn: Definition~\ref{def:process} does not select their order.
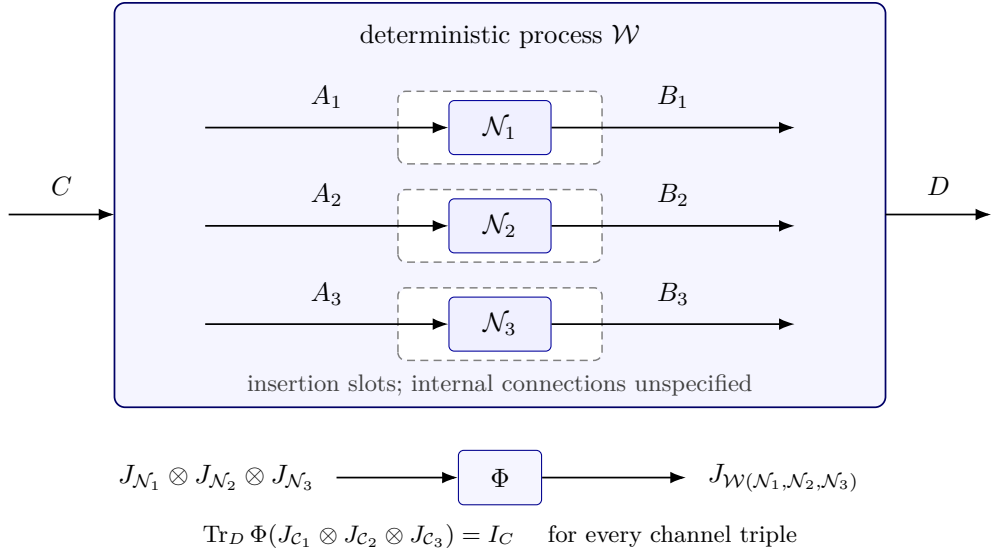
\begin{figure}[!t]
\centering
\begin{tikzpicture}[font=\small,>=Latex,
 channel/.style={draw=blue!60!black,fill=blue!6,rounded corners=2pt,
   minimum width=1.35cm,minimum height=.60cm,
   inner xsep=9pt,inner ysep=6pt},
 wire/.style={->,line width=.65pt},
 wirelabel/.style={inner sep=3pt}]
\path[fill=blue!4,even odd rule,rounded corners=4pt]
 (-5.1,-2.55) rectangle (5.1,2.8)
 (-1.35,.67) rectangle (1.35,1.63)
 (-1.35,-.63) rectangle (1.35,.33)
 (-1.35,-1.93) rectangle (1.35,-.97);
\draw[blue!40!black,line width=.65pt,rounded corners=4pt]
 (-5.1,-2.55) rectangle (5.1,2.8);
\foreach \y in {1.15,-.15,-1.45}{
 \draw[black!50,densely dashed,line width=.55pt,rounded corners=3pt]
   (-1.35,\y-.48) rectangle (1.35,\y+.48);
}
\node[inner sep=0pt] at (0,2.35)
 {deterministic process $\mathcal W$};
\node[channel] (n1) at (0,1.15) {$\mathcal N_1$};
\node[channel] (n2) at (0,-.15) {$\mathcal N_2$};
\node[channel] (n3) at (0,-1.45) {$\mathcal N_3$};
\draw[wire] (-3.9,1.15)--node[above=4pt,wirelabel]{$A_1$}(n1.west);
\draw[wire] (n1.east)--node[above=4pt,wirelabel]{$B_1$}(3.9,1.15);
\draw[wire] (-3.9,-.15)--node[above=4pt,wirelabel]{$A_2$}(n2.west);
\draw[wire] (n2.east)--node[above=4pt,wirelabel]{$B_2$}(3.9,-.15);
\draw[wire] (-3.9,-1.45)--node[above=4pt,wirelabel]{$A_3$}(n3.west);
\draw[wire] (n3.east)--node[above=4pt,wirelabel]{$B_3$}(3.9,-1.45);
\draw[wire] (-6.5,0)--node[above=4pt,wirelabel]{$C$}(-5.1,0);
\draw[wire] (5.1,0)--node[above=4pt,wirelabel]{$D$}(6.5,0);
\node[font=\footnotesize,text=black!75,inner sep=0pt] at (0,-2.25)
 {insertion slots; internal connections unspecified};
\node[inner sep=5pt] (input) at (-3.75,-3.48)
 {$J_{\mathcal N_1}\otimes J_{\mathcal N_2}\otimes J_{\mathcal N_3}$};
\node[draw=blue!55!black,fill=blue!4,rounded corners=2pt,
 minimum width=1.1cm,minimum height=.72cm,
 inner xsep=9pt,inner ysep=6pt] (phi) at (0,-3.48) {$\Phi$};
\node[inner sep=5pt] (output) at (3.75,-3.48)
 {$J_{\mathcal W(\mathcal N_1,\mathcal N_2,\mathcal N_3)}$};
\draw[wire,shorten <=3pt] (input.east)--(phi.west);
\draw[wire,shorten >=3pt] (phi.east)--(output.west);
\node[font=\footnotesize,inner sep=0pt] at (0,-4.25)
 {$\operatorname{Tr}_D\Phi(J_{\mathcal C_1}\otimes
 J_{\mathcal C_2}\otimes J_{\mathcal C_3})=I_C$
 \quad for every channel triple};
\end{tikzpicture}

\caption{Channel slots and their Choi representation. The shaded region
is a deterministic process with external input $C$ and output $D$.
The three unshaded holes are insertion slots, each containing a
replaceable channel and its own labeled input and output wires.
Each local input-to-output direction runs from left to right.
The internal connections are unspecified; the drawing does not impose
a relative temporal order on the three slots.
The lower line is an algebraic representation: the tensor product
records the independently replaceable slots, and the completely
positive transformation $\Phi$ contains their wiring. Normalization
must hold for every independent triple of inserted channels.}
\label{fig:slots}
\end{figure}

\subsection{Positive testers}
To analyze outcome probabilities, we first specialize the Choi action
to scalar-valued functionals. These are the objects bounded in
Theorem~\ref{thm:normalization}; the same bound will later control
channel outputs through Corollary~\ref{cor:effect-bounds}.
A \emph{positive functional} is a complex-linear scalar
functional $f$ satisfying $f(Z)\geq0$ for every positive semidefinite
$Z$ on the entire joint slot space. In finite dimensions it has
the form $f_T(Z)=\Tr(TZ)$ for a positive operator $T$.
Normalization on products of channels is a separate condition;
we do not require $f(I)=1$.
The operator $T$ is a deterministic tester normalizer if
\begin{equation}\label{eq:tester-normalization}
\Tr\!\left[T\bigotimes_iJ_i\right]=1
\quad\text{for every }J_i\in\Ch(A_i\to B_i).
\end{equation}
A measurement consists of positive semidefinite operators $T_x$ satisfying
$\sum_xT_x=T$. The outcome probability on a channel tuple is
$\Tr[T_x\bigotimes_iJ_i]$.
This is the positive process-tester model. The two-slot formulation
appears in the general testers of Bavaresco, Murao, and
Quintino~\cite[Definition 3 and Appendix A]{BavarescoMuraoQuintino2021};
the normalization on arbitrary products of local channels is stated
for an arbitrary number of slots by Zhu and
Wang~\cite[Definition 2.3]{ZhuWang2026}.
These references concern the tester class used here, without an
assertion that every member has a realization by a circuit with
definite causal order.
In a convention that contracts fully transposed Choi operators with a
process matrix $W$, the transpose is taken in fixed product bases of
all slot input and output systems. For a joint slot operator $Z$, the
identity $\Tr[WZ^{\transpose}]=\Tr[W^{\transpose}Z]$ allows the
transpose to be absorbed into the tester $T=W^{\transpose}$;
full transpose preserves positivity.
Equation~\eqref{eq:tester-normalization} fixes our contraction
convention.

For one slot and an input state $\eta$, the familiar example is
\begin{equation}\label{eq:one-slot-functional}
f_\eta(X)=\Tr[(\eta^{\transpose}\otimes I_B)X]
=\Tr[\eta(\Tr_BX)^{\transpose}].
\end{equation}
It equals the output weight of the operation with Choi operator
$X$ on input $\eta$ and equals one on every channel.
Its weight lies between the minimum and maximum eigenvalues of
$\Tr_BX$. Theorem~\ref{thm:normalization} extends these bounds
to a product of local positive operators. Its proof uses only
normalization on independent channel replacements, so it does not
require a state specified before each channel use in a fixed-order
circuit.

For a process with an external input, fixing a unit vector
$\ket{v}_C$ gives a positive functional
\begin{equation}\label{eq:fv}
f_v(Z)\coloneqq
\bra{v}\Tr_D\Phi(Z)\ket{v}.
\end{equation}
Equation~\eqref{eq:process-normalization} makes this functional
normalized on every product of channel Choi operators.
Bounding $f_v$ for every $\ket{v}_C$ is therefore enough to
bound the input effect of an output operation.

\subsection{What is and is not required}\label{subsec:process-assumptions}
For a single slot, a deterministic process is an ordinary superchannel,
which admits a realization by preprocessing and postprocessing around
that slot~\cite{ChiribellaDArianoPerinotti2008}. For several slots,
the normalization condition applies to independently inserted channels.
Definition~\ref{def:process} requires complete positivity of the Choi
action, not trace preservation on the full Choi matrix algebra.
It requires normalization on every product of channel Choi operators,
not on every joint channel Choi operator from $\bigotimes_i A_i$
to $\bigotimes_iB_i$. This distinction allows sequential and
indefinite-order connections of separate slots.
Supermaps on product channels are discussed in
\cite[Sections III E--III F, especially Definition 6 and
Theorem 2]{ChiribellaEtAl2013}.
For a concrete reason, consider connecting the first slot's output
to the second slot's input. An arbitrary joint channel on the
aggregate interface could itself send information from the second
input back to the first output. Closing the intended sequential
connection would then form a feedback loop. Normalization on
independently inserted channels avoids imposing a condition on
such arbitrary aggregate replacements.

Normalization only on the two tested hypotheses, or only on tuples
with the same channel in all slots, is a different assumption.
The proof uses independent replacement by auxiliary channels in
each slot. Section~\ref{subsec:normalization-counterexample}
gives a counterexample when normalization is imposed only on the
two tested hypotheses.

\section{Normalization of positive local operations}\label{sec:normalization}
Section~\ref{sec:processes} requires an experiment to have total
weight one whenever all inserted operations are channels. The
perspective operators that enter the divergence proof, however,
need not have the normalization of channel Choi operators. We
therefore need to determine how much weight a process can assign
to a product of arbitrary positive local operators.

For one slot, \eqref{eq:one-slot-functional} bounds that weight
by the extremal eigenvalues of the input marginal. Our goal in
Theorem~\ref{thm:normalization} is to obtain the products of these
local bounds for an arbitrary number of slots. A sequential
argument would follow the input state before each use, but such
a sequence is not part of Definition~\ref{def:process}.
We instead use the product-channel normalization identity
directly. The proof has two ingredients: a linear extension of
that identity and a pair of positive operators whose process
weights are related by Cauchy--Schwarz.

\subsection{Extending the normalization identity}
The normalization hypothesis initially applies only to positive
operators having output partial trace equal to the identity.
The lower-bound proof will also produce a cross operator with
that same partial trace but without positivity or Hermiticity;
see \eqref{eq:cross-marginal}. To evaluate its process weight,
we first show that linearity extends the normalization identity
to every operator with the prescribed partial trace. This step
is purely algebraic and does not identify the cross operator
with a physical operation.
A complex affine combination is a finite sum $\sum_j z_jJ_j$
with complex coefficients satisfying $\sum_jz_j=1$.
Unlike a probabilistic mixture, it need not preserve positivity
or Hermiticity. It does preserve the linear constraint
$\Tr_BJ_j=I_A$, which is the property used here.

\begin{lemma}[Complex affine extension]\label{lem:affine}
Let $A,B$ be finite-dimensional systems. The complex affine hull
of the channel Choi set $\Ch(A\to B)$ in
\eqref{eq:channel-choi-set}, consisting of finite complex
combinations of its elements whose coefficients sum to one, is
\begin{equation}\label{eq:complex-affine}
\{Z\in\cL(AB):\Tr_BZ=I_A\}.
\end{equation}
Consequently, for an integer $n\geq1$ and finite-dimensional
systems $A_i,B_i$, if a complex-linear functional
$f\colon\cL(\bigotimes_{i=1}^n A_iB_i)\to\mathbb C$ satisfies
$f(\bigotimes_iJ_i)=1$ for every $J_i\in\Ch(A_i\to B_i)$, then
\begin{equation}\label{eq:affine-normalization}
f\!\left(\bigotimes_iZ_i\right)=1
\quad\text{whenever }\Tr_{B_i}Z_i=I_{A_i}\text{ for all }i.
\end{equation}
Here each $Z_i\in\cL(A_iB_i)$ need not be Hermitian or positive.
\end{lemma}
\begin{proof}
Let $J_0=I_A\otimes I_B/d_B$, a positive definite channel Choi
operator. A Hermitian $H$ satisfying $\Tr_BH=0$ is a real multiple
of a difference of channel Choi operators: choose $t>0$ small
enough that $J_0+tH\geq0$, and write
$H=t^{-1}[(J_0+tH)-J_0]$.

If $\Tr_BZ=I_A$, write
$Z=J_0+H+iK$, where $H,K$ are Hermitian and have zero output
partial trace. The preceding argument expresses $H$ and $K$
as linear combinations of channel Choi operators with coefficient
sum zero. Hence $Z$ is a complex affine combination of channel
Choi operators, with coefficient sum one. Conversely, every such
affine combination has output partial trace $I_A$.

For each $i$, write
$Z_i=\sum_{j=1}^{q_i}z_{ij}J_{ij}$, where the coefficients satisfy
$\sum_{j=1}^{q_i}z_{ij}=1$ and each
$J_{ij}\in\Ch(A_i\to B_i)$. Multilinearity of the tensor
product and complex linearity of $f$ give
\begin{align}
f\!\left(\bigotimes_{i=1}^nZ_i\right)
&=f\!\left(
\sum_{j_1=1}^{q_1}\cdots\sum_{j_n=1}^{q_n}
\left(\prod_{i=1}^n z_{ij_i}\right)
\bigotimes_{i=1}^nJ_{ij_i}\right)
\label{eq:affine-expansion:line-1}\\
&=\sum_{j_1=1}^{q_1}\cdots\sum_{j_n=1}^{q_n}
\left(\prod_{i=1}^n z_{ij_i}\right)
f\!\left(\bigotimes_{i=1}^nJ_{ij_i}\right)
\label{eq:affine-expansion:line-2}\\
&=\sum_{j_1=1}^{q_1}\cdots\sum_{j_n=1}^{q_n}
\prod_{i=1}^n z_{ij_i}
=\prod_{i=1}^n\left(\sum_{j=1}^{q_i}z_{ij}\right)
=1.\label{eq:affine-expansion}
\end{align}
The third expression uses normalization on every independent tuple
$(J_{1j_1},\ldots,J_{nj_n})$ of channel Choi operators. The next
equality factors the finite sums, and the last uses the coefficient
sum in each slot. Positivity of $f$ is not needed for this affine
extension. The real affine statement is also
used in the supermap characterization of
\cite[Lemma 2]{ChiribellaEtAl2013}.
\end{proof}

\subsection{The product bound and its proof}
Lemma~\ref{lem:affine} supplies the identity needed for the cross
operator. We now combine that identity with positivity to bound
the weight of a product of local operations. The upper bound
comes from completing each operation to a channel. The lower
bound requires a reciprocal construction, whose role is
explained in the proof strategy following the statement.

\begin{theorem}[Normalization of positive local operations]
\label{thm:normalization}
Let $n\geq1$ be an integer and let $A_i,B_i$ be finite-dimensional
systems. Let $f\colon\cL(\bigotimes_{i=1}^n A_iB_i)\to\mathbb C$
be complex-linear and positive, meaning $f(Z)\geq0$ whenever
$Z\geq0$. Suppose $f(\bigotimes_iJ_i)=1$ for every independent
choice $J_i\in\Ch(A_i\to B_i)$, with the channel Choi sets
defined in \eqref{eq:channel-choi-set}.
For positive semidefinite $X_i\in\cL(A_iB_i)$, set
$K_i\coloneqq\Tr_{B_i}X_i$,
$m_i\coloneqq\lambda_{\min}(K_i)$, and $r_i\coloneqq\norm{K_i}$.
Then
\begin{equation}\label{eq:normalization-bound}
\prod_i m_i
\leq f\!\left(\bigotimes_iX_i\right)
\leq\prod_i r_i.
\end{equation}
For every fixed tuple $(X_i)_i$, both bounds are attained by
positive functionals satisfying the stated product-channel
normalization.
\end{theorem}

\paragraph{Proof strategy.}
The upper bound uses a physical completion: after division by its
largest input-effect eigenvalue, a local operation can be completed
to a channel. A corresponding lower completion need not exist,
as illustrated in Section~\ref{subsec:reciprocal-obstruction}.
Instead, we write each operation as a channel Choi operator
conjugated by the square root of its input marginal. Replacing this
square root by its inverse gives a second positive operation.
Lemma~\ref{lem:affine} fixes the cross term between these two
operations to one. Cauchy--Schwarz then turns an upper bound
for the second operation into a lower bound for the first.
Figure~\ref{fig:reciprocal} displays this construction.

\begin{proof}
We first prove the upper bound in \eqref{eq:normalization-bound}. If $r_i=0$ for some $i$, positivity
and $\Tr X_i=\Tr K_i=0$ imply $X_i=0$, so both the product and its
process weight vanish. Otherwise, choose a state $\tau_i$ on $B_i$
and set
\begin{equation}\label{eq:upper-completion}
J_i^{\mathrm{comp}}\coloneqq
\frac{X_i+(r_iI_{A_i}-K_i)\otimes\tau_i}{r_i}.
\end{equation}
The numerator is positive semidefinite and has output partial trace
$r_iI_{A_i}$. Thus $J_i^{\mathrm{comp}}\in\Ch(A_i\to B_i)$ and
$r_iJ_i^{\mathrm{comp}}-X_i=(r_iI_{A_i}-K_i)\otimes\tau_i\geq0$.
These local inequalities imply a joint operator inequality.
To see it explicitly, use the telescoping identity
\begin{align}
\bigotimes_{i=1}^n(r_iJ_i^{\mathrm{comp}})-\bigotimes_{i=1}^n X_i
=\sum_{k=1}^n
\left(\bigotimes_{i=1}^{k-1}r_iJ_i^{\mathrm{comp}}\right)
\otimes(r_kJ_k^{\mathrm{comp}}-X_k)
\otimes\left(\bigotimes_{i=k+1}^n X_i\right)
\geq0.\label{eq:completion-tensor-order}
\end{align}
Empty tensor products are scalars equal to one. Each summand is
positive semidefinite because all of its tensor factors are
positive semidefinite. Positivity and linearity of $f$ then give
\begin{align}
0&\leq f\!\left(
\bigotimes_{i=1}^n(r_iJ_i^{\mathrm{comp}})-\bigotimes_{i=1}^nX_i\right)
\label{eq:upper-weight:line-1}\\
&=\left(\prod_{i=1}^n r_i\right)
f\!\left(\bigotimes_{i=1}^nJ_i^{\mathrm{comp}}\right)
-f\!\left(\bigotimes_{i=1}^nX_i\right)
\label{eq:upper-weight:line-2}\\
&=\prod_{i=1}^n r_i-f\!\left(\bigotimes_{i=1}^nX_i\right).
\label{eq:upper-weight}
\end{align}
The last equality uses normalization on the auxiliary channel
tuple $(J_i^{\mathrm{comp}})_{i=1}^n$. Rearranging proves the upper bound.

For the lower bound in \eqref{eq:normalization-bound}, if some $m_i=0$, positivity of $f$ proves
the claim. We therefore assume $K_i>0$ for every $i$ and define
\begin{equation}\label{eq:reciprocal-normalization}
F_i\coloneqq K_i^{1/2}\otimes I_{B_i},
\qquad
J_i\coloneqq F_i^{-1}X_iF_i^{-1}.
\end{equation}
Each $J_i$ is positive semidefinite and
\begin{equation}\label{eq:normalized-marginal}
\Tr_{B_i}J_i=K_i^{-1/2}K_iK_i^{-1/2}=I_{A_i}.
\end{equation}
Hence $J_i$ is a channel Choi operator, even if $X_i$ is singular.
Write $F=\bigotimes_iF_i$ and $J=\bigotimes_iJ_i$, so
$FJF=\bigotimes_iX_i$.

The cross operator has the prescribed partial trace in each slot:
\begin{equation}\label{eq:cross-marginal}
\Tr_{B_i}(F_iJ_iF_i^{-1})
=K_i^{1/2}I_{A_i}K_i^{-1/2}=I_{A_i}.
\end{equation}
Lemma~\ref{lem:affine} therefore gives
\begin{equation}\label{eq:cross-one}
f(FJF^{-1})=1.
\end{equation}
No commutation of $F_i$ and $J_i$ is used in this identity.
Complex linearity evaluates this cross term. Positivity will
instead be applied to operators of the form $A^\dagger A$
and $B^\dagger B$ in the next step.

Cauchy--Schwarz for a positive functional states that
$|f(A^\dagger B)|^2\leq f(A^\dagger A)f(B^\dagger B)$.
It follows, for example, by writing $f(Z)=\Tr(TZ)$ with $T\geq0$
and applying Hilbert--Schmidt Cauchy--Schwarz to
$AT^{1/2}$ and $BT^{1/2}$. Applying it with
$A=J^{1/2}F$ and $B=J^{1/2}F^{-1}$ yields
\begin{equation}\label{eq:reciprocal-cs}
1=|f(FJF^{-1})|^2
\leq f(FJF)\,f(F^{-1}JF^{-1}).
\end{equation}
To bound the second factor, set
$Y_i=F_i^{-1}J_iF_i^{-1}$. Its output partial trace is
\begin{equation}\label{eq:inverse-effect}
\Tr_{B_i}Y_i=K_i^{-1}.
\end{equation}
The upper bound \eqref{eq:upper-weight:line-1}--\eqref{eq:upper-weight}, applied to these positive
$Y_i$, implies
\begin{equation}\label{eq:inverse-upper}
f(F^{-1}JF^{-1})
\leq\prod_i\norm{K_i^{-1}}
=\prod_i m_i^{-1}.
\end{equation}
Combining \eqref{eq:reciprocal-cs} and \eqref{eq:inverse-upper}
proves the lower bound in \eqref{eq:normalization-bound}.

For tightness, select a unit eigenvector $\ket{v_i}$ of $K_i$
corresponding to its smallest eigenvalue, and define
\begin{equation}\label{eq:sharp-tester}
f_-(Z)\coloneqq
\Tr\!\left[
\bigotimes_i(\ket{v_i}\bra{v_i}\otimes I_{B_i})\,Z
\right].
\end{equation}
It is positive and normalized on product channels, and its value
on $\bigotimes_iX_i$ is $\prod_i m_i$. Largest-eigenvalue
eigenvectors give the upper endpoint in the same way.
\end{proof}

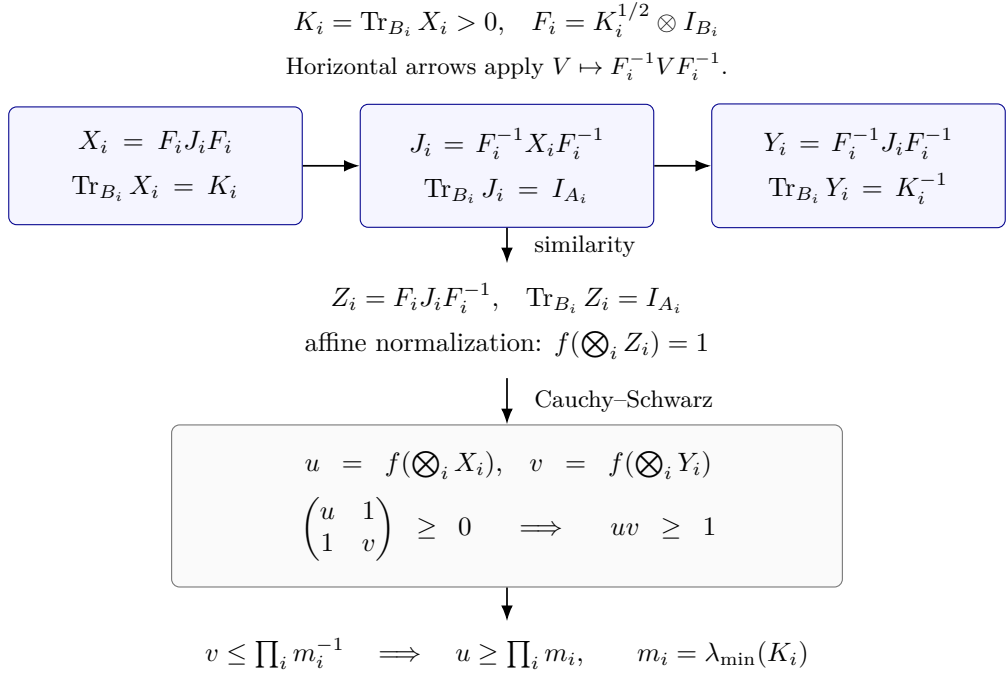
\begin{figure}[!t]
\centering
\begin{tikzpicture}[font=\small,>=Latex,
 box/.style={draw=blue!55!black,fill=blue!4,rounded corners=3pt,
   align=center,text width=3.25cm,minimum height=1.45cm,
   inner xsep=9pt,inner ysep=9pt},
 arrow/.style={->,line width=.65pt}]
\node[align=center,inner sep=0pt] at (0,1.65)
 {$K_i=\Tr_{B_i}X_i>0$,\quad $F_i=K_i^{1/2}\otimes I_{B_i}$\\[.45em]
 \footnotesize Horizontal arrows apply $V\mapsto F_i^{-1}VF_i^{-1}$.};
\node[box] (x) at (-4.65,0)
 {$X_i=F_iJ_iF_i$\\[.5em]$\Tr_{B_i}X_i=K_i$};
\node[box] (j) at (0,0)
 {$J_i=F_i^{-1}X_iF_i^{-1}$\\[.5em]$\Tr_{B_i}J_i=I_{A_i}$};
\node[box] (y) at (4.65,0)
 {$Y_i=F_i^{-1}J_iF_i^{-1}$\\[.5em]$\Tr_{B_i}Y_i=K_i^{-1}$};
\draw[arrow] (x.east)--(j.west);
\draw[arrow] (j.east)--(y.west);
\node[align=center,inner xsep=8pt,inner ysep=7pt] (cross) at (0,-2.05)
 {$Z_i=F_iJ_iF_i^{-1}$,\quad $\Tr_{B_i}Z_i=I_{A_i}$\\[.5em]
 affine normalization: $f(\bigotimes_i Z_i)=1$};
\draw[arrow] (j.south)--node[right=7pt,font=\footnotesize,inner sep=3pt]
 {similarity}(cross.north);
\node[draw=black!55,rounded corners=3pt,fill=black!2,align=center,
 text width=8.1cm,minimum height=1.55cm,
 inner xsep=11pt,inner ysep=10pt] (gram) at (0,-4.50)
 {$u=f(\bigotimes_iX_i)$,\quad $v=f(\bigotimes_iY_i)$\\[.55em]
 $\begin{pmatrix}u&1\\1&v\end{pmatrix}\geq0
 \quad\Longrightarrow\quad uv\geq1$};
\draw[arrow] (cross.south)--node[right=7pt,font=\footnotesize,inner sep=3pt]
 {Cauchy--Schwarz}(gram.north);
\node[align=center,inner xsep=8pt,inner ysep=6pt] (bound) at (0,-6.45)
 {$v\leq\prod_i m_i^{-1}
 \quad\Longrightarrow\quad
 u\geq\prod_i m_i$,\qquad $m_i=\lambda_{\min}(K_i)$};
\draw[arrow] (gram.south)--(bound.north);
\end{tikzpicture}

\caption{The lower-bound mechanism in
Theorem~\ref{thm:normalization}. In each slot, the positive Choi
operator $X_i$ is normalized by its input marginal $K_i$ to a
channel Choi operator $J_i$. The reciprocal Choi operator $Y_i$ has marginal
$K_i^{-1}$. The cross operator $Z_i=F_iJ_iF_i^{-1}$ need not be
positive, but its marginal is $I$, so affine normalization fixes
$f(\bigotimes_iZ_i)=1$. Cauchy--Schwarz relates the two positive
weights through the displayed Gram matrix for
$A=J^{1/2}F$ and $B=J^{1/2}F^{-1}$ in
\eqref{eq:reciprocal-cs}. Its off-diagonal entries are one by
affine normalization. The upper bound on the reciprocal weight supplies
the desired lower bound. The arrows denote algebraic constructions,
not physical implementation of inverse filters.}
\label{fig:reciprocal}
\end{figure}

\FloatBarrier
\subsection{The operator consequence}
Theorem~\ref{thm:normalization} controls scalar process weights.
To apply it to channel outputs, we evaluate the output Choi marginal
in an arbitrary unit vector and use the functional in
\eqref{eq:fv}. Bounds valid in every vector direction give an
operator inequality. The input effect is the transpose of this
marginal by \eqref{eq:effect-transpose}, so the same bounds by
scalar multiples of the identity hold for the input effect.
This connects the normalization argument to the channel Choi
formula \eqref{eq:channel-choi-formula}.
\begin{corollary}[Output input-effect bounds]\label{cor:effect-bounds}
Let $n\geq1$ be an integer, let $A_i,B_i,C,D$ be finite-dimensional
systems, and let $\Phi\colon\cL(\bigotimes_{i=1}^n A_iB_i)
\to\cL(CD)$ be completely positive and normalized as in
\eqref{eq:process-normalization} for every independent tuple of
channel Choi operators. For positive semidefinite
$X_i\in\cL(A_iB_i)$, set $K_i\coloneqq\Tr_{B_i}X_i$. Then
\begin{equation}\label{eq:operator-bound}
\left(\prod_i\lambda_{\min}(K_i)\right)I_C
\leq
\Tr_D\Phi\!\left(\bigotimes_iX_i\right)
\leq
\left(\prod_i\norm{K_i}\right)I_C.
\end{equation}
\end{corollary}
\begin{proof}
For each unit vector $\ket{v}_C$, the functional in
\eqref{eq:fv} satisfies the hypotheses of
Theorem~\ref{thm:normalization}. Apply that theorem to every
$\ket{v}_C$. Bounds on all quadratic forms are equivalent to
the operator inequalities in \eqref{eq:operator-bound}.
\end{proof}
In particular, if each inserted positive operator is the Choi
operator of a trace-nonincreasing operation, its output under
$\Phi$ is also a trace-nonincreasing operation. This consequence
will allow subnormalized comparisons in a communication converse.

\section{The geometric R\'enyi process inequality}\label{sec:main}
The proof of one of our main results (Theorem~\ref{thm:main}) now combines two order inequalities. Operator Jensen
compares the perspective of the output channels with the process
applied to the local perspectives. Corollary~\ref{cor:effect-bounds}
then controls the partial trace of that positive operator.

\begin{theorem}[Geometric R\'enyi divergence under deterministic processes]
\label{thm:main}
Let $n\geq1$ be an integer and let
$\cN_i,\cM_i\colon A_i\to B_i$ be channels between
finite-dimensional systems for $i\in\{1,\ldots,n\}$.
Let $\cW$ be a deterministic completely
positive process with Choi action $\Phi$ as in
\eqref{eq:process-domain}--\eqref{eq:process-normalization}, with
finite-dimensional external interface $C\to D$.
Thus $\cW(\cC_1,\ldots,\cC_n)$ has Choi operator
$\Phi(\bigotimes_iJ_{\cC_i})$ for every channel tuple, using
the convention \eqref{eq:choi}.
For every $\alpha\in(0,1)\cup(1,2]$, the channel divergence
in \eqref{eq:channel-divergence} satisfies
\begin{equation}\label{eq:main}
\Dh_\alpha\!\left(
\cW(\cN_1,\ldots,\cN_n)
\middle\|
\cW(\cM_1,\ldots,\cM_n)\right)
\leq\sum_{i=1}^n\Dh_\alpha(\cN_i\|\cM_i).
\end{equation}
The inequality holds in the extended-real sense and imposes no
causal order on the slots.
\end{theorem}

\paragraph{Proof strategy.}
Let $H_i$ be the geometric perspective of the two Choi operators
in slot $i$. Above order one, the channel divergence depends on
the largest eigenvalue of $\Tr_{B_i}H_i$, so we use the upper
normalization bound. Below order one, it depends on the smallest
eigenvalue; Jensen reverses, and we use the lower normalization
bound. The negative logarithmic prefactor below one then restores
the same direction of the divergence inequality.

\begin{proof}
If an input divergence is infinite, \eqref{eq:main} is immediate.
Assume henceforth that every input divergence is finite.
For $\alpha>1$ this implies $J_{\cN_i}\ll J_{\cM_i}$ for
each $i$. Tensor products and the positive transformation $\Phi$
preserve this support inclusion, so the output perspective is
also defined on the comparison support.
Set
\begin{equation}\label{eq:local-H}
H_i\coloneqq P_\alpha(J_{\cN_i},J_{\cM_i}),\qquad
H_{\mathrm{out}}\coloneqq
P_\alpha\!\left(
\Phi\!\left(\bigotimes_iJ_{\cN_i}\right),
\Phi\!\left(\bigotimes_iJ_{\cM_i}\right)\right).
\end{equation}
First suppose $1<\alpha\leq2$. Lemma~\ref{lem:perspective} gives
\begin{equation}\label{eq:main-upper-perspective}
H_{\mathrm{out}}
\leq
\Phi\!\left(\bigotimes_iH_i\right).
\end{equation}
Let $r_i=\norm{\Tr_{B_i}H_i}$. Taking the output partial trace
and applying the upper bound in \eqref{eq:operator-bound} yields
\begin{equation}\label{eq:main-upper-effect}
\Tr_DH_{\mathrm{out}}
\leq
\Tr_D\Phi\!\left(\bigotimes_iH_i\right)
\leq
\left(\prod_i r_i\right)I_C.
\end{equation}
The second line of \eqref{eq:channel-choi-formula}, followed by
the logarithm and division by $\alpha-1>0$, gives \eqref{eq:main}.

Now suppose $0<\alpha<1$, using the singular extension
\eqref{eq:mean-extension}. Finiteness of the input divergences
means $m_i=\lambda_{\min}(\Tr_{B_i}H_i)>0$ for every $i$.
Lemma~\ref{lem:perspective} and the lower bound in
\eqref{eq:operator-bound} give
\begin{equation}\label{eq:main-lower-effect}
\Tr_DH_{\mathrm{out}}
\geq
\Tr_D\Phi\!\left(\bigotimes_iH_i\right)
\geq
\left(\prod_i m_i\right)I_C.
\end{equation}
Therefore $\lambda_{\min}(\Tr_DH_{\mathrm{out}})\geq\prod_i m_i$.
Taking the logarithm and dividing by $\alpha-1<0$ reverses the
inequality. The first line of \eqref{eq:channel-choi-formula}
then proves \eqref{eq:main}.
\end{proof}

\subsection{Exact optimization over processes}
Theorem~\ref{thm:main} limits the divergence produced by every
allowed process. To determine whether that bound is sharp, we
can use separate reference-assisted probes and apply the channels
in parallel. The product structure of the resulting output states
makes their geometric divergences add, so this simple choice
approaches the upper bound after optimization.

\begin{corollary}[Exact process optimization]\label{cor:collapse}
Let $n\geq1$ be an integer, let
$\cN_i,\cM_i\colon A_i\to B_i$ be channel pairs on
finite-dimensional systems, and let $\alpha\in(0,1)\cup(1,2]$.
Let $\mathfrak P$ be the class
of channel-to-state processes satisfying Definition~\ref{def:process},
with arbitrary finite-dimensional output systems: each Choi action
$\Theta$ is completely positive and has trace one on every
independent product of channel Choi operators.
With Choi convention \eqref{eq:choi}, state divergence
\eqref{eq:state-divergence}, and channel divergence
\eqref{eq:channel-divergence},
\begin{equation}\label{eq:collapse}
\sup\Biggl\{
\Dh_\alpha\!\left(
\Theta\!\left(\bigotimes_iJ_{\cN_i}\right)
\middle\|
\Theta\!\left(\bigotimes_iJ_{\cM_i}\right)\right):\Theta\in\mathfrak P
\Biggr\}
=\sum_{i=1}^n\Dh_\alpha(\cN_i\|\cM_i),
\end{equation}
where the equality is understood in the extended-real sense.
No attainment of the supremum is asserted.
\end{corollary}
\begin{proof}
Theorem~\ref{thm:main}, applied to a process with one-dimensional
external input, supplies the upper bound. To prove the reverse
inequality, we construct a parallel process and check both its
normalization and its output divergence.

First suppose that the sum on the right of \eqref{eq:collapse}
is finite, and fix $\eta>0$. By the supremum definition and
the pure-probe reduction \eqref{eq:channel-pure-reduction},
choose a normalized pure probe $\psi_i$ in each slot whose
two output states
\begin{equation}\label{eq:parallel-probe-outputs}
\rho_i\coloneqq(\id_{R_i}\otimes\cN_i)(\psi_i),
\qquad
\sigma_i\coloneqq(\id_{R_i}\otimes\cM_i)(\psi_i)
\end{equation}
satisfy
\begin{equation}\label{eq:parallel-near-optimal}
\Dh_\alpha(\rho_i\|\sigma_i)
\geq\Dh_\alpha(\cN_i\|\cM_i)-\frac{\eta}{n},
\qquad i\in\{1,\ldots,n\}.
\end{equation}
Here $R_i\simeq A_i$. In the Choi convention \eqref{eq:choi},
write the probe as
$\ket{\psi_i}_{R_iA_i'}=(Z_i\otimes I_{A_i'})
\ket{\Gamma}_{A_iA_i'}$, where $A_i'$ is a copy of the channel
input and $Z_i\colon A_i\to R_i$ satisfies
$\Tr(Z_i^\dagger Z_i)=1$. No invertibility is required.
The corresponding parallel Choi action is
\begin{equation}\label{eq:parallel-choi-action}
\Theta_\psi(Z)\coloneqq VZV^\dagger,
\qquad
V\coloneqq\bigotimes_{i=1}^n(Z_i\otimes I_{B_i}).
\end{equation}
This action is completely positive. For every independent tuple
of channels $\cC_i\colon A_i\to B_i$, its trace is
\begin{align}
\Tr\Theta_\psi\!\left(\bigotimes_{i=1}^nJ_{\cC_i}\right)
&=\prod_{i=1}^n
\Tr\!\left[(Z_i^\dagger Z_i\otimes I_{B_i})J_{\cC_i}\right]
\label{eq:parallel-normalization:line-1}\\
&=\prod_{i=1}^n\Tr(Z_i^\dagger Z_i)=1.
\label{eq:parallel-normalization}
\end{align}
The second equality uses $\Tr_{B_i}J_{\cC_i}=I_{A_i}$.
Thus $\Theta_\psi$ belongs to the process class in the statement.
It produces $\bigotimes_{i=1}^n\rho_i$ and
$\bigotimes_{i=1}^n\sigma_i$ under the two hypotheses.

The tensor identity \eqref{eq:perspective-tensor} makes the
geometric moment multiplicative. Consequently,
\begin{align}
\Dh_\alpha\!\left(\bigotimes_{i=1}^n\rho_i
\middle\|\bigotimes_{i=1}^n\sigma_i\right)
&=\frac{1}{\alpha-1}
\log\prod_{i=1}^n\Tr P_\alpha(\rho_i,\sigma_i)
\label{eq:parallel-additivity:line-1}\\
&=\sum_{i=1}^n\Dh_\alpha(\rho_i\|\sigma_i)
\label{eq:parallel-additivity:line-2}\\
&\geq\sum_{i=1}^n\Dh_\alpha(\cN_i\|\cM_i)-\eta.
\label{eq:parallel-additivity}
\end{align}
All local divergences are finite in this case, so the trace and
logarithm operations are well defined. Taking the supremum over
processes and then letting $\eta\downarrow0$ proves the reverse
inequality for a finite sum.

If a channel divergence on the right of \eqref{eq:collapse} is
infinite, then for every $L>0$ there is a pure probe in that slot
whose output divergence is at least $L$. Choose arbitrary probes
in the other slots and use \eqref{eq:parallel-choi-action} again.
Corollary~\ref{cor:state-dpi} gives nonnegative divergences for
all these normalized state pairs. If a factor is unsupported
above one, so is the tensor-product pair; if a factor has zero
moment below one, so does the tensor product by
\eqref{eq:perspective-tensor}. In either case the output divergence
is infinite. Otherwise tensor additivity gives an output
divergence at least $L$. Since $L$ is arbitrary, the process
supremum is infinite, completing the extended-real case.
\end{proof}
Geometric divergences can decrease at a limiting rank-deficient
probe; the explicit
approximating probes in Appendix~\ref{app:choi} avoid an
attainment assumption. Equation~\eqref{eq:collapse} does not
assert equality of finite-use error probabilities between
different strategy classes.

\Needspace{18\baselineskip}
\subsection{The order-one limit and comparison operations}
The applications use two variations of Theorem~\ref{thm:main}.
First, the fixed-pair limit \eqref{eq:bs-limit} gives the
Belavkin--Staszewski inequality. Second, communication converses
sometimes compare a channel with a completely positive operation
that can lose trace. We verify that deterministic process
normalization still gives a valid trace-nonincreasing comparison
output in that setting.

\begin{corollary}[Belavkin--Staszewski process inequality]
\label{cor:bs}
Let $n\geq1$ be an integer and let
$\cN_i,\cM_i\colon A_i\to B_i$ be channels between
finite-dimensional systems. Let $\cW$ have a completely positive
Choi action \eqref{eq:process-domain} satisfying the normalization
\eqref{eq:process-normalization}, with finite-dimensional external
interface $C\to D$. The Belavkin--Staszewski channel divergence
in \eqref{eq:bs-channel} satisfies
\begin{equation}\label{eq:bs-main}
\Dh\!\left(\cW(\cN_1,\ldots,\cN_n)
\middle\|\cW(\cM_1,\ldots,\cM_n)\right)
\leq\sum_{i=1}^n\Dh(\cN_i\|\cM_i).
\end{equation}
Moreover, the supremum of the output state divergence
\eqref{eq:bs-state} over all such channel-to-state processes with
arbitrary finite-dimensional output equals
$\sum_{i=1}^n\Dh(\cN_i\|\cM_i)$, as in \eqref{eq:collapse}.
Both assertions are understood in the extended-real sense.
\end{corollary}
\begin{proof}
Take $\alpha\downarrow1$ in \eqref{eq:main} using
\eqref{eq:bs-limit} for each fixed channel pair, including the
fixed output pair. If an input term is infinite, the bound is
trivial. To obtain state additivity, view states as channels
with one-dimensional input and take the same fixed-pair limit
of the tensor identity \eqref{eq:perspective-tensor} after taking
the trace and logarithm and dividing by $\alpha-1$.
If support inclusion fails in a factor,
the product pair is unsupported as well, giving infinity on
both sides. The reduction \eqref{eq:bs-channel-pure-reduction}
and the parallel construction \eqref{eq:parallel-choi-action}
then give the exact supremum by the same approximation argument
as in Corollary~\ref{cor:collapse}, without interchanging the
limit and a process optimization.
\end{proof}

\begin{proposition}[Trace-nonincreasing comparisons]\label{prop:subchannels}
Let $n\geq1$ be an integer and let
$\cN_i\colon A_i\to B_i$ be channels and
$\cM_i\colon A_i\to B_i$ completely positive,
trace-nonincreasing transformations between finite-dimensional
systems. Let $\cW$ have a completely positive Choi action $\Phi$
\eqref{eq:process-domain} satisfying
\eqref{eq:process-normalization}, with finite-dimensional
external interface $C\to D$. Define its action on the comparison
tuple through the Choi operator $\Phi(\bigotimes_iJ_{\cM_i})$
in convention \eqref{eq:choi}. The output
$\cW(\cM_1,\ldots,\cM_n)$ is completely positive and
trace nonincreasing. For $1<\alpha\leq2$, the channel divergence
with a completely positive comparison in
\eqref{eq:channel-divergence} obeys
\begin{equation}\label{eq:main-subchannels}
\Dh_\alpha\!\left(\cW(\cN_1,\ldots,\cN_n)
\middle\|\cW(\cM_1,\ldots,\cM_n)\right)
\leq\sum_{i=1}^n\Dh_\alpha(\cN_i\|\cM_i)
\end{equation}
in the extended-real sense.
\end{proposition}
\begin{proof}
The output trace bound follows from Corollary~\ref{cor:effect-bounds}
applied to $J_{\cM_i}$, since $\Tr_{B_i}J_{\cM_i}\leq I_{A_i}$.
For the divergence bound, the proof of
\eqref{eq:main-upper-effect} uses normalization only for the
auxiliary channels that complete the positive $H_i$.
It does not require the comparison operators themselves to be
channel Choi operators. The channel Choi formula for a completely
positive comparison, recalled in Appendix~\ref{app:choi},
therefore gives the same conclusion.
\end{proof}

\section{Channel discrimination and exclusion}\label{sec:discrimination}
We now turn the process inequality into constraints on observable
error probabilities. A final measurement can be included in the
process, but it is useful to keep it explicit: the quantum part
is controlled by Theorem~\ref{thm:main}, and the remaining argument
is a classical binary R\'enyi calculation.
For probability distributions $p=(p_x)_x$ and $q=(q_x)_x$ on
the same finite alphabet, the classical R\'enyi divergence is defined as
\begin{equation}\label{eq:classical-divergence}
D_\alpha(p\|q)\coloneqq\frac1{\alpha-1}
\log\sum_xp_x^\alpha q_x^{1-\alpha},
\qquad\alpha\in(0,1)\cup(1,2].
\end{equation}
For classical moment terms above one we use the perspective
convention $0^\alpha0^{1-\alpha}=0$; a term with positive
first argument and zero second argument is $+\infty$.
Below one, a term is zero if either probability is zero, and a
zero total moment gives $D_\alpha=+\infty$.
These conventions also apply to the comparison-test argument
in Section~\ref{sec:capacities}.

\subsection{Binary discrimination and a strong converse}
In binary discrimination the experiment has two possible channel
hypotheses, and one final decision must distinguish them. We first
bound the two error probabilities for any fixed experiment. Optimizing
that bound will then give a hypothesis-testing converse for the entire
process class.

Fix $n\geq1$ finite-dimensional slot interfaces $A_i\to B_i$.
Write $\mathfrak P_n$ for the class of deterministic completely positive
channel-to-state processes in Definition~\ref{def:process} on these
slots, normalized on every independent channel tuple, with arbitrary
finite-dimensional output. Fix channel tuples $(\cN_i)_{i=1}^n$ and
$(\cM_i)_{i=1}^n$ and a process $\Theta\in\mathfrak P_n$.
Under the two hypotheses its outputs are
\begin{equation}\label{eq:testing-states}
\rho=\Theta\!\left(\bigotimes_{i=1}^nJ_{\cN_i}\right),
\qquad
\sigma=\Theta\!\left(\bigotimes_{i=1}^nJ_{\cM_i}\right).
\end{equation}
All initial probes, memories, and control states are independent
of the unknown hypothesis and are part of the same $\Theta$.
A final test $0\leq T\leq I$ accepts the first hypothesis.
Define its type-I and type-II errors by
\begin{equation}\label{eq:errors}
a\coloneqq1-\Tr(T\rho),
\qquad b\coloneqq\Tr(T\sigma).
\end{equation}
This model includes every positive process tester with these slots:
its positive effects define a completely positive transformation
to a classical outcome register.

\begin{proposition}[Binary error constraints]\label{prop:binary}
Let $n\geq1$ be an integer, let
$\cN_i,\cM_i\colon A_i\to B_i$ be channels on
finite-dimensional systems, and let $\Theta$ be a deterministic
completely positive channel-to-state process in
Definition~\ref{def:process}, normalized on every independent
channel tuple. For its output states $\rho,\sigma$ in
\eqref{eq:testing-states} and a test $0\leq T\leq I$ accepting
the first hypothesis, define the errors $a,b$ by \eqref{eq:errors}. With the channel
divergence in \eqref{eq:channel-divergence}, we have
\begin{equation}\label{eq:binary-divergence}
D_\alpha\bigl((1-a,a)\|(b,1-b)\bigr)\leq \sum_{i=1}^n\Dh_\alpha(\cN_i\|\cM_i),
\qquad \alpha\in(0,1)\cup(1,2],
\end{equation}
where $D_\alpha$ is the classical divergence in
\eqref{eq:classical-divergence}.
For $1<\alpha\leq2$ and finite $\sum_{i=1}^n\Dh_\alpha(\cN_i\|\cM_i)$,
\begin{equation}\label{eq:binary-strong}
1-a\leq
2^{-\frac{\alpha-1}{\alpha}[-\log b-\sum_{i=1}^n\Dh_\alpha(\cN_i\|\cM_i)]}.
\end{equation}
The convention at $b=0$ is that finite $\sum_{i=1}^n\Dh_\alpha(\cN_i\|\cM_i)$ forces $1-a=0$.
For $0<\alpha<1$,
\begin{equation}\label{eq:binary-below}
(1-a)^\alpha b^{1-\alpha}
+a^\alpha(1-b)^{1-\alpha}
\geq 2^{-(1-\alpha)\sum_{i=1}^n\Dh_\alpha(\cN_i\|\cM_i)}.
\end{equation}
\end{proposition}
\begin{proof}
Theorem~\ref{thm:main} bounds $\Dh_\alpha(\rho\|\sigma)$ by
$\sum_{i=1}^n\Dh_\alpha(\cN_i\|\cM_i)$. The channel that records the outcome of $\{T,I-T\}$
converts these states to the distributions in
\eqref{eq:binary-divergence}. Corollary~\ref{cor:state-dpi} and
the classical reduction of \eqref{eq:state-divergence} prove
that inequality.
For $\alpha>1$, exponentiation and retention of the first
nonnegative term give
\begin{equation}\label{eq:first-binary-term}
(1-a)^\alpha b^{1-\alpha}
\leq 2^{(\alpha-1)\sum_{i=1}^n\Dh_\alpha(\cN_i\|\cM_i)}.
\end{equation}
Rearranging proves \eqref{eq:binary-strong}. If $b=0$ and $1-a>0$,
the binary divergence is infinite, contradicting finite $\sum_{i=1}^n\Dh_\alpha(\cN_i\|\cM_i)$.
For $\alpha<1$, multiplying \eqref{eq:binary-divergence} by
$\alpha-1$ reverses its direction and yields
\eqref{eq:binary-below}.
\end{proof}

For asymmetric testing, the type-I error is constrained and the type-II
error is minimized. We record the convention explicitly below.

\begin{definition}[Process hypothesis-testing divergence]
\label{def:process-hypothesis-testing}
Let $n\geq1$ be an integer, let $\cN,\cM\colon A\to B$ be channels
between finite-dimensional systems, and let $0\leq\varepsilon<1$.
Let $\mathfrak P_n$ be the class of deterministic completely positive
channel-to-state processes in Definition~\ref{def:process} on $n$
slots $A\to B$, with arbitrary finite-dimensional output. Define
\begin{align}
\beta_{\varepsilon}^{\ind}(\cN^{\otimes n}\|\cM^{\otimes n})
&\coloneqq
\inf\Bigl\{
\Tr[T\Theta(J_{\cM}^{\otimes n})]:
\Theta\in\mathfrak P_n,\ 0\leq T\leq I,\notag\\[-0.25em]
&\hspace{8em}
\Tr[T\Theta(J_{\cN}^{\otimes n})]\geq1-\varepsilon
\Bigr\},\label{eq:testing-beta}\\
D_H^{\varepsilon,\ind}(\cN^{\otimes n}\|\cM^{\otimes n})
&\coloneqq-\log
\beta_{\varepsilon}^{\ind}(\cN^{\otimes n}\|\cM^{\otimes n}).
\label{eq:testing-definition}
\end{align}
The effect $T$ acts on the output of $\Theta$ and accepts $\cN$;
the same process and test are used under both hypotheses.
The tensor-power notation specifies $n$ individually addressable
channel slots. All probes, memories, and control systems are included
in $\Theta$ and are independent of the unknown hypothesis.
The logarithm has base two, $-\log0=+\infty$, and the type-II
probability is not divided by $1-\varepsilon$. The feasible set is
nonempty, since $T=I$ is allowed, and the infimum need not be attained.
\end{definition}

The normalization follows the state hypothesis-testing relative
entropy of Wang and Renner~\cite[Eq.~(1)]{WangRenner2012}.
Related smoothing through acceptance operators appears in the
quasi-entropies of Buscemi and
Datta~\cite[Section II.B, Eqs.~(10), (12), and (14)]{BuscemiDatta2010}
and the smooth min-relative entropy of entanglement of Brand\~ao and
Datta~\cite[Definition 4, Eq.~(12)]{BrandaoDatta2011}.
These references concern quantities defined for states; the process
quantity defined here optimizes the testing criterion over the
specified class of positive process testers.

\begin{corollary}[Hypothesis-testing and strong-converse bounds]
\label{cor:hypothesis-testing}
Let $n\geq1$ be an integer and let $\cN,\cM\colon A\to B$ be
channels between finite-dimensional systems. For
$0\leq\varepsilon<1$ and $1<\alpha\leq2$, the process
hypothesis-testing divergence in
\eqref{eq:testing-definition}, optimized over the processes in
Definition~\ref{def:process-hypothesis-testing}, satisfies
\begin{equation}\label{eq:testing-oneshot}
D_H^{\varepsilon,\ind}(\cN^{\otimes n}\|\cM^{\otimes n})
\leq
n\Dh_\alpha(\cN\|\cM)
+\frac{\alpha}{\alpha-1}\log\frac1{1-\varepsilon}.
\end{equation}
Here $\Dh_\alpha$ is the channel divergence in
\eqref{eq:channel-divergence}, and the inequality is understood
in the extended-real sense.
Moreover, if $\Dh_\alpha(\cN\|\cM)<+\infty$ and a process and
test have errors $a_n,b_n$ defined by \eqref{eq:errors} and satisfying
$b_n\leq2^{-nr}$ for a real rate
$r$, then
\begin{equation}\label{eq:testing-exponent}
1-a_n\leq
2^{-n\frac{\alpha-1}{\alpha}
[r-\Dh_\alpha(\cN\|\cM)]}.
\end{equation}
If the Belavkin--Staszewski channel divergence
$\Dh(\cN\|\cM)$ in \eqref{eq:bs-channel} is finite, every
$r>\Dh(\cN\|\cM)$ therefore forces the probability $1-a_n$
of correctly accepting $\cN$ to decay exponentially, uniformly
over all allowed processes and tests satisfying the type-II constraint.
\end{corollary}
\begin{proof}
When $\Dh_\alpha(\cN\|\cM)$ is finite, every feasible test in
\eqref{eq:testing-beta} obeys, by \eqref{eq:binary-strong},
\begin{equation}\label{eq:testing-beta-lower}
b_n\geq
(1-\varepsilon)^{\alpha/(\alpha-1)}
2^{-n\Dh_\alpha(\cN\|\cM)}.
\end{equation}
Taking the infimum over these tests and then the negative logarithm
proves \eqref{eq:testing-oneshot}; no minimizing process or test is
required. An infinite channel divergence makes that inequality
vacuous. Substituting $b_n\leq2^{-nr}$ in
\eqref{eq:binary-strong} gives \eqref{eq:testing-exponent},
including $b_n=0$ by the convention in Proposition~\ref{prop:binary}.
Finally, for $r>\Dh(\cN\|\cM)$, the fixed-pair limit
\eqref{eq:bs-limit} provides one fixed $\alpha\in(1,2]$ such
that $\Dh_\alpha(\cN\|\cM)<r$. The exponent in
\eqref{eq:testing-exponent} is then positive.
\end{proof}
The last assertion is a strong-converse guarantee above the stated
threshold. It does not identify the optimal Stein threshold.
For each fixed $n$ or rate $r$, the bounds can also be optimized
over $1<\alpha\leq2$; no limit needs to be exchanged with that
optimization. There is no polynomial factor in
\eqref{eq:testing-exponent} or \eqref{eq:testing-oneshot}.

\subsection{Symmetric discrimination}
The range for $\alpha$ below one is useful when both errors are penalized.
For equal priors, define the error probability and, for $n$
identical uses of a channel pair, its process optimum by
\begin{equation}\label{eq:symmetric-error}
P_{\err}\coloneqq\frac{a+b}{2},
\qquad
P_{\err}^{\ind}(n)\coloneqq
\inf\{P_{\err}:\Theta\in\mathfrak P_n,\ 0\leq T\leq I\}.
\end{equation}
Here $a,b$ are evaluated using \eqref{eq:testing-states}--\eqref{eq:errors}
with $\cN_i=\cN$ and $\cM_i=\cM$, and $\mathfrak P_n$ is the class of processes on these $n$ slots
in Definition~\ref{def:process-hypothesis-testing}.
At $\alpha=1/2$, the binary R\'enyi moment is a classical
overlap that can be related directly to $P_{\err}$.

\begin{corollary}[Equal-prior error bound]\label{cor:symmetric}
Let $n\geq1$ be an integer and let
$\cN_i,\cM_i\colon A_i\to B_i$ be finite-dimensional
channel pairs with equal prior probabilities for the two tuples.
For every deterministic completely positive channel-to-state
process in Definition~\ref{def:process} and every binary test,
let $a,b$ be as in \eqref{eq:errors}, and let $P_{\err}=(a+b)/2$ as in
\eqref{eq:symmetric-error}. Then
\begin{equation}\label{eq:symmetric-bound}
P_{\err}\geq
\frac{1-\sqrt{1-2^{-\sum_{i=1}^n\Dh_{1/2}(\cN_i\|\cM_i)}}}{2}
\geq\frac14\,2^{-\sum_{i=1}^n\Dh_{1/2}(\cN_i\|\cM_i)}.
\end{equation}
Consequently, for $\cN_i=\cN$ and $\cM_i=\cM$ in every slot,
the optimal error $P_{\err}^{\ind}(n)$ in
\eqref{eq:symmetric-error} satisfies
\begin{equation}\label{eq:symmetric-exponent}
\limsup_{n\to\infty}-\frac1n\log P_{\err}^{\ind}(n)
\leq\Dh_{1/2}(\cN\|\cM).
\end{equation}
The channel divergence is defined in \eqref{eq:channel-divergence};
the bounds include its infinite values, with $2^{-\infty}=0$.
\end{corollary}
\begin{proof}
Equation~\eqref{eq:binary-below} implies
\begin{align}
2^{-\sum_{i=1}^n\Dh_{1/2}(\cN_i\|\cM_i)/2}
&\leq\sqrt{(1-a)b}+\sqrt{a(1-b)}
\label{eq:symmetric-cs:line-1}\\
&\leq\sqrt{(a+b)(2-a-b)}
=2\sqrt{P_{\err}(1-P_{\err})},\label{eq:symmetric-cs}
\end{align}
where the second inequality is Cauchy--Schwarz. Solving the
resulting quadratic inequality gives the first bound in
\eqref{eq:symmetric-bound}. Rationalizing its numerator gives
the second bound. Both hold uniformly over all strategies, so
they remain true after taking the infimum of their error
probabilities. For identical channel pairs, the sum is
$n\Dh_{1/2}(\cN\|\cM)$. Taking $\limsup$ proves \eqref{eq:symmetric-exponent}.
\end{proof}
If $\sum_{i=1}^n\Dh_{1/2}(\cN_i\|\cM_i)=0$, the first bound is exactly $1/2$. If
$\sum_{i=1}^n\Dh_{1/2}(\cN_i\|\cM_i)=+\infty$, it is zero. The corresponding geometric-fidelity
bound for sequential strategies was developed
in~\cite[Section 6.3, Conclusion 49]{KatariyaWilde2021}; here the process
inequality extends its validity to the general class.

\subsection{Hypothesis exclusion and a channel-radius bound}
For several hypotheses, an exclusion measurement reports a label
that is intended \emph{not} to be the true label. An error occurs
when it reports the actual label. Binary exclusion is binary
discrimination after relabeling the two outcomes.

Let $k\geq2$ and $n\geq1$ be integers. Let
$\cN_1,\ldots,\cN_k\colon A\to B$ have positive fixed priors
$p_1,\ldots,p_k$, satisfying $\sum_xp_x=1$. Each hypothesis supplies
$n$ memoryless uses of the corresponding channel. For a fixed
process $\Theta_n$, let
\begin{equation}\label{eq:exclusion-output}
\rho_x\coloneqq\Theta_n(J_{\cN_x}^{\otimes n}).
\end{equation}
An exclusion measurement $\{\Lambda_x\}_{x=1}^k$ has error
$P_{\err}=\sum_xp_x\Tr(\Lambda_x\rho_x)$.
For this task, define
\begin{equation}\label{eq:exclusion-optimum}
P_{\err,\mathrm{excl}}^{\ind}(n)
\coloneqq\inf\left\{
\sum_xp_x\Tr[\Lambda_x\Theta_n(J_{\cN_x}^{\otimes n})]:
\Theta_n\in\mathfrak P_n,\ \{\Lambda_x\}_x\text{ a POVM}
\right\},
\end{equation}
where $\mathfrak P_n$ is the class of deterministic completely
positive channel-to-state processes in Definition~\ref{def:process}
on these $n$ slots, with arbitrary finite-dimensional output.
The same process and measurement are used under all hypotheses.
A POVM here has $k$ positive effects summing to the identity.

\Needspace{17\baselineskip}
\begin{theorem}[Exclusion under general processes]\label{thm:exclusion}
Let $k\geq2$ and $n\geq1$ be integers, and let
$\cN_1,\ldots,\cN_k\colon A\to B$ be channels between
finite-dimensional systems. Fix priors $p_x>0$ satisfying
$\sum_{x=1}^k p_x=1$, independent of $n$. Each hypothesis
supplies $n$ memoryless uses of its channel. Let
$P_{\err,\mathrm{excl}}^{\ind}(n)$ be the infimum in
\eqref{eq:exclusion-optimum} over deterministic completely positive
channel-to-state processes of Definition~\ref{def:process} and
$k$-outcome measurements, with the same process and measurement
under every hypothesis; an error means reporting the true label.
Let $\mathfrak C$ be the set of all channels $\cT\colon A\to B$.
For every $1<\alpha\leq2$, the exclusion error satisfies the
following bound, where $\Dh_\alpha$ is the channel divergence
in \eqref{eq:channel-divergence}:
\begin{equation}\label{eq:exclusion-finite}
-\frac1n\log P_{\err,\mathrm{excl}}^{\ind}(n)
\leq
\inf\Biggl\{
\frac{1}{n(\alpha-1)}
\log\!\left[
\sum_{x=1}^k p_x^{1-\alpha}
2^{n(\alpha-1)\Dh_\alpha(\cT\|\cN_x)}
\right]:\cT\in\mathfrak C\Biggr\}.
\end{equation}
Furthermore,
\begin{equation}\label{eq:exclusion-radius}
\limsup_{n\to\infty}-\frac1n\log P_{\err,\mathrm{excl}}^{\ind}(n)
\leq
\inf\left\{\max_x\Dh(\cT\|\cN_x):
\cT\colon A\to B\text{ a channel}\right\}.
\end{equation}
Here $\Dh$ is the Belavkin--Staszewski channel divergence in
\eqref{eq:bs-channel}, and both bounds are understood in the
extended-real sense, with $-\log0=+\infty$.
\end{theorem}
\paragraph{Proof strategy.}
An auxiliary channel produces a common reference distribution
$q_x$ satisfying $\sum_x q_x=1$. Bounding each $q_x$ by the corresponding
erroneous-report probability forces a lower bound on the total
exclusion error. This is why the useful divergence has the auxiliary
channel in its first argument, $\Dh_\alpha(\cT\|\cN_x)$.
\begin{proof}
The auxiliary channel supplies one common reference output,
\begin{equation}\label{eq:exclusion-reference}
\omega_{\cT}\coloneqq\Theta_n(J_{\cT}^{\otimes n}).
\end{equation}
It is normalized by the process assumption. Define
$q_x=\Tr(\Lambda_x\omega_{\cT})$ and
$e_x=\Tr(\Lambda_x\rho_x)$, so $\sum_xq_x=1$.
Theorem~\ref{thm:main}, followed by the binary measurement
$\{\Lambda_x,I-\Lambda_x\}$, gives
\begin{equation}\label{eq:exclusion-one-outcome}
q_x\leq
e_x^{(\alpha-1)/\alpha}
2^{n(\alpha-1)\Dh_\alpha(\cT\|\cN_x)/\alpha}.
\end{equation}
An infinite divergence makes the final bound trivial for that
choice of $\cT$; otherwise the inequality includes the zero-probability
cases by support inclusion.

Set $d_x=\Dh_\alpha(\cT\|\cN_x)$. To introduce the priors,
factor each term on the right of
\eqref{eq:exclusion-one-outcome} as
\begin{equation}\label{eq:exclusion-factorization}
e_x^{(\alpha-1)/\alpha}2^{n(\alpha-1)d_x/\alpha}
=(p_xe_x)^{(\alpha-1)/\alpha}
\left(p_x^{1-\alpha}2^{n(\alpha-1)d_x}\right)^{1/\alpha}.
\end{equation}
Summing and applying H\"older with exponents
$\alpha/(\alpha-1)$ and $\alpha$ separates the total exclusion
error from the comparison-divergence factor:
\begin{equation}\label{eq:exclusion-holder}
1\leq
\left(\sum_xp_xe_x\right)^{(\alpha-1)/\alpha}
\left(\sum_xp_x^{1-\alpha}2^{n(\alpha-1)d_x}\right)^{1/\alpha}.
\end{equation}
Rearrange, take logarithms, and optimize over $\cT$. The estimate
is uniform over all processes and measurements, proving
\eqref{eq:exclusion-finite}.

In finite dimension, the supported Choi expression and the normalized
maximally entangled probe give
$\Dh(\cT\|\cN_x)<+\infty$ if and only if
$J_{\cT}\ll J_{\cN_x}$; see Appendix~\ref{app:choi} and
\eqref{eq:bs-choi-limit}. This is also the support condition for
finite $\Dh_\alpha(\cT\|\cN_x)$ at every $1<\alpha\leq2$
in \eqref{eq:channel-choi-formula}. Thus every auxiliary channel
contributing a finite radius is eligible for the following argument.
For every fixed $\cT$ satisfying $d_x<+\infty$ for every $x$,
the logarithmic sum in
\eqref{eq:exclusion-finite}, divided by $n(\alpha-1)$, converges
to $\max_xd_x$. Therefore the left side of
\eqref{eq:exclusion-radius} is at most $\max_xd_x$ for every
fixed $\cT$ and fixed $\alpha>1$. Let $\alpha\downarrow1$
using \eqref{eq:bs-limit} for this finite collection, and then
take the infimum over $\cT$. This order of operations requires
no interchange of an infimum and a limit. If there is no channel
with all comparison divergences finite, the claimed upper bound
is infinite and the assertion is automatic.
\end{proof}

The right side of \eqref{eq:exclusion-radius} is the left
Belavkin--Staszewski channel radius studied in
\cite[Theorem 26]{JiEtAl2026}.
That paper establishes its barycentric representation and its
use as an adaptive exclusion converse. Its Remark 30 asks
whether the geometric bound extends beyond definite causal order;
Theorem~\ref{thm:exclusion} gives this extension for the completely
positive process class. The separately affine abstract class
treated by its Appendix H is broader, so no claim about every
transformation in that broader class is needed here.

The binary specialization was noted for adaptive strategies
in~\cite[Remark 28]{JiEtAl2026}. We state its consequence for
the present process class explicitly.

\begin{corollary}[Binary discrimination from exclusion]
\label{cor:binary-radius}
Let $\cN,\cM\colon A\to B$ be channels between finite-dimensional
systems, with equal prior probabilities. Let $n\geq1$, and let
$P_{\err}^{\ind}(n)$ be the optimal error defined by the infimum in
\eqref{eq:symmetric-error} over deterministic completely positive
channel-to-state processes of Definition~\ref{def:process} on $n$
slots and binary tests. Write $\mathfrak C$ for the set of channels
$A\to B$. For every $1<\alpha\leq2$,
\begin{equation}\label{eq:binary-exclusion-finite}
P_{\err}^{\ind}(n)
\geq\frac12\sup\Biggl\{
\left[
2^{n(\alpha-1)\Dh_\alpha(\cT\|\cN)}
+2^{n(\alpha-1)\Dh_\alpha(\cT\|\cM)}
\right]^{-1/(\alpha-1)}:
\cT\in\mathfrak C\Biggr\}.
\end{equation}
Define the left Belavkin--Staszewski radius of the pair by
\begin{equation}\label{eq:binary-bs-radius}
\widehat R(\cN,\cM)
\coloneqq
\inf\left\{
\max\{\Dh(\cT\|\cN),\Dh(\cT\|\cM)\}:
\cT\in\mathfrak C\right\},
\end{equation}
using the channel divergence in \eqref{eq:bs-channel}. Then
\begin{equation}\label{eq:binary-combined-exponent}
\limsup_{n\to\infty}-\frac1n\log P_{\err}^{\ind}(n)
\leq
\min\left\{
\Dh_{1/2}(\cN\|\cM),\widehat R(\cN,\cM)
\right\}.
\end{equation}
The finite-use bound uses the channel divergence in
\eqref{eq:channel-divergence}. Its bracket contributes zero after
the negative power if either divergence is infinite. All exponent
bounds use $-\log0=+\infty$ and hold in the extended-real sense.
\end{corollary}
\begin{proof}
For a discrimination test $T$ that accepts $\cN$, choose the
exclusion effects $\Lambda_{\cN}=I-T$ and
$\Lambda_{\cM}=T$. Their equal-prior exclusion error is
\begin{equation}\label{eq:binary-exclusion-identification}
\frac12\Tr[(I-T)\rho]+\frac12\Tr[T\sigma]
=\frac{a+b}{2}.
\end{equation}
Conversely, every two-outcome exclusion measurement has this form,
so the two optimized errors coincide. Setting $k=2$ and
$p_1=p_2=1/2$ in \eqref{eq:exclusion-finite}, exponentiating,
and optimizing over the auxiliary channel gives
\eqref{eq:binary-exclusion-finite}. The factor $1/2$ comes from
$p_x^{1-\alpha}=2^{\alpha-1}$.
Equation~\eqref{eq:exclusion-radius} gives the exponent bound
$\widehat R(\cN,\cM)$, and Corollary~\ref{cor:symmetric} gives
$\Dh_{1/2}(\cN\|\cM)$. Taking the smaller right side proves
\eqref{eq:binary-combined-exponent}.
\end{proof}

Both bounds constrain how small the error can be: a larger lower
bound on $P_{\err}^{\ind}(n)$ is stronger, while a smaller upper
bound on its decay exponent is stronger. The radius uses a
common auxiliary channel in the first argument of two divergences.
It is therefore a different comparison from the direct order-$1/2$
divergence in \eqref{eq:symmetric-exponent}. The minimum in
\eqref{eq:binary-combined-exponent} is valid without assuming an
ordering between those two quantities. The following commuting
example shows that the radius can give a strict improvement.

\begin{example}[A strict improvement of the binary exponent bound]
\label{ex:binary-radius-improvement}
Let $A$ be one-dimensional and $B$ two-dimensional. For
$0<p<1$, $p\ne1/2$, let $\cN$ and $\cM$ prepare, respectively,
\begin{equation}\label{eq:binary-radius-example-states}
\rho_p=\operatorname{diag}(p,1-p),\qquad
\sigma_p=\operatorname{diag}(1-p,p).
\end{equation}
A channel on this input prepares a state, so its channel divergences
in \eqref{eq:channel-divergence} and \eqref{eq:bs-channel}
are the corresponding state divergences. The radius in
\eqref{eq:binary-bs-radius} satisfies
\begin{equation}\label{eq:binary-radius-example-values}
\widehat R(\cN,\cM)
=-\log\!\left(2\sqrt{p(1-p)}\right)
=\frac12\Dh_{1/2}(\cN\|\cM)
<\Dh_{1/2}(\cN\|\cM).
\end{equation}
To verify the radius, dephase an arbitrary auxiliary state in the
common eigenbasis of $\rho_p$ and $\sigma_p$. The data-processing
inequality \eqref{eq:bs-state-data-processing} shows that this cannot
increase either comparison divergence, so
it suffices to optimize over $\operatorname{diag}(t,1-t)$ satisfying
$0\leq t\leq1$. Write
$h_2(t)=-t\log t-(1-t)\log(1-t)$, with $0\log0=0$.
The larger of the two classical relative entropies is at least their
average, which equals
\begin{equation}\label{eq:binary-radius-example-average}
-h_2(t)-\frac12\log[p(1-p)]
\geq-\log\!\left(2\sqrt{p(1-p)}\right).
\end{equation}
The uniform auxiliary state $t=1/2$ attains equality in both
comparisons. The order-$1/2$ divergence is
$-2\log(2\sqrt{p(1-p)})$, proving
\eqref{eq:binary-radius-example-values}. Thus the upper bound on
the error exponent is halved for this family; no achievability
claim is needed for the comparison.

The finite-use bound can improve as well. For $p=9/10$, the
uniform auxiliary state and the fixed order $\alpha=3/2$ have
$\Dh_{3/2}(\cT\|\cN)=\Dh_{3/2}(\cT\|\cM)=\log(20/9)$.
Equation~\eqref{eq:binary-exclusion-finite} then gives
\begin{equation}\label{eq:binary-radius-example-finite}
P_{\err}^{\ind}(n)\geq\frac18\left(\frac9{20}\right)^n.
\end{equation}
At $n=4$, this is $0.00512578125$, whereas the full lower bound
in \eqref{eq:symmetric-bound} is
$(1-\sqrt{1-(9/25)^4})/2\approx0.00421682$.
\end{example}

The two finite-use lower bounds do not dominate one another
uniformly. For identical channels, \eqref{eq:symmetric-bound}
gives the exact value $1/2$, whereas optimizing
\eqref{eq:binary-exclusion-finite} over the stated order range
and the auxiliary channel gives $1/4$. Indeed, both divergences
are nonnegative, and choosing $\cT=\cN=\cM$ makes them zero;
the remaining factor is $2^{-1-1/(\alpha-1)}$, whose maximum
on $1<\alpha\leq2$ is $1/4$. Thus one should take the larger
of the two finite-use lower bounds and the smaller of their
asymptotic exponent bounds. In particular, the order-one radius
is obtained after taking the large-$n$ limit at fixed auxiliary
channel and fixed $\alpha>1$, as in Theorem~\ref{thm:exclusion};
it is not obtained by substituting $\alpha=1$ into
\eqref{eq:binary-exclusion-finite}.

\section{Communication and resource-generation converses}\label{sec:capacities}
A divergence bound becomes a coding converse when a successful output
passes a test that comparison outputs rarely pass. For processes without
definite order, one must state which comparison outputs the process
preserves. Deterministic normalization alone does not impose this
resource restriction. We first give a common test argument and then
specialize it to three operational settings.

\subsection{A comparison-test principle}
The comparison experiment replaces each physical channel by an
operation from a specified comparison set. The essential requirement
is that the complete experiment then produces an output that has
small acceptance probability for the target test. Figure~\ref{fig:comparison}
illustrates the two experiments to which we apply the divergence bound.

For each slot, let $\cF_i$ be a nonempty set of completely positive,
trace-nonincreasing comparison transformations. Let $\mathfrak S$
be a set of positive output operators of trace at most one.
A channel-to-state process $\mathcal P$ satisfies the required preservation
condition when
\begin{equation}\label{eq:comparison-preservation}
\mathcal P\!\left(\bigotimes_iJ_{\cM_i}\right)\in\mathfrak S
\quad\text{for every }\cM_i\in\cF_i.
\end{equation}
This is a condition on the complete protocol, including its initial
resources and final decoding. It can be verified from a physical
construction or imposed on an admissible process class. General
resource converses for processes, including indefinite order, have
also been developed using robustness and resource
weight~\cite[Theorem 2 and Supplementary Theorem 13]{RegulaTakagi2021}.
The following statement supplies
a finite-order geometric R\'enyi version.

\begin{theorem}[Comparison-test converse]\label{thm:comparison-test}
Let $n\geq1$ be an integer and let $\mathcal P$ be a
deterministic completely positive
channel-to-state process as in Definition~\ref{def:process},
with finite-dimensional slot interfaces $A_i\to B_i$ and output
system $E$, normalized on every independent channel tuple.
Let the nonempty comparison
sets $\cF_i$ consist of completely positive, trace-nonincreasing
transformations $\cM_i\colon A_i\to B_i$, and suppose
\eqref{eq:comparison-preservation} holds for a set
$\mathfrak S\subseteq\{\sigma\in\cL(E):\sigma\geq0,\ \Tr\sigma\leq1\}$.
Suppose $0\leq T\leq I_E$ is a test satisfying
$\Tr(T\sigma)\leq2^{-\ell}$ for every $\sigma\in\mathfrak S$,
where $\ell\geq0$. For channels $\cN_i\colon A_i\to B_i$,
use the Choi convention \eqref{eq:choi} and the channel divergence
\eqref{eq:channel-divergence} to set
\begin{equation}\label{eq:resource-distance}
d_{\alpha,i}(\cN_i)
\coloneqq
\inf\{\Dh_\alpha(\cN_i\|\cM_i):\cM_i\in\cF_i\},
\qquad 1<\alpha\leq2.
\end{equation}
If $s=\Tr[T\mathcal P(\bigotimes_iJ_{\cN_i})]$, then
\begin{equation}\label{eq:comparison-success}
s\leq
2^{-\frac{\alpha-1}{\alpha}
\left[\ell-\sum_{i=1}^n d_{\alpha,i}(\cN_i)\right]}.
\end{equation}
In particular, for $0\leq\varepsilon<1$, the condition $s\geq1-\varepsilon$
implies
\begin{equation}\label{eq:comparison-size}
\ell\leq
\sum_{i=1}^n d_{\alpha,i}(\cN_i)
+\frac{\alpha}{\alpha-1}\log\frac1{1-\varepsilon}.
\end{equation}
\end{theorem}
\begin{proof}
If some $d_{\alpha,i}(\cN_i)$ is infinite, both conclusions
are trivial. Otherwise fix an arbitrary tolerance $\eta>0$
and choose, independently in each nonempty $\cF_i$, a
finite-divergence comparison satisfying
$\Dh_\alpha(\cN_i\|\cM_i)\leq d_{\alpha,i}(\cN_i)+\eta/n$.
Write
$\rho=\mathcal P(\bigotimes_iJ_{\cN_i})$,
$\sigma=\mathcal P(\bigotimes_iJ_{\cM_i})$, and $t=\Tr(T\sigma)$.
Proposition~\ref{prop:subchannels} bounds the divergence of this
normalized first output and subnormalized second output.
The measurement $\{T,I-T\}$ and the first-outcome term give
\begin{equation}\label{eq:resource-binary}
s^\alpha t^{1-\alpha}
\leq 2^{(\alpha-1)\Dh_\alpha(\rho\|\sigma)}
\leq
2^{(\alpha-1)\sum_{i=1}^n\Dh_\alpha(\cN_i\|\cM_i)}.
\end{equation}
The measurement channel $\{T,I-T\}$ is trace preserving, so
Corollary~\ref{cor:state-dpi} applies even though its second input
is subnormalized. Retaining the first-outcome term gives the first
inequality in \eqref{eq:resource-binary}.
If the divergence is finite and $t=0$, support inclusion forces
$s=0$. Otherwise substitute $t\leq2^{-\ell}$ in
\eqref{eq:resource-binary} and rearrange.
The sum of the comparison divergences exceeds the sum of
the infima in \eqref{eq:resource-distance} by at most $\eta$.
Letting $\eta\downarrow0$ proves
\eqref{eq:comparison-success} without an attainment assumption.
Its logarithmic form is \eqref{eq:comparison-size}.
\end{proof}

Figure~\ref{fig:comparison} separates the two hypotheses used in this
argument. The output test is a mathematical tool for the converse
and need not be implementable by the communicating parties.
\begin{figure}[!t]
\centering
\begin{tikzpicture}[font=\small,>=Latex,
 block/.style={draw=blue!55!black,fill=blue!4,rounded corners=2pt,
   minimum width=1.55cm,minimum height=.95cm,
   inner xsep=9pt,inner ysep=8pt},
 test/.style={draw=black!65,fill=black!3,rounded corners=2pt,
   minimum width=1.1cm,minimum height=.95cm,
   inner xsep=9pt,inner ysep=8pt},
 wire/.style={->,line width=.65pt},
 output/.style={align=left,text width=4.2cm,inner xsep=6pt,inner ysep=7pt}]
\node[inner sep=5pt] (n) at (-5.20,1.05) {$\bigotimes_i J_{\mathcal N_i}$};
\node[inner sep=5pt] (m) at (-5.20,-1.05) {$\bigotimes_i J_{\mathcal M_i}$};
\node[block] (tn) at (-2.65,1.05) {$\mathcal P$};
\node[block] (tm) at (-2.65,-1.05) {$\mathcal P$};
\node[test] (an) at (.90,1.05) {$T$};
\node[test] (am) at (.90,-1.05) {$T$};
\node[output,anchor=west] (pn) at (2.65,1.05)
 {acceptance probability\\[.45em]$\Tr(T\rho)\geq1-\varepsilon$};
\node[output,anchor=west] (pm) at (2.65,-1.05)
 {unconditioned acceptance\\[.45em]$\Tr(T\sigma)\leq2^{-\ell}$};
\draw[wire,shorten <=3pt] (n.east)--(tn.west);
\draw[wire,shorten <=3pt] (m.east)--(tm.west);
\draw[wire] (tn.east)--node[above=5pt,inner sep=3pt]{$\rho$}(an.west);
\draw[wire] (tm.east)--node[above=5pt,inner sep=3pt]{$\sigma\in\mathfrak S$}
 node[below=5pt,font=\footnotesize,inner sep=3pt]{$\Tr\sigma\leq1$}(am.west);
\draw[wire,shorten >=4pt] (an.east)--(pn.west);
\draw[wire,shorten >=4pt] (am.east)--(pm.west);
\end{tikzpicture}

\caption{A capacity converse compares two experiments with the same
process $\mathcal P$ and target test $T$. Resource channels produce
$\rho$, which must pass with probability at least $1-\varepsilon$.
Allowed comparison operations produce $\sigma\in\mathfrak S$
with $\Tr\sigma\leq1$. Its unconditioned acceptance weight
$\Tr(T\sigma)$ is at most $2^{-\ell}$. The process
inequality bounds the divergence between these outputs.
Preservation of $\mathfrak S$ is an additional assumption on the
complete communication protocol; it is not implied by deterministic
normalization. The diagram is a comparison of experiments, not
two stages of one protocol.}
\label{fig:comparison}
\end{figure}
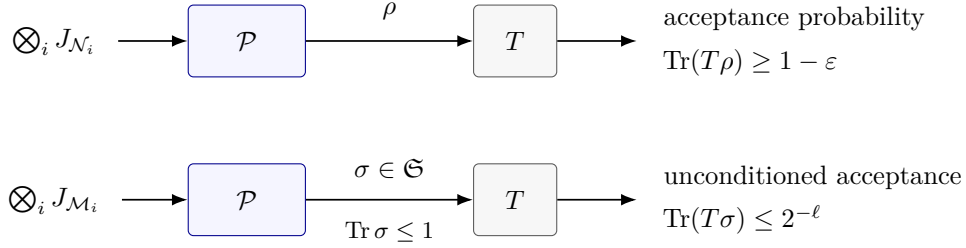

For the asymptotic consequence, use $\cN$ in every slot and
one common comparison class $\cF$ that is independent of both
the slot and the blocklength. Then \eqref{eq:resource-distance}
defines one quantity $d_\alpha(\cN)$ and the sum in
\eqref{eq:comparison-success} is $n d_\alpha(\cN)$.
For a target $\ell_n\geq nr$, a fixed R\'enyi order
$1<\alpha\leq2$ gives the success bound
$2^{-n(\alpha-1)[r-d_\alpha(\cN)]/\alpha}$.
Thus every $r>\inf\{d_\alpha(\cN):1<\alpha\leq2\}$ is a
strong-converse rate for a class obeying
\eqref{eq:comparison-preservation} at every blocklength.
For rates specified by $\liminf_n\ell_n/n\geq r$, first use a
slightly smaller rate and then take its limit.
A capacity, defined as the supremum of rates achievable with
error tending to zero, is bounded by this infimum.
This statement does not identify the infimum with a limit of
optimized order-one quantities.

Table~\ref{tab:comparisons} summarizes the three specializations.
Their output conditions apply after the complete protocol,
including its preparation and decoding.
\begin{table}[t]
\centering
\small
\begin{tabular}{@{}>{\raggedright\arraybackslash}p{.24\linewidth}
>{\raggedright\arraybackslash}p{.34\linewidth}
>{\raggedright\arraybackslash}p{.34\linewidth}@{}}
\toprule
Target and error & Local comparisons & Required comparison output \\
\midrule
Entanglement overlap, \eqref{eq:entanglement-target}
& Completely positive operations in $\mathfrak V_\Theta$,
\eqref{eq:theta-class}
& Rains set across the final laboratories,
\eqref{eq:rains-preservation} \\[.5em]
Uniform-message decoding
& Replacer channels $\mathcal R_{\tau_i}$
& Message-independent estimate,
\eqref{eq:replacer-preservation} \\[.5em]
Complete private-state trace distance,
\eqref{eq:private-error}
& Entanglement-breaking channels
& Separable state across both complete laboratories \\
\bottomrule
\end{tabular}
\caption{Comparison conditions used in the communication converses.
Each row is a separate application of Theorem~\ref{thm:comparison-test}.
Deterministic normalization does not by itself imply any of the
listed output conditions.}
\label{tab:comparisons}
\end{table}

\subsection{Entanglement generation and quantum communication}
We first apply Theorem~\ref{thm:comparison-test} to the task of creating
shared entanglement. The target is a maximally entangled state, and
partial transposition provides a norm that separates this target from
the comparison outputs.

Alice and Bob start without shared entanglement. We bound the
process class satisfying the Rains-set preservation condition
\eqref{eq:rains-preservation}. Ordinary sequential protocols
with completely positive-partial-transpose-preserving assistance
belong to this class; the verification following
Corollary~\ref{cor:quantum} is given in
\eqref{eq:rains-transfer}--\eqref{eq:rains-sequential}. Such an
assisting channel remains completely positive when Bob's input
and output systems are partially transposed before and after
its action, respectively.
Their protocol
may use local memories and specified assisting operations.
Its final state is $\rho_{A_fB_f}$, with local target dimensions
$K$. For $0\leq\varepsilon<1$, its error criterion is
\begin{equation}\label{eq:entanglement-target}
\Tr(\Phi_{A_fB_f}^{K}\rho_{A_fB_f})\geq1-\varepsilon,
\qquad
\ket{\Phi^{K}}_{A_fB_f}
\coloneqq\frac1{\sqrt K}\sum_{j=1}^{K}\ket{j}_{A_f}\ket{j}_{B_f}.
\end{equation}
Here $\Phi^K=\ket{\Phi^K}\bra{\Phi^K}$. Entanglement transmission
is included by preparing this entangled state entirely in Alice's
laboratory, retaining one reference system there, and encoding
the other system. No initially shared entanglement is thereby
introduced.

Let $\mathsf t_B$ denote matrix transposition in a fixed orthonormal
basis of $B$ and write
$\sigma^{\transpose_B}=(\id_A\otimes\mathsf t_B)(\sigma)$.
For $B_f$, we use the basis in \eqref{eq:entanglement-target}.
A normalized separable state remains positive under partial
transposition, so its partial-transpose trace norm is one.
For the target, partial transposition gives
\begin{equation}\label{eq:maxent-partial-transpose}
(\Phi^K)^{\transpose_{B_f}}
=\frac1K\sum_{j,k=1}^K
\ket{j}\bra{k}_{A_f}\otimes\ket{k}\bra{j}_{B_f}
=\frac1K S,
\end{equation}
where $S(\ket{j}\otimes\ket{k})=\ket{k}\otimes\ket{j}$ is
the swap operator. Since $S^\dagger S=I_{A_fB_f}$, all of its
$K^2$ singular values equal one. Consequently,
\begin{equation}\label{eq:maxent-partial-transpose-norm}
\norm{(\Phi^K)^{\transpose_{B_f}}}_1
=\frac1K\Tr\sqrt{S^\dagger S}
=\frac{K^2}{K}=K.
\end{equation}
We therefore track this norm in the comparison experiment.
The Rains set collects the positive operators for which it is
at most one. The diamond-norm condition in \eqref{eq:theta-class}
controls this
norm after a comparison operation even when the transmitted
system is correlated with retained memories.
Define the Rains set and a class of comparison operations by
\begin{align}
\Rains(A;B)
&\coloneqq\{\sigma_{AB}\geq0:
\norm{\sigma_{AB}^{\transpose_B}}_1\leq1\},
\label{eq:rains-set}\\
\mathfrak V_\Theta(A\to B)
&\coloneqq\{\cM\colon A\to B\text{ completely positive}:
\norm{\mathsf t_B\circ\cM}_{\diamond}\leq1\}.
\label{eq:theta-class}
\end{align}
The diamond norm is defined as the stabilized induced trace norm:
$\norm{\Lambda}_\diamond=
\sup\{\norm{(\id_R\otimes\Lambda)(X)}_1:
X\in\cL(RA),\ \norm{X}_1=1,\ R\simeq A\}$
for $\Lambda\colon\cL(A)\to\cL(B)$.
Every element of \eqref{eq:rains-set} has trace at most one.
Every comparison in \eqref{eq:theta-class} is trace nonincreasing,
because for a state $\eta$ its positive output has trace at most
$\norm{\mathsf t_B\cM(\eta)}_1\leq1$.
These are the comparison sets used in Rains-type
converses~\cite{Rains2001,TomamichelWildeWinter2017,FangFawzi2021}.
The subscript $\Theta$ is the conventional name of this
comparison class and is unrelated to the protocol $\mathcal P$.

We require the complete protocol to satisfy
\begin{equation}\label{eq:rains-preservation}
\mathcal P\!\left(\bigotimes_iJ_{\cM_i}\right)
\in\Rains(A_f;B_f)
\quad\text{whenever }\cM_i\in\mathfrak V_\Theta(A_i\to B_i).
\end{equation}
The maximally entangled projector has comparison acceptance
at most $1/K$. By \eqref{eq:maxent-partial-transpose},
$(\Phi^K)^{\transpose_{B_f}}=S/K$. Since the swap operator is
unitary, $\norm{S}=1$, and trace-norm duality gives
\begin{equation}\label{eq:rains-overlap}
\Tr(\Phi^K\sigma)
=\frac1K\Tr(S\sigma^{\transpose_{B_f}})
\leq\frac1K\norm{\sigma^{\transpose_{B_f}}}_1
\leq\frac1K.
\end{equation}
This gives the following direct capacity application.

\Needspace{24\baselineskip}
\begin{corollary}[Entanglement-generation converse]\label{cor:quantum}
Let $n,K\geq1$ be integers, and let
$\cN_i\colon A_i\to B_i$ be channels on finite-dimensional
systems. Let $\mathcal P$ be a deterministic completely positive
channel-to-state process in Definition~\ref{def:process}, with
output $\rho_{A_fB_f}=\mathcal P(\bigotimes_iJ_{\cN_i})$ on
two $K$-dimensional target systems. Suppose the complete protocol
satisfies \eqref{eq:rains-preservation}: every tuple of comparisons
from $\mathfrak V_\Theta$ in \eqref{eq:theta-class} produces an
operator in the Rains set \eqref{eq:rains-set} across the final
laboratories. If the target overlap
$\Tr(\Phi_{A_fB_f}^{K}\rho_{A_fB_f})\geq1-\varepsilon$ in
\eqref{eq:entanglement-target} holds for $0\leq\varepsilon<1$,
then, for every $1<\alpha\leq2$,
\begin{equation}\label{eq:quantum-converse}
\log K\leq
\sum_{i=1}^n\widehat R_{\alpha,\Theta}(\cN_i)
+\frac{\alpha}{\alpha-1}\log\frac1{1-\varepsilon},
\end{equation}
where, for a channel $\cN\colon A\to B$,
\begin{equation}\label{eq:theta-information}
\widehat R_{\alpha,\Theta}(\cN)
\coloneqq
\inf\{\Dh_\alpha(\cN\|\cM):
\cM\in\mathfrak V_\Theta(A\to B)\}.
\end{equation}
Here $\Dh_\alpha$ is the channel divergence in
\eqref{eq:channel-divergence}. For a sequence of such protocols
using one fixed channel $\cN$ in all $n$ slots and satisfying
\eqref{eq:rains-preservation} for every $n$, let the target
dimension be $K_n$. If
$\liminf_{n\to\infty}n^{-1}\log K_n\geq r$ and
$r>\inf\{\widehat R_{\alpha,\Theta}(\cN):1<\alpha\leq2\}$,
then the target overlap tends to zero exponentially in $n$.
\end{corollary}
\begin{proof}
Apply Theorem~\ref{thm:comparison-test} with
$T=\Phi^K$, $\ell=\log K$, and the comparison sets
\eqref{eq:rains-set}--\eqref{eq:theta-class}.
Equation~\eqref{eq:rains-overlap} verifies its test hypothesis.
\end{proof}
The quantity in \eqref{eq:theta-information} is Fang and Fawzi's
geometric R\'enyi \emph{Theta-information}
\cite[Definition 14 and Theorem 17]{FangFawzi2021}.
It differs from their geometric R\'enyi Rains information, which
optimizes a state Rains divergence over channel inputs.
Their semidefinite representations at
$\alpha=1+2^{-m}$, $m\in\{0,1,\ldots\}$, remain available for
the same single-letter comparison quantity
\cite[Proposition 19]{FangFawzi2021}.

For clarity, \eqref{eq:rains-preservation} can be verified for
ordinary sequential protocols with completely
positive-partial-transpose-preserving assistance. Such an assisting
channel $\mathcal L$, with Bob's registers denoted by $R_B$
before it and $R_B'$ afterward, has
$\mathsf t_{R_B'}\circ\mathcal L\circ\mathsf t_{R_B}$ completely positive
and trace preserving, and hence contracts the trace norm of a
partially transposed Hermitian input.

To verify the effect of a transmitted system, let
$\sigma_{R_A A R_B}$ be the current positive operator, with
Alice holding $R_AA$ and Bob holding $R_B$.
A comparison operation $\cM\colon A\to B$ delivers $B$ to
Bob, producing
\begin{equation}\label{eq:rains-transfer}
\omega_{R_A B R_B}
\coloneqq(\id_{R_A}\otimes\cM\otimes\id_{R_B})
(\sigma_{R_A A R_B}).
\end{equation}
Transposing Bob's complete output gives
\begin{equation}\label{eq:rains-transpose-identity}
\omega^{\transpose_{B R_B}}
=(\id_{R_A}\otimes(\mathsf t_B\circ\cM)\otimes\id_{R_B})
(\sigma^{\transpose_{R_B}}).
\end{equation}
The transposition on his old register $R_B$ commutes with the
operation on $A$. The diamond norm then bounds
\begin{equation}\label{eq:rains-sequential}
\norm{\omega^{\transpose_{B R_B}}}_1
\leq
\norm{\mathsf t_B\circ\cM}_\diamond
\norm{\sigma^{\transpose_{R_B}}}_1.
\end{equation}
The retained systems $R_A R_B$ act as the reference in this
inequality. To check that their dimension is unrestricted here,
diagonalize the Hermitian operator $\sigma^{\transpose_{R_B}}$.
Each normalized eigenvector has Schmidt rank at most $d_A$
across $A;R_A R_B$, so its reference support fits in the
diamond-norm definition. Apply that bound to each eigenvector
projector and use the triangle inequality; the sum of the
absolute eigenvalues is $\norm{\sigma^{\transpose_{R_B}}}_1$.
Starting with a separable state, whose partial-transpose trace
norm is one, induction proves \eqref{eq:rains-preservation}.
Local operations and two-way classical communication are among
these assisting channels. Thus Corollary~\ref{cor:quantum}
recovers a definite-order converse and extends it to every
process without definite order that satisfies the same global
preservation condition. Preservation only for normalized
positive-partial-transpose channels would be a weaker assumption
and is not substituted for \eqref{eq:rains-preservation}.

\subsection{Classical communication}
For classical communication, a successful output correlates Bob's
estimate with Alice's message. A useful comparison replaces every
forward channel by a fixed-output channel. When the complete comparison
experiment leaves Bob's state independent of the message, the target
test succeeds with probability only $1/K$. This is the hypothesis of
Theorem~\ref{thm:comparison-test} that we verify here.

Let $W$ be a uniform message with $K$ possible values and let
$\widehat W$ be Bob's estimate. The encoding, channel-slot process,
and decoding define a channel-to-state process with output
$\rho_{W\widehat W}$. Preshared entanglement independent of $W$
may be included. The bound concerns protocols satisfying
\eqref{eq:replacer-preservation}; ordinary protocols with
feedback and no other forward communication satisfy it by
Lemma~\ref{lem:feedback-replacer-preservation}.
A comparison channel is a replacer
$\mathcal R_\tau(X)=\Tr(X)\tau$. We impose
\begin{equation}\label{eq:replacer-preservation}
\mathcal P\!\left(\bigotimes_iJ_{\mathcal R_{\tau_i}}\right)
=\pi_W\otimes\omega_{\widehat W}
\quad\text{for every tuple of states }(\tau_i)_i,
\end{equation}
where $\pi_W=I_W/K$ and $\omega_{\widehat W}$ may depend on
the tuple. This condition says that after all communication
slots are replaced, Bob's estimate contains no message information.

The success projector
$T=\sum_{w=1}^K\ket{w}\bra{w}_W\otimes
\ket{w}\bra{w}_{\widehat W}$ has acceptance $1/K$
on \eqref{eq:replacer-preservation}. For $0\leq\varepsilon<1$,
a code with average decoding error at most $\varepsilon$ satisfies
\begin{equation}\label{eq:classical-converse}
\log K\leq
\sum_{i=1}^n\inf\{\Dh_\alpha(\cN_i\|\mathcal R_{\tau_i}):
\tau_i\in\cD(B_i)\}
+\frac{\alpha}{\alpha-1}\log\frac1{1-\varepsilon},
\qquad 1<\alpha\leq2.
\end{equation}
The corresponding exponential success bound follows from
\eqref{eq:comparison-success}.

Feedback-assisted communication has an established converse theory.
Bowen proved that noiseless quantum feedback does not increase
the entanglement-assisted classical capacity of a memoryless
quantum channel~\cite[Theorem 3]{Bowen2004}.
Cooney, Mosonyi, and Wilde gave a finite-block exponential converse
in terms of sandwiched R\'enyi information
\cite[Theorem 4]{CooneyMosonyiWilde2016}.
Their comparison experiment replaces the forward channel by a
replacer and makes Bob's state independent of the message
\cite[Section 5, Eq.~(5.5)]{CooneyMosonyiWilde2016}.
For classical feedback without initial shared entanglement,
Ding and Wilde proved a strong converse for classical communication
over entanglement-breaking channels
\cite[Section 3.2, Theorem 5, and Corollary 6]{DingWilde2018}.
Their exponential success bound also holds for arbitrary channels
when the encoder does not entangle inputs across different
channel uses.

Equation~\eqref{eq:classical-converse} uses the geometric
R\'enyi divergence, permits nonidentical forward slots, and applies
to the larger class of processes satisfying
\eqref{eq:replacer-preservation}. The following verification
identifies the ordinary feedback protocols contained in that class.

\begin{lemma}[Replacer preservation for feedback protocols]
\label{lem:feedback-replacer-preservation}
Let $n,K\geq1$ be integers. Consider a finite-dimensional
sequential protocol with $n$ forward slots
$\cN_i\colon A_i\to B_i$, used to communicate a uniform
message $W\in\{1,\ldots,K\}$ from Alice to Bob.
The parties may share an initial state independent of $W$, retain
arbitrary finite-dimensional local memories and classical records,
apply local channels, and transmit classical or quantum feedback
from Bob to Alice. Their only communication from Alice to Bob
is through the specified forward slots. The protocol includes
all measurement outcomes and a trace-preserving final decoder.
If every forward slot is replaced by
$\mathcal R_{\tau_i}(X)=\Tr(X)\tau_i$, for fixed states
$\tau_i\in\cD(B_i)$ independent of $W$, then Bob's complete
final marginal, and therefore his estimate $\widehat W$, is
independent of $W$. In particular,
\eqref{eq:replacer-preservation} holds.
\end{lemma}
\begin{proof}
The joint state of the laboratories may depend on the message.
We prove only the needed invariant: Bob's complete marginal is
the same for every message. Keeping every local record permits
adaptive operations to be described by trace-preserving channels,
so the proof never conditions on a selected outcome.

Let $L_0$ and $R_0$ be the initial memories of Alice and Bob,
respectively, and write their shared state as $\eta_{L_0R_0}$.
For each message value $w$, set
$\sigma^{w,0}_{L_0R_0}=\eta_{L_0R_0}$ and
$\beta^0_{R_0}=\Tr_{L_0}\eta_{L_0R_0}$.
Before the $i$th forward use, Alice applies a channel
\begin{equation}\label{eq:feedback-encoder-domain}
\mathcal E_i^w\colon L_{i-1}\longrightarrow\widetilde L_i A_i.
\end{equation}
Her complete memory $L_{i-1}$ includes previously received
feedback and all retained local records. After receiving $B_i$,
Bob applies a channel
\begin{equation}\label{eq:feedback-channel-domain}
\mathcal B_i\colon B_iR_{i-1}\longrightarrow R_iF_i,
\end{equation}
retains $R_i$, and sends $F_i$ to Alice. The channel
$\mathcal B_i$ is independent of the unknown value $w$; its
action may be controlled by records in $R_{i-1}$.
After this transmission Alice's memory is
$L_i=\widetilde L_iF_i$. Classical feedback is included by
restricting $F_i$ to a classical register, and absent feedback
is represented by a one-dimensional $F_i$.
Several feedback transmissions between two forward uses can be
collected into $F_i$, with the local processing absorbed into
these channels. The last feedback register may be taken to be
one-dimensional.

Suppose inductively that
$\Tr_{L_{i-1}}\sigma^{w,i-1}=\beta^{i-1}_{R_{i-1}}$
for every $w$. Alice's encoded state is
\begin{equation}\label{eq:feedback-encoded-state}
\omega^{w,i}_{\widetilde L_i A_iR_{i-1}}
=(\mathcal E_i^w\otimes\id_{R_{i-1}})
  (\sigma^{w,i-1}_{L_{i-1}R_{i-1}}).
\end{equation}
Trace preservation, equivalently
$(\mathcal E_i^w)^\dagger(I_{\widetilde L_iA_i})=I_{L_{i-1}}$,
gives
\begin{equation}\label{eq:feedback-local-marginal}
\Tr_{\widetilde L_i A_i}\omega^{w,i}
=\Tr_{L_{i-1}}\sigma^{w,i-1}
=\beta^{i-1}_{R_{i-1}}.
\end{equation}
The replacer then produces
\begin{align}
\xi^{w,i}_{\widetilde L_i B_iR_{i-1}}
&=(\id_{\widetilde L_i}\otimes\mathcal R_{\tau_i}
  \otimes\id_{R_{i-1}})(\omega^{w,i}),
\label{eq:feedback-replacer-state}\\
\xi^{w,i}_{B_iR_{i-1}}
&=\tau_{B_i}\otimes\beta^{i-1}_{R_{i-1}}.
\label{eq:feedback-replacer-marginal}
\end{align}
The second identity follows because the replacer traces out
$A_i$ and tensors its fixed output with the remaining systems.
After Bob's channel and the transfer of $F_i$, define the new
joint state by
\begin{equation}\label{eq:feedback-updated-joint-state}
\sigma^{w,i}_{L_iR_i}
\coloneqq (\id_{\widetilde L_i}\otimes\mathcal B_i)
(\xi^{w,i}_{\widetilde L_iB_iR_{i-1}}),
\qquad L_i=\widetilde L_iF_i,
\end{equation}
where the registers on the right are reordered according to their
owners. Its Bob marginal is
\begin{equation}\label{eq:feedback-marginal-induction}
\Tr_{L_i}\sigma^{w,i}_{L_iR_i}=\beta^i_{R_i},
\qquad
\beta^i_{R_i}\coloneqq\Tr_{F_i}\mathcal B_i
  (\tau_{B_i}\otimes\beta^{i-1}_{R_{i-1}}).
\end{equation}
The right side is independent of $w$, so the induction is complete.

To make the treatment of measurements explicit, an instrument
$\{\mathcal I_z\}_z$ with retained classical outcome $Z$ is
represented by the channel
$X\mapsto\sum_z\mathcal I_z(X)\otimes\ket{z}\bra{z}_Z$,
where $\sum_z\mathcal I_z$ is trace preserving.
Subsequent controls act on this retained register. No outcome
is selected and no conditional state is renormalized in
\eqref{eq:feedback-local-marginal}--\eqref{eq:feedback-marginal-induction}.

Finally, let $\mathcal D\colon R_n\to\widehat W$ be Bob's
decoding channel, fixed independently of $w$. The joint output is
\begin{equation}\label{eq:feedback-final-product}
\frac1K\sum_{w=1}^K\ket{w}\bra{w}_W
\otimes\mathcal D(\beta^n_{R_n})
=\pi_W\otimes\omega_{\widehat W},
\qquad
\omega_{\widehat W}\coloneqq\mathcal D(\beta^n_{R_n}).
\end{equation}
This is \eqref{eq:replacer-preservation}.
\end{proof}

For a process without definite order, the preservation condition
must be verified for the complete protocol. The quantum switch
of two completely depolarizing channels illustrates why it does
not follow from deterministic normalization alone
\cite{EblerSalekChiribella2018}; see
Section~\ref{subsec:free-conditions}.

\subsection{Private-state generation}
For private resources, a separable comparison set supplies the
appropriate target-test bound. Let Alice and Bob have key systems
$K_A,K_B$ of dimension $K$ and shield systems $S_A,S_B$.
A private state has the form~\cite[Theorem 1]{HorodeckiEtAl2005}
\begin{equation}\label{eq:private-state}
\gamma_{K_AK_BS_AS_B}
=U(\Phi_{K_AK_B}^{K}\otimes\tau_{S_AS_B})U^\dagger,
\qquad
U=\sum_{i,j=1}^K\ket{ij}\bra{ij}_{K_AK_B}\otimes U_{ij},
\end{equation}
where $\tau_{S_AS_B}$ is a density operator and each $U_{ij}$ is
unitary on the joint shield.
The key systems carry the classical key, while the shield retains
correlations needed for its privacy. To see this, purify
$\tau_{S_AS_B}$ to $\ket{\tau}_{S_AS_BE}$, with $E$ an
external reference. A purification of $\gamma$ is
\begin{equation}\label{eq:private-purification}
\ket{\gamma}_{K_AK_BS_AS_BE}
=\frac1{\sqrt K}\sum_{i=1}^K
\ket{ii}_{K_AK_B}(U_{ii}\otimes I_E)\ket{\tau}_{S_AS_BE}.
\end{equation}
Measuring the key systems gives identical, uniformly distributed
labels. For each label $i$, the reference's conditional state
is $\tau_E$, because tracing out the shield removes its
conditional unitary. Thus the key is independent of a purifying
reference. The twisting formula specifies the target state;
the protocol is not assumed to implement a joint shield unitary.
Its privacy test is defined as the projector
\begin{equation}\label{eq:privacy-test}
\Pi_\gamma\coloneqq
U(\Phi_{K_AK_B}^{K}\otimes I_{S_AS_B})U^\dagger.
\end{equation}
It accepts $\gamma$ with probability one and every separable
state across $K_AS_A;K_BS_B$ with probability at most $1/K$
\cite[Definition 8 and Lemma 10]{WildeTomamichelBerta2017}.

The factor $1/K$ has a short explanation that also accommodates
the joint shield unitary. Consider a normalized pure product
vector $\ket{a}_{K_AS_A}\otimes\ket{b}_{K_BS_B}$ and write
$\ket{a}=\sum_{i=1}^K\ket{i}_{K_A}\ket{a_i}_{S_A}$ and
$\ket{b}=\sum_{j=1}^K\ket{j}_{K_B}\ket{b_j}_{S_B}$.
Then
\begin{align}
\bra{ab}\Pi_\gamma\ket{ab}
&=\frac1K\normket{\sum_{i=1}^K U_{ii}^\dagger
(\ket{a_i}\otimes\ket{b_i})}^{\!2}
\label{eq:privacy-overlap-proof:line-1}\\
&\leq\frac1K\left(\sum_{i=1}^K
\normket{\ket{a_i}}\normket{\ket{b_i}}\right)^2
\leq\frac1K.\label{eq:privacy-overlap-proof}
\end{align}
The first inequality follows from the triangle inequality and unitarity
of each $U_{ii}$. The second follows from Cauchy--Schwarz, because
$\sum_{i=1}^K\normket{\ket{a_i}}^2=
\sum_{i=1}^K\normket{\ket{b_i}}^2=1$.
Every separable state is a convex combination of pure product
states, so linearity proves the same acceptance bound for it.

Assume no initially shared entanglement or secret key. Require
that inserting entanglement-breaking channels in every slot
makes the final state separable across the complete laboratories.
Here a channel is entanglement breaking if its action on one
share of every bipartite state produces a separable state.
For $0\leq\varepsilon<1$, require closeness of the complete
output state to a private state, including its shield:
\begin{equation}\label{eq:private-error}
\tfrac12\norm{\rho_{K_AK_BS_AS_B}-\gamma_{K_AK_BS_AS_B}}_1
\leq\varepsilon.
\end{equation}
The variational characterization of normalized trace distance gives
$\Tr(\Pi_\gamma\rho)\geq1-\varepsilon$. Applying
Theorem~\ref{thm:comparison-test} yields
\begin{align}\label{eq:private-converse}
\log K\leq{}&
\sum_{i=1}^n\inf\{\Dh_\alpha(\cN_i\|\cM_i):
\cM_i\colon A_i\to B_i\text{ an entanglement-breaking channel}\}
\notag\\
&+\frac{\alpha}{\alpha-1}\log\frac1{1-\varepsilon},
\qquad 1<\alpha\leq2.
\end{align}

Fang and Fawzi define geometric R\'enyi Sigma-information by
optimizing the channel divergence over completely positive,
trace-nonincreasing comparisons with separable Choi operators
\cite[Definition 26]{FangFawzi2021}. They use this quantity in
a two-way-assisted private-capacity converse
\cite[Theorem 29]{FangFawzi2021}.
For $1<\alpha\leq2$, the following completion argument shows
that the local infimum in \eqref{eq:private-converse} equals
this quantity: allowing subchannels gives the same infimum
as using normalized entanglement-breaking channels.

Let $\cM\colon A\to B$ be such a subchannel and let $\tau_B$
be an arbitrary state. Its completion
\begin{equation}\label{eq:eb-completion}
J_{\widetilde{\cM}}
\coloneqq J_{\cM}+
(I_A-\Tr_BJ_{\cM})\otimes\tau_B
\end{equation}
is positive and separable across $A;B$, and has output partial
trace $I_A$. Thus $\widetilde{\cM}$ is an entanglement-breaking
channel and $\widetilde{\cM}-\cM$ is completely positive.
Here a separable positive operator is a sum of positive product
operators; separability of a Choi operator characterizes an
entanglement-breaking operation, as used in the cited definition.

Adding a positive operator to the second argument cannot increase
the geometric divergence in this range. Indeed, for a state
$\rho$ and positive $\sigma,\omega$, discarding a two-valued
flag sends $(\rho\oplus0,\sigma\oplus\omega)$ to
$(\rho,\sigma+\omega)$. Data processing gives
\begin{equation}\label{eq:second-argument-monotonicity}
\Dh_\alpha(\rho\|\sigma+\omega)
\leq\Dh_\alpha(\rho\oplus0\|\sigma\oplus\omega)
=\Dh_\alpha(\rho\|\sigma).
\end{equation}
The equality follows by evaluating the perspective block by block;
its block with zero first argument contributes zero. If the right
side is infinite the inequality is automatic, and otherwise its
support condition makes each perspective well defined.
Apply \eqref{eq:second-argument-monotonicity} to every
reference-assisted probe and take the supremum to obtain
\begin{equation}\label{eq:eb-comparison-order}
\Dh_\alpha(\cN\|\widetilde{\cM})
\leq\Dh_\alpha(\cN\|\cM).
\end{equation}
Every subchannel comparison therefore has a channel completion
with no larger divergence. Conversely, every
entanglement-breaking channel is a permitted subchannel.
The two infima are equal, without an assumption that either
is attained. This is an equality of local comparison quantities;
the global protocol condition remains preservation of separability
on every tuple of entanglement-breaking channels.

Local operations and public classical communication, starting
from a separable state, preserve separability after these
replacements. They therefore satisfy the stated condition
in an ordinary sequential protocol. The same condition is
required explicitly for a process without definite order.
The privacy test is a global mathematical measurement in this
converse; no local implementation of it is assumed.
The error criterion \eqref{eq:private-error} concerns the full
private state, including the shield. Agreement of the classical
key marginal alone would not imply privacy. Consequently,
\eqref{eq:private-converse} is stated for private-state generation;
applying it to a different secret-key error convention requires
the corresponding protocol conversion and error relation.

\section{Examples and the role of the assumptions}\label{sec:examples}
The applications depend on both the divergence inequality and its
precise hypotheses. The examples in this section make those
distinctions concrete: erasure channels give an exact evaluation,
a rank-one operator explains the reciprocal construction, and two
further examples show why comparison preservation and universal
normalization must be stated explicitly. All calculations are
analytic and require no numerical optimization.

\subsection{Erasure channels: an exactly classical divergence}
For erasure channels, the only difference between the two
hypotheses is the probability of an orthogonally flagged erasure.
The Choi formula therefore reduces to the classical divergence
between the corresponding two-outcome distributions.
Let $A\simeq B_0$ have dimension $d$, and let
$B=B_0\oplus\operatorname{span}\{\ket{e}\}$ with an erasure flag
orthogonal to $B_0$. For $0<p<1$, define
\begin{equation}\label{eq:erasure}
\mathcal E_p(X)
=(1-p)X\oplus p\,\Tr(X)\ket{e}\bra{e}.
\end{equation}
The Choi operator has two orthogonal blocks,
\begin{equation}\label{eq:erasure-choi}
J_{\mathcal E_p}
=(1-p)\ket{\Gamma}\bra{\Gamma}_{AB_0}
\oplus p\,I_A\otimes\ket{e}\bra{e}.
\end{equation}
For $0<p,q<1$ and every order in the main theorem, the perspective
acts separately on these blocks. Consequently,
\begin{equation}\label{eq:erasure-marginal}
\Tr_BP_\alpha(J_{\mathcal E_p},J_{\mathcal E_q})
=\left[(1-p)^\alpha(1-q)^{1-\alpha}
+p^\alpha q^{1-\alpha}\right]I_A.
\end{equation}
Both extremal eigenvalues coincide. The Choi formula gives
\begin{equation}\label{eq:erasure-divergence}
\Dh_\alpha(\mathcal E_p\|\mathcal E_q)
=D_\alpha\bigl((1-p,p)\|(1-q,q)\bigr).
\end{equation}
Every probe produces the same orthogonal erasure-flag mixture
of a common unerased state and a common erased state, so this
value can be attained already by reading the erasure flag.
Retaining independent flags in parallel attains the sum in
Theorem~\ref{thm:main}.

In this example there is also a direct explanation for the
absence of a general-process testing advantage. Write each
erasure channel as a mixture of two fixed channels, the identity
embedding and the erasing replacer. Multilinearity gives the
output of an arbitrary $n$-slot channel-to-state process,
including any fixed external probe and final processing, as
\begin{equation}\label{eq:erasure-mixture}
\rho_p=\sum_{s\in\{0,1\}^n}
(1-p)^{n-|s|}p^{|s|}\,\rho_s,
\end{equation}
where $|s|$ is the number of ones in $s$, namely the number of
erasing replacements, and $\rho_s$ is the normalized output obtained
by inserting the corresponding fixed channel in every slot. The states
$\rho_s$ are independent of $p$. Thus the output is obtained
by a channel from the classical independent erasure pattern.
Parallel use reveals that pattern, so every general-process
test can be simulated from a parallel experiment in this
special family. This additional conclusion uses the common
mixture representation and is not asserted for arbitrary
channel pairs.

\subsection{Why reciprocal normalization is needed}
\label{subsec:reciprocal-obstruction}
The upper-bound proof of Theorem~\ref{thm:normalization} completes
a positive operator to a scalar multiple of a channel Choi
operator. One might hope to reverse that operator inequality
for the lower bound. A two-qubit example explains why the proof
instead uses reciprocal normalization and Cauchy--Schwarz.

Let $a,b>0$ satisfy $a\neq b$, and define
\begin{equation}\label{eq:rank-one-obstruction}
\ket{v}\coloneqq\sqrt a\ket{00}+\sqrt b\ket{11},
\qquad X\coloneqq\ket{v}\bra{v},
\qquad K\coloneqq\Tr_BX=\operatorname{diag}(a,b).
\end{equation}
The smallest eigenvalue $m=\min\{a,b\}$ is positive. Suppose
that a channel Choi operator $J$ satisfied $X\geq mJ$.
Every vector orthogonal to $\ket v$ would then have zero
expectation in $J$. Positivity of $J$ would force its support
to lie in the one-dimensional span of $\ket v$, so $J=cX$
for a scalar $c\geq0$. But the channel normalization would
require $I_A=\Tr_BJ=c\operatorname{diag}(a,b)$, which is
impossible when $a\neq b$. Thus a lower completion of this
form need not exist even when $K$ is invertible. The lower
bound in \eqref{eq:normalization-bound} is a bound on positive
functionals, and the reciprocal construction proves it without
asserting this false operator inequality.

\subsection{Why comparison preservation is an extra condition}
\label{subsec:free-conditions}
The basic divergence theorem compares two hypotheses under the
same process. A capacity converse additionally needs the
comparison output to fail the target test. These are different
requirements.

For example, for $d\geq2$, let
$\mathcal D(X)=\Tr(X)I/d$ be a completely depolarizing channel
on a $d$-dimensional target. It is a replacer channel and has no
classical communication ability by itself. Nevertheless, a quantum
switch using two copies can transmit classical information when
its control is prepared in the fixed state
$\ket{+}=(\ket0+\ket1)/\sqrt2$ and Bob receives both the output
target and the output control~\cite{EblerSalekChiribella2018}.
The message is encoded only in the target; the input control is
independent of the message. If one applied
\eqref{eq:classical-converse} to an unrestricted switch merely
because each slot was a replacer, comparison with
$\cM_i=\mathcal D$ would incorrectly give a zero asymptotic communication rate
at every fixed error strictly below one. The error would be the unverified condition
\eqref{eq:replacer-preservation}: the switched comparison
experiment need not make Bob's output independent of the
message.

There is no conflict with Theorem~\ref{thm:main}. Comparing
the tuple $(\mathcal D,\mathcal D)$ with itself gives identical
output channels and divergence zero, regardless of the
communication ability of that common output channel.
This example explains why the operational model in
Section~\ref{sec:capacities} specifies the complete comparison
preservation condition rather than assigning it to local
channels alone.

\subsection{Normalization only on the hypotheses is insufficient}
\label{subsec:normalization-counterexample}
Universal normalization cannot be replaced by the requirement
that the two selected hypotheses happen to have normalized
outputs. A one-slot classical example already demonstrates
the distinction.

Choose $\delta,a,b\in(0,1)$ satisfying $a\neq b$, and consider state-preparation
channels with diagonal output distributions
\begin{equation}\label{eq:postselect-inputs}
p=(1-\delta,\delta a,\delta(1-a)),
\qquad
q=(1-\delta,\delta b,\delta(1-b)).
\end{equation}
Let $P$ project onto the last two coordinates and define the
completely positive linear transformation
$\Psi(X)=\delta^{-1}PXP$. Its outputs on these two states
are normalized, with distributions $(a,1-a)$ and $(b,1-b)$.
Set
\begin{equation}\label{eq:postselect-moment}
Q_\alpha=a^\alpha b^{1-\alpha}
+(1-a)^\alpha(1-b)^{1-\alpha}.
\end{equation}
Write
$Q_\alpha=b(a/b)^\alpha+(1-b)((1-a)/(1-b))^\alpha$.
The two ratios have weighted mean one and are different because
$a\ne b$. Strict scalar convexity for $\alpha>1$ and strict
concavity for $0<\alpha<1$ give $Q_\alpha>1$ and $Q_\alpha<1$,
respectively. The input and output moments are
$1-\delta+\delta Q_\alpha$ and $Q_\alpha$. Since $0<\delta<1$,
the output moment is strictly larger above order one and strictly
smaller below order one. Taking logarithms and dividing by
$\alpha-1$ therefore makes the output divergence strictly larger
in both ranges.

The transformation $\Psi$ is not normalized on every input
state-preparation channel. It therefore violates
\eqref{eq:process-normalization}. An identical unused second
slot can be appended to obtain a two-slot counterexample.
The same issue arises if one conditions on a favorable
outcome and renormalizes: the deterministic theorem concerns
the full normalized experiment, not a nonlinear conditional
output.

\Needspace{12\baselineskip}
\section{RLD Fisher information and estimation without definite causal order}
\label{sec:estimation}

Discrimination compares distinct channel hypotheses. Estimation instead
asks how accurately an experiment can resolve a small change in a
continuous channel parameter. Fisher information connects these tasks:
it describes the quadratic change of a divergence between nearby
models. In this section, we use Theorem~\ref{thm:main} to bound the
right logarithmic derivative (RLD) Fisher information of a channel
produced by an arbitrary parameter-independent deterministic process.
The resulting Cram\'er--Rao bound applies even when the channel slots
have no definite causal order.

The connection between geometric R\'enyi divergence and RLD Fisher
information is established in prior work. Katariya and Wilde prove
state and channel limit formulas, including regularization for singular
operators, and use RLD information to bound sequential estimation
\cite[Propositions 53 and 89, and Conclusion 41]{KatariyaWilde2021}.
Wilde gives a detailed Hessian calculation for positive definite
state families \cite[Theorem 5]{Wilde2025Fisher}.
Our contribution is the additive RLD inequality for the full process
class of Definition~\ref{def:process}, including nonidentical channel
families, and its estimation consequence. This complements the
optimization of metrological strategies with different causal
structures in \cite{LiuHuYuanYang2023} and the finite-use bounds on symmetric logarithmic derivative (SLD)
Fisher information for repeated identical channels in
\cite[Theorems 2 and 3]{WeiLiPang2026}. That work also gives
asymptotic scaling and leading-coefficient results for general
processes. Our statement concerns an additive RLD bound for
possibly different channel families at each finite number of uses.
For a fixed single-parameter output state family, SLD Fisher information
equals the largest classical Fisher information obtainable by a measurement
at the parameter point of interest~\cite{KatariyaWilde2021}.
The RLD information used here bounds this quantity from above and can
be strictly larger. Corollary~\ref{cor:rld-cramer-rao} proves the
measurement inequality directly.

\subsection{State and channel RLD Fisher information}
We first define the local quantity to be bounded. Throughout this
section, $\theta$ is a single real parameter in an open interval,
$\theta_0$ is an interior point, and differentiation is evaluated at
$\theta_0$ unless specified otherwise. We take all systems to be
finite-dimensional. For a differentiable state family $\rho_\theta$,
write
\begin{equation}\label{eq:rld-state-tangent}
\rho\coloneqq\rho_{\theta_0},\qquad
\dot\rho\coloneqq
\left.\frac{d\rho_\theta}{d\theta}\right|_{\theta=\theta_0}.
\end{equation}
For a positive semidefinite operator $X$, throughout this section
$X^{-1}$ denotes the inverse on $\supp X$, extended by zero on its
orthogonal complement, and $X^{-1/2}\coloneqq(X^{-1})^{1/2}$.
Thus $XX^{-1}=X^{-1}X=P_X$, where $P_X$ is the projector onto
$\supp X$.
The RLD Fisher information is defined as
\cite[Definition 3 and Eqs.~(5.25)--(5.30)]{KatariyaWilde2021}
\begin{equation}\label{eq:rld-state}
\IR(\theta_0;\rho_\theta)
\coloneqq
\begin{cases}
\Tr[\dot\rho\rho^{-1}\dot\rho],
 &\supp\dot\rho\subseteq\supp\rho,\\
+\infty,&\text{otherwise}.
\end{cases}
\end{equation}
The support of the Hermitian operator
$\dot\rho$ is the orthogonal complement of its kernel; thus the
finite case means that $\dot\rho$ has no matrix elements to or from
$\ker\rho$. This convention is essential for singular states.
If $\rho$ is invertible, the right logarithmic derivative is
$R=\rho^{-1}\dot\rho$, satisfying $\dot\rho=\rho R$, and
$\IR(\theta_0;\rho_\theta)=\Tr[\rho RR^\dagger]$.
The operator $R$ need not be Hermitian, whereas this information quantity
is real and nonnegative. Corollary~\ref{cor:rld-cramer-rao}
will show that RLD information bounds the Fisher information
available from a measurement. The bound need not be attained;
an infinite value makes this particular precision converse
uninformative.

For a classical family
$\rho_\theta=\sum_xp_\theta(x)\ket{x}\bra{x}$, the finite expression
in \eqref{eq:rld-state} reduces to
\begin{equation}\label{eq:rld-classical}
I_{\mathrm{cl}}(\theta_0;p_\theta)
\coloneqq
\sum_{x:p_{\theta_0}(x)>0}
\frac{\dot p_{\theta_0}(x)^2}{p_{\theta_0}(x)}.
\end{equation}
For a differentiable probability on an open interval, a zero value
at an interior point has zero first derivative because the
probability is nonnegative on both sides of that point. Such
outcomes are therefore omitted from the sum.
In particular, the Fisher information has no logarithm-base factor
in its definition. The factors of $\ln2$ below arise only because
our divergences use base-two logarithms.

For a differentiable family of channels
$\cN_\theta\colon A\to B$, the channel RLD Fisher information is
defined by optimizing over parameter-independent probes:
\begin{equation}\label{eq:rld-channel}
\IR(\theta_0;\cN_\theta)
\coloneqq
\sup\left\{
\begin{aligned}
&\IR\!\left(\theta_0;
  (\id_R\otimes\cN_\theta)(\rho_{RA})\right):\\
&\rho_{RA}\in\cD(RA),\ R\text{ finite-dimensional},
  \quad\rho_{RA}\text{ independent of }\theta
\end{aligned}
\right\}.
\end{equation}
The definition permits every mixed bipartite state and an
arbitrarily large finite reference. The state RLD information
obeys data processing under parameter-independent channels
\cite[Eq.~(5.41)]{KatariyaWilde2021}. Consequently the purification
argument in \eqref{eq:mixed-probe-from-pure} reduces
\eqref{eq:rld-channel} to pure probes with $R\simeq A$;
see also \cite[Definition 9 and Remark 10]{KatariyaWilde2021}.

For completeness, the state data-processing fact just used has a
short proof from Lemma~\ref{lem:perspective}. Let $\mathcal T$
be a parameter-independent quantum channel. If the input RLD
information is finite, there is a finite $a\geq0$ satisfying
$-a\rho\leq\dot\rho\leq a\rho$; one can take
$a=\norm{\rho^{-1/2}\dot\rho\rho^{-1/2}}$.
Set $\sigma=\mathcal T(\rho)$ and
$T=\mathcal T(\dot\rho)$. Positivity gives
$-a\sigma\leq T\leq a\sigma$, so $T$ is supported on
$\sigma$. For sufficiently small nonzero $h$, the input
operators $\rho+h\dot\rho$ and $\rho$ are states, and
the normalization follows from $\Tr\dot\rho=0$.
Lemma~\ref{lem:perspective} at order two gives
\begin{align}
\sigma+2hT+h^2T\sigma^{-1}T
&=P_2(\sigma+hT,\sigma)
\label{eq:rld-state-quadratic-contraction:line-1}\\
&\leq\mathcal T\!\left(P_2(\rho+h\dot\rho,\rho)\right)
\label{eq:rld-state-quadratic-contraction:line-2}\\
&=\sigma+2hT+
 h^2\mathcal T(\dot\rho\rho^{-1}\dot\rho).
\label{eq:rld-state-quadratic-contraction}
\end{align}
Here the support conditions justify the identity
$P_2(X,Y)=XY^{-1}X$. Cancel the constant and linear terms,
divide by $h^2$, and take the trace. Trace preservation and
parameter independence then give
\begin{equation}\label{eq:rld-state-data-processing}
\IR(\theta_0;\mathcal T(\rho_\theta))
\leq\IR(\theta_0;\rho_\theta).
\end{equation}
If the input RLD information is infinite, the inequality
holds in the extended-real sense. This also verifies the
data-processing input to the mixed-probe reduction above.

The channel quantity has a particularly useful Choi formula. Set
\begin{equation}\label{eq:rld-choi-tangent}
J\coloneqq J_{\cN_{\theta_0}},\qquad
H\coloneqq
\left.\frac{dJ_{\cN_\theta}}{d\theta}\right|_{\theta=\theta_0}.
\end{equation}
Because every $\cN_\theta$ is trace preserving,
$\Tr_BJ=I_A$ and $\Tr_BH=0$. Hayashi's channel formula
\cite[Theorem 1]{Hayashi2011}, also proved in
\cite[Proposition 29 and Appendix G]{KatariyaWilde2021}, is
\begin{equation}\label{eq:rld-choi-formula}
\IR(\theta_0;\cN_\theta)
=
\begin{cases}
\norm{\Tr_B[HJ^{-1}H]},&\supp H\subseteq\supp J,\\
+\infty,&\text{otherwise}.
\end{cases}
\end{equation}
This is a local statement: only the channel and its first derivative
at $\theta_0$ enter the formula. It does not require the nearby
Choi operators to be invertible or to have the same support.

Here is a direct derivation of \eqref{eq:rld-choi-formula}, including
the possibility of singular probes. Suppose first that
$\supp H\subseteq\supp J$, and set $Q\coloneqq\Tr_B[HJ^{-1}H]$.
There is a finite $a\geq0$ satisfying $-aJ\leq H\leq aJ$:
one can take $a=\norm{J^{-1/2}HJ^{-1/2}}$, where
$J^{-1/2}=(J^{-1})^{1/2}$, and conjugate the corresponding
operator-norm bound by $J^{1/2}$. Thus $J+hH$ is positive
semidefinite for sufficiently small $|h|$, and it is a channel
Choi operator because $\Tr_BH=0$.

Write a normalized pure probe as in
\eqref{eq:probe-parameter-main}, now allowing a singular $Z$,
and set $L\coloneqq Z\otimes I_B$. Its output state and
derivative are
\begin{equation}\label{eq:rld-probe-tangents}
\sigma\coloneqq LJL^\dagger,\qquad
T\coloneqq LHL^\dagger.
\end{equation}
Indeed, $\Tr\sigma=\Tr[(Z^\dagger Z)\Tr_BJ]=1$ and
$\Tr T=0$. The inequalities $-a\sigma\leq T\leq a\sigma$
show that $\supp T\subseteq\supp\sigma$. Applying
Lemma~\ref{lem:perspective} at order two to the completely
positive congruence $X\mapsto LXL^\dagger$ yields
\begin{align}
\sigma+2hT+h^2T\sigma^{-1}T
&=P_2(\sigma+hT,\sigma)
\label{eq:rld-probe-quadratic:line-1}\\
&\leq L P_2(J+hH,J)L^\dagger
\label{eq:rld-probe-quadratic:line-2}\\
&=\sigma+2hT+h^2LHJ^{-1}HL^\dagger.
\label{eq:rld-probe-quadratic}
\end{align}
The two expansions use $P_2(X,Y)=XY^{-1}X$ for supported
pairs. Canceling the constant and linear terms, dividing by
$h^2$ for a nonzero sufficiently small $h$, and taking the trace
give
\begin{equation}\label{eq:rld-probe-upper}
\IR\!\left(\theta_0;
  (\id\otimes\cN_\theta)(\psi_Z)\right)
=\Tr[T\sigma^{-1}T]
\leq\Tr[(Z^\dagger Z)Q]\leq\norm Q.
\end{equation}
The pure-probe reduction established after
\eqref{eq:rld-channel} therefore proves the upper bound for
every allowed probe.
The first inequality in \eqref{eq:rld-probe-upper} can be strict when
$Z$ is singular. We approach the channel supremum using invertible
filters, as follows.

If $Z$ is invertible, the perspective is covariant under the
congruence, as in \eqref{eq:choi-moment-main}, so the first
inequality in \eqref{eq:rld-probe-upper} is an equality.
Let $\ket v$ be a unit eigenvector of $Q$ for its largest
eigenvalue. For $0<\varepsilon<1$, choose
\begin{equation}\label{eq:rld-probe-approximation}
Z_\varepsilon\coloneqq
\left((1-\varepsilon)\ket v\bra v+
      \varepsilon I_A/d_A\right)^{1/2}.
\end{equation}
These invertible filters define normalized pure probes, and
their RLD informations equal
$\Tr[(Z_\varepsilon^\dagger Z_\varepsilon)Q]\to\norm Q$.
This proves the reverse bound without assuming attainment.
Finally, if $\supp H\nsubseteq\supp J$, the maximally
entangled probe gives the output state $J/d_A$ and its
derivative $H/d_A$. Their supports violate the finite case
of \eqref{eq:rld-state}, so the channel RLD information is
infinite. This completes the derivation of the Choi formula.

\subsection{The local geometric identity}
The geometric divergence becomes a quadratic measure of a small
parameter change. For a twice continuously differentiable family
of positive definite states, the state Hessian result gives
\begin{equation}\label{eq:rld-general-state-limit}
\Dh_\alpha(\rho_{\theta_0}\|\rho_{\theta_0+h})
=\frac{\alpha h^2}{2\ln2}\,
  \IR(\theta_0;\rho_\theta)+o(h^2),
\qquad \alpha\in(0,1)\cup(1,2].
\end{equation}
This is the single-parameter specialization of
\cite[Theorem 5]{Wilde2025Fisher}, with its natural logarithm
converted to our base-two convention. The related state and
channel formulas in \cite[Propositions 53 and 89]{KatariyaWilde2021}
use an additional noise-regularization limit when needed.

For our process inequality, the order $\alpha=2$ already suffices.
We now reverse the argument order, keeping the second Choi operator
fixed so that its inverse does not vary with the parameter. At this
order, the quadratic coefficient can be obtained from an exact
identity. The following proof is independent of the positive definite
state Hessian formula \eqref{eq:rld-general-state-limit}. We state and prove the identity at the channel
level so that the proof of Theorem~\ref{thm:rld-process} will not
need to move a supremum over probes through a limit.

\begin{lemma}[Order-two identity and local RLD limit]
\label{lem:rld-local-limit}
Let $h\mapsto\cC_h\colon A\to B$ be a channel family between
finite-dimensional systems, defined for $h$ in a neighborhood of
zero and differentiable at zero. Write $J_h=J_{\cC_h}$ using
\eqref{eq:choi}, and suppose that
$\supp J_h\subseteq\supp J_0$ for all sufficiently small $h$.
Set $\Delta_h=J_h-J_0$ and
$H=\left.dJ_h/dh\right|_{h=0}$.
The channel divergence \eqref{eq:channel-divergence} satisfies
\begin{equation}\label{eq:rld-d2-exact}
\Dh_2(\cC_h\|\cC_0)
=\log\!\left(1+
   \norm{\Tr_B[\Delta_hJ_0^{-1}\Delta_h]}\right).
\end{equation}
Consequently, for the channel RLD information in
\eqref{eq:rld-channel},
\begin{equation}\label{eq:rld-d2-limit}
\lim_{h\to0}\frac{\ln2}{h^2}\,
\Dh_2(\cC_h\|\cC_0)
=\norm{\Tr_B[HJ_0^{-1}H]}
=\IR(0;\cC_h).
\end{equation}
The same assertions hold for state families, with the output
partial trace replaced by the full trace and the operator norm
omitted.
\end{lemma}
\begin{proof}
The support hypothesis implies that $\Delta_h$ and $H$ are
supported on $J_0$. Let $P$ be the projector onto this support.
Then $J_0J_0^{-1}=J_0^{-1}J_0=P$, and the order-two perspective is
\begin{align}
P_2(J_h,J_0)
&=J_hJ_0^{-1}J_h
\label{eq:rld-quadratic-expansion:line-1}\\
&=J_0+2\Delta_h+\Delta_hJ_0^{-1}\Delta_h.
\label{eq:rld-quadratic-expansion}
\end{align}
Trace preservation gives $\Tr_BJ_0=I_A$ and
$\Tr_B\Delta_h=0$. Thus
\begin{equation}\label{eq:rld-quadratic-effect}
\Tr_BP_2(J_h,J_0)
=I_A+\Tr_B[\Delta_hJ_0^{-1}\Delta_h].
\end{equation}
The second term is positive semidefinite. For every positive
semidefinite $Q$, $\norm{I_A+Q}=1+\norm Q$.
Substituting \eqref{eq:rld-quadratic-effect} into the Choi
formula \eqref{eq:channel-choi-formula} therefore proves
\eqref{eq:rld-d2-exact}.

Differentiability means $\Delta_h/h\to H$ in operator norm.
All spaces are finite-dimensional and $J_0^{-1}$ is fixed, so
\begin{equation}\label{eq:rld-quadratic-convergence}
\frac{1}{h^2}\Tr_B[\Delta_hJ_0^{-1}\Delta_h]
\longrightarrow\Tr_B[HJ_0^{-1}H]
\end{equation}
in operator norm. Taking norms preserves this limit. Finally,
$\log(1+x)=x/\ln2+o(x)$ as $x\to0$ proves
\eqref{eq:rld-d2-limit}, including the case in which the limit
is zero, and \eqref{eq:rld-choi-formula} identifies its value.
For states, the same expansion uses $\Tr\Delta_h=0$ and
$\Tr\rho_0=1$, and \eqref{eq:rld-state} identifies the limit.
\end{proof}

The support hypothesis in Lemma~\ref{lem:rld-local-limit} cannot
simply be dropped. For example,
$\rho_h=\operatorname{diag}(1-h^2,h^2)$ has
$\dot\rho_0=0$ and therefore zero RLD information at zero, but
$\Dh_2(\rho_h\|\rho_0)=+\infty$ whenever $h\ne0$ is sufficiently
small. This is why the next proof uses auxiliary channel families
with the same value and first derivative as the given families,
instead of assuming a divergence limit for the original families.

\subsection{The RLD process inequality}
RLD information depends only on a local tangent. Whenever this
information is finite, the channel's tangent can be followed along
a short affine path of channels. That path has fixed support and
therefore permits the exact calculation in
Lemma~\ref{lem:rld-local-limit}. Applying the geometric process
inequality to these paths yields the desired bound, without any
assumption about the causal order of the slots.
Figure~\ref{fig:rld-estimation} summarizes this passage from a
finite-divergence bound to an estimation bound.

\begin{theorem}[RLD Fisher information under deterministic processes]
\label{thm:rld-process}
Let $n\geq1$ be an integer, and let
$\cN_{i,\theta}\colon A_i\to B_i$, $i\in\{1,\ldots,n\}$,
be finite-dimensional channel families on a real open parameter
interval, differentiable at an interior point $\theta_0$.
Let $\cW$ be a deterministic completely positive process with
finite-dimensional external interface $C\to D$ and Choi action
$\Phi$ satisfying
\eqref{eq:process-domain}--\eqref{eq:process-normalization}.
Assume that $\Phi$ is independent of $\theta$, and define
\begin{equation}\label{eq:rld-output-family}
\cK_\theta\coloneqq
\cW(\cN_{1,\theta},\ldots,\cN_{n,\theta}).
\end{equation}
Then the channel RLD Fisher information in
\eqref{eq:rld-channel} obeys
\begin{equation}\label{eq:rld-process}
\IR(\theta_0;\cK_\theta)
\leq\sum_{i=1}^n\IR(\theta_0;\cN_{i,\theta}).
\end{equation}
The inequality holds in the extended-real sense. It requires
neither invertible Choi operators nor constant supports for the
given channel families near $\theta_0$.
\end{theorem}

\paragraph{Proof strategy.}
If a local RLD information is infinite, the bound is immediate.
Otherwise each Choi derivative is supported on the corresponding
Choi operator. We replace each family locally by its affine tangent
path, check that these paths consist of channels, and show that
their processed output has fixed support. Theorem~\ref{thm:main}
at order two then supplies an inequality of divergences, and
Lemma~\ref{lem:rld-local-limit} extracts its quadratic coefficient.
The replacement preserves all derivatives that define the RLD
information.

\begin{proof}
Assume that every term on the right of \eqref{eq:rld-process} is
finite. Set
\begin{equation}\label{eq:rld-local-choi-data}
J_i\coloneqq J_{\cN_{i,\theta_0}},\qquad
H_i\coloneqq
\left.\frac{dJ_{\cN_{i,\theta}}}{d\theta}\right|_{\theta=\theta_0},
\qquad
a_i\coloneqq\norm{J_i^{-1/2}H_iJ_i^{-1/2}}.
\end{equation}
Here $J_i^{-1/2}=(J_i^{-1})^{1/2}$.
Equation~\eqref{eq:rld-choi-formula} gives
$\supp H_i\subseteq\supp J_i$.
Moreover, $\Tr_{B_i}H_i=0$.
On the support of $J_i$, conjugating the inequality
$-a_iI\leq J_i^{-1/2}H_iJ_i^{-1/2}\leq a_iI$
by $J_i^{1/2}$ gives
\begin{equation}\label{eq:rld-affine-order}
(1-|h|a_i)J_i\leq J_i+hH_i\leq(1+|h|a_i)J_i.
\end{equation}
Choose $|h|$ small enough that $|h|a_i<1$ for every $i$;
an index with $a_i=0$ imposes no restriction.
The operators
\begin{equation}\label{eq:rld-affine-channels}
J_{\cL_{i,h}}\coloneqq J_i+hH_i
\end{equation}
are positive semidefinite by \eqref{eq:rld-affine-order}, and
their output partial traces are $I_{A_i}$.
They therefore define channels $\cL_{i,h}\colon A_i\to B_i$.
Their Choi operators have the same support as $J_i$.

Let
\begin{equation}\label{eq:rld-affine-output}
Y(h)\coloneqq
\Phi\!\left(\bigotimes_{i=1}^n(J_i+hH_i)\right),\qquad
Y_0\coloneqq\Phi\!\left(\bigotimes_{i=1}^nJ_i\right).
\end{equation}
Determinism makes $Y(h)$ the Choi operator of the output channel
$\cW(\cL_{1,h},\ldots,\cL_{n,h})$ for every such $h$.
Taking tensor products in \eqref{eq:rld-affine-order} and
applying the positive transformation $\Phi$ gives
\begin{equation}\label{eq:rld-output-support}
\left[\prod_{i=1}^n(1-|h|a_i)\right]Y_0
\leq Y(h)\leq
\left[\prod_{i=1}^n(1+|h|a_i)\right]Y_0.
\end{equation}
Both scalar factors are strictly positive. These two inequalities
imply $\supp Y(h)=\supp Y_0$, even if $Y_0$ is singular.
Thus the output family satisfies the support hypothesis of
Lemma~\ref{lem:rld-local-limit}.

Linearity of $\Phi$ and the tensor product rule give
\begin{align}
\left.\frac{dY(h)}{dh}\right|_{h=0}
&=
\left.\frac{dJ_{\cK_\theta}}{d\theta}\right|_{\theta=\theta_0}
\label{eq:rld-output-tangent:line-1}\\
&
=\Phi\!\left(
\sum_{i=1}^n J_1\otimes\cdots\otimes J_{i-1}
  \otimes H_i\otimes J_{i+1}\otimes\cdots\otimes J_n
\right).
\label{eq:rld-output-tangent}
\end{align}
Also $Y_0=J_{\cK_{\theta_0}}$. The actual output family and
the auxiliary output family therefore have identical channel
RLD information at their respective base points, by
\eqref{eq:rld-choi-formula}.

Apply Theorem~\ref{thm:main} to the channel pairs
$(\cL_{i,h},\cL_{i,0})$ at order two:
\begin{align}
\Dh_2\!\left(
 \cW(\cL_{1,h},\ldots,\cL_{n,h})
 \middle\|\cW(\cL_{1,0},\ldots,\cL_{n,0})\right)
\leq\sum_{i=1}^n
       \Dh_2(\cL_{i,h}\|\cL_{i,0}).
\label{eq:rld-affine-divergence-bound}
\end{align}
Multiply both sides by $\ln2/h^2$ and let $h\to0$.
Lemma~\ref{lem:rld-local-limit} applies to every term. On the
left it gives $\IR(\theta_0;\cK_\theta)$ by
\eqref{eq:rld-output-tangent:line-1}--\eqref{eq:rld-output-tangent}; on the right it gives
$\IR(\theta_0;\cN_{i,\theta})$ because
$\cL_{i,h}$ has Choi value $J_i$ and derivative $H_i$ at zero.
The sum has finitely many terms, so the limit yields
\eqref{eq:rld-process}.
\end{proof}

The role of parameter independence is concrete. If $\Phi$ itself
depended on the unknown parameter, differentiating the output
would add a derivative of $\Phi$ to
\eqref{eq:rld-output-tangent:line-1}--\eqref{eq:rld-output-tangent}. That extra source of information
is not counted by the right-hand side of
\eqref{eq:rld-process}. The theorem does allow an experiment
chosen for a specified nominal point $\theta_0$, provided that
the experiment remains fixed when $\theta$ is varied locally.
Adaptive controls based on recorded outcomes are also permitted
when the complete experiment is represented by one such fixed
process.

\begin{figure}[!t]
\centering
\begin{tikzpicture}[font=\small,>=Latex,
 channel/.style={draw=blue!60!black,fill=blue!6,rounded corners=2pt,
   minimum width=1.65cm,minimum height=.62cm,
   inner xsep=8pt,inner ysep=6pt},
 wire/.style={->,line width=.65pt},
 proofbox/.style={draw=blue!55!black,fill=blue!4,rounded corners=3pt,
   text width=4.25cm,minimum height=2.05cm,align=center,
   inner xsep=8pt,inner ysep=10pt}]
\node[align=center,inner sep=0pt] at (0,4.65)
 {$J_{\mathcal L_{i,h}}\coloneqq J_i+hH_i$,\qquad
  $J_i\coloneqq J_{\mathcal N_{i,\theta_0}}$,\qquad
  $H_i\coloneqq\left.\partial_\theta J_{\mathcal N_{i,\theta}}\right|_{\theta=\theta_0}$};
\path[fill=blue!4,even odd rule,rounded corners=4pt]
 (-7.8,-.72) rectangle (-.5,4.03)
 (-5.35,2.23) rectangle (-2.95,3.17)
 (-5.35,1.13) rectangle (-2.95,2.07)
 (-5.35,.03) rectangle (-2.95,.97);
\draw[blue!40!black,line width=.65pt,rounded corners=4pt]
 (-7.8,-.72) rectangle (-.5,4.03);
\foreach \y in {2.7,1.6,.5}{
 \draw[black!50,densely dashed,line width=.55pt,rounded corners=3pt]
  (-5.35,\y-.47) rectangle (-2.95,\y+.47);
}
\node[inner sep=0pt] at (-4.15,3.65)
 {fixed process $\Phi$};
\foreach \i/\y in {1/2.7,2/1.6,3/.5}{
 \node[channel] (n\i) at (-4.15,\y) {$\mathcal L_{\i,h}$};
 \draw[wire] (-7.1,\y)--node[above=3pt,inner sep=2pt]{$A_{\i}$}(n\i.west);
 \draw[wire] (n\i.east)--node[above=3pt,inner sep=2pt]{$B_{\i}$}(-1.2,\y);
}
\node[draw=blue!55!black,fill=blue!4,rounded corners=3pt,
 text width=2.1cm,align=center,minimum height=1.05cm,
 inner xsep=8pt,inner ysep=9pt] (measurement) at (3.25,1.6)
 {fixed POVM\\[.35em]$\{M_x\}$};
\draw[wire] (-.5,1.6)--node[above=5pt,inner sep=3pt]{$\omega_h$}(measurement.west);
\node[inner sep=3pt] (probability) at (7.03,1.6) {$x\sim p_h$};
\draw[wire,double,double distance=1.2pt] (measurement.east)--(probability.west);
\node[font=\footnotesize,align=center,inner sep=2pt] at (-4.15,-.48)
 {slot order unspecified};
\node[font=\footnotesize,align=center,text width=14.7cm,inner sep=4pt] at (0,-1.48)
 {$\displaystyle\omega_h=\Phi\!\left(\bigotimes_{i=1}^{n}(J_i+hH_i)\right)$\\[.4em]
 Its value and first derivative at $h=0$ agree with those of the true output family.};
\node[proofbox] (finite) at (-5.45,-4.52)
 {Geometric divergence\\[.5em]
  process inequality\\[.5em]
  Theorem~\ref{thm:main}};
\node[proofbox] (local) at (0,-4.52)
 {RLD Fisher information\\[.5em]
  additive bound for $\IR$\\[.5em]
  Theorem~\ref{thm:rld-process}};
\node[proofbox] (error) at (5.45,-4.52)
 {Locally unbiased\\[.5em]
  variance lower bound\\[.5em]
  Corollary~\ref{cor:rld-cramer-rao}};
\draw[wire] (finite.east)--(local.west);
\draw[wire] (local.east)--(error.west);
\node[font=\footnotesize,align=center,inner sep=3pt] at (-2.72,-2.86)
 {multiply by $(\ln2)/h^2$\\[.35em]then take $h\to0$};
\node[font=\footnotesize,align=center,inner sep=3pt] at (2.72,-2.86)
 {measure, then apply\\[.35em]Cauchy--Schwarz};
\end{tikzpicture}
\caption{From the process divergence inequality to an estimation bound.
Three slots of a channel-to-state process are shown.
The proof replaces each channel family by the affine Choi path
$J_i+hH_i$, which agrees with its value and first derivative at
$\theta_0$. When the local RLD informations are finite, these are channels
with fixed support for sufficiently small $h$, as proved in
Theorem~\ref{thm:rld-process}.
The parameter $h$ is a displacement, not physical time; the
left-to-right slot wires do not specify an order between slots.
Solid wires carry quantum systems; the doubled wire carries the
classical measurement outcome, whose probability is
$p_h(x)=\Tr[M_x\omega_h]$ for the auxiliary affine family.
The process and measurement are independent of $h$.
Local unbiasedness is defined in
\eqref{eq:rld-local-unbiasedness}.
The lower arrows represent mathematical implications.
Lemma~\ref{lem:rld-local-limit} identifies the divergence limit
with RLD Fisher information. Since the original and affine output
families have the same value and derivative, the resulting Fisher
information and locally unbiased estimation bounds apply to the
original families. No divergence limit for the original families
at a point of changing support is presumed.}
\label{fig:rld-estimation}
\end{figure}
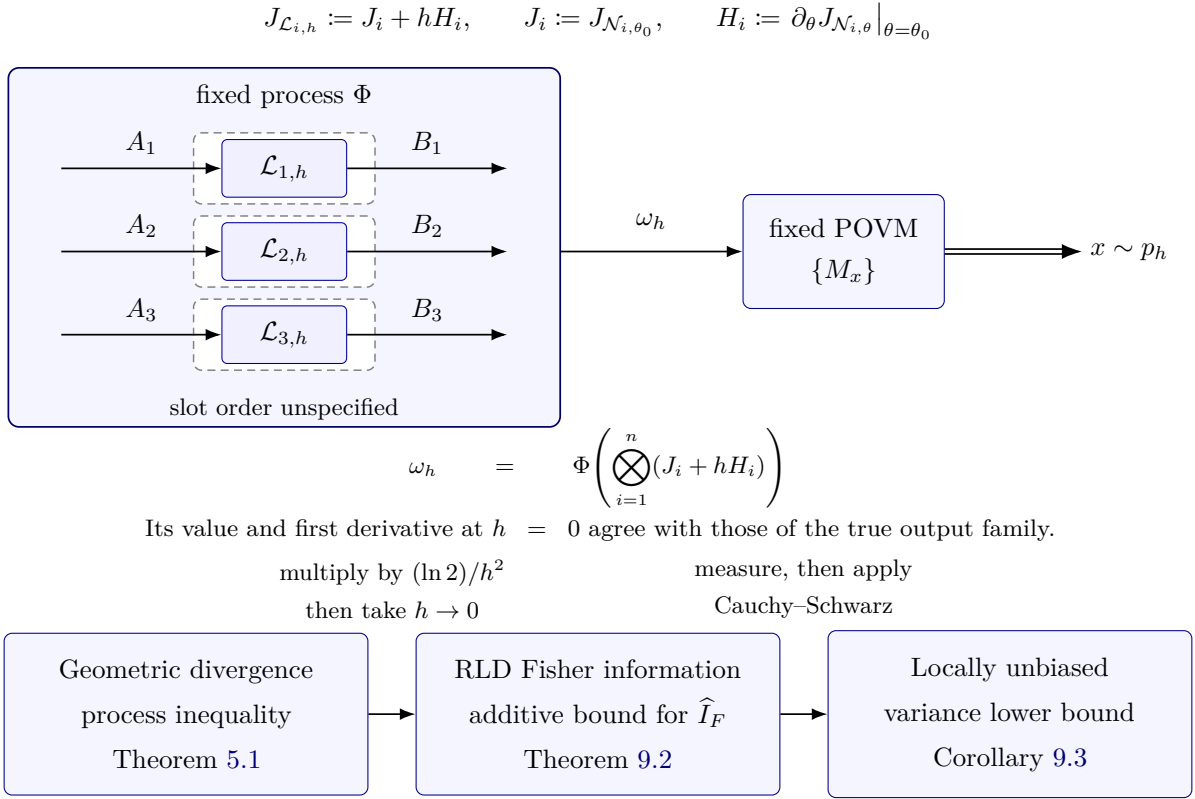

\subsection{A Cram\'er--Rao bound for general estimation processes}
We now connect the RLD inequality to a measurable estimation
error. A complete estimation protocol prepares its input,
uses the parameter-dependent channels, and produces a final
quantum state $\omega_\theta$. Its preparation, retained
reference systems, and parameter-independent controls can all
be included in a channel-to-state process. A final measurement
then produces a classical distribution. We keep all outcomes;
conditioning on a selected successful outcome would not, by
itself, define a deterministic experiment.

Let $\{M_x\}_{x\in\mathcal X}$ be a finite-outcome positive
operator-valued measurement (POVM), meaning $M_x\geq0$ and
$\sum_xM_x=I$. It is independent of $\theta$, and its
probabilities are
\begin{equation}\label{eq:rld-outcome-probabilities}
p_\theta(x)\coloneqq\Tr[M_x\omega_\theta].
\end{equation}
An estimator is a real number $\widehat\theta(x)$ assigned to
each outcome, independently of the unknown parameter. It is
locally unbiased at $\theta_0$ if
\begin{equation}\label{eq:rld-local-unbiasedness}
\sum_xp_{\theta_0}(x)\widehat\theta(x)=\theta_0,
\qquad
\left.\frac{d}{d\theta}
 \sum_xp_\theta(x)\widehat\theta(x)\right|_{\theta=\theta_0}=1.
\end{equation}
Its mean-square error at that point is
\begin{equation}\label{eq:rld-mse}
\operatorname{MSE}_{\theta_0}(\widehat\theta)
\coloneqq
\sum_xp_{\theta_0}(x)(\widehat\theta(x)-\theta_0)^2.
\end{equation}
The first condition in \eqref{eq:rld-local-unbiasedness} makes
this equal to the estimator's variance at $\theta_0$.

\begin{corollary}[Cram\'er--Rao bound without a definite causal order]
\label{cor:rld-cramer-rao}
Let $n\geq1$ be an integer, and let
$\cN_{i,\theta}\colon A_i\to B_i$, $i\in\{1,\ldots,n\}$,
be finite-dimensional channel families differentiable at an
interior point $\theta_0$ of a real open interval. Let $\Theta$
be a parameter-independent deterministic channel-to-state
process satisfying Definition~\ref{def:process}, and set
\begin{equation}\label{eq:rld-estimation-output}
\omega_\theta\coloneqq
\Theta\!\left(\bigotimes_{i=1}^nJ_{\cN_{i,\theta}}\right).
\end{equation}
For every parameter-independent finite-outcome POVM, its
distribution \eqref{eq:rld-outcome-probabilities} obeys
\begin{equation}\label{eq:rld-measurement-bound}
I_{\mathrm{cl}}(\theta_0;p_\theta)
\leq\IR(\theta_0;\omega_\theta)
\leq\sum_{i=1}^n\IR(\theta_0;\cN_{i,\theta}),
\end{equation}
where the quantities are defined in
\eqref{eq:rld-classical}, \eqref{eq:rld-state}, and
\eqref{eq:rld-channel}. If the sum on the right is finite and
strictly positive, every estimator locally unbiased in the
sense of \eqref{eq:rld-local-unbiasedness} satisfies
\begin{equation}\label{eq:rld-cramer-rao}
\operatorname{MSE}_{\theta_0}(\widehat\theta)
\geq\frac{1}{I_{\mathrm{cl}}(\theta_0;p_\theta)}
\geq\frac{1}{\displaystyle
       \sum_{i=1}^n\IR(\theta_0;\cN_{i,\theta})}.
\end{equation}
If $I_{\mathrm{cl}}(\theta_0;p_\theta)=0$, no locally unbiased
estimator exists for that POVM, regardless of the channel-RLD
sum. In particular, a zero channel-RLD sum rules out every
such estimator. If that sum is infinite, its reciprocal is
zero, with the convention $1/(+\infty)=0$. The classical
inequality $\operatorname{MSE}_{\theta_0}(\widehat\theta)
\geq1/I_{\mathrm{cl}}(\theta_0;p_\theta)$ still holds whenever
a locally unbiased estimator exists.
\end{corollary}

\begin{proof}
The second inequality of \eqref{eq:rld-measurement-bound}
follows from Theorem~\ref{thm:rld-process} with a
one-dimensional external input. For completeness, we prove
the measurement inequality and the classical Cram\'er--Rao
step directly.

The measurement inequality is immediate if
$\IR(\theta_0;\omega_\theta)=+\infty$. Otherwise write
$\omega=\omega_{\theta_0}$ and
$\dot\omega=\left.d\omega_\theta/d\theta\right|_{\theta_0}$.
The support condition in \eqref{eq:rld-state} holds. For each
outcome, apply Hilbert--Schmidt Cauchy--Schwarz to
$A_x=\omega^{1/2}M_x^{1/2}$ and
$B_x=\omega^{-1/2}\dot\omega M_x^{1/2}$. Since
$\omega^{1/2}\omega^{-1/2}\dot\omega=\dot\omega$, this gives
\begin{align}
\dot p_{\theta_0}(x)^2
&=\left|\Tr[A_x^\dagger B_x]\right|^2
\label{eq:rld-measurement-cs:line-1}\\
&\leq\Tr[A_x^\dagger A_x]\Tr[B_x^\dagger B_x]
\label{eq:rld-measurement-cs:line-2}\\
&=p_{\theta_0}(x)
  \Tr[M_x\dot\omega\omega^{-1}\dot\omega].
\label{eq:rld-measurement-cs}
\end{align}
An outcome with zero probability therefore has zero derivative.
Dividing by each positive probability, summing, and using
$\sum_xM_x=I$ proves the first inequality in
\eqref{eq:rld-measurement-bound}.

The following classical argument also applies if the output RLD
information is infinite: at an interior point, a zero probability
has zero first derivative by nonnegativity on both sides, as
explained after \eqref{eq:rld-classical}.
Set $g_x=\widehat\theta(x)-\theta_0$ and define the score
$s_x=\dot p_{\theta_0}(x)/p_{\theta_0}(x)$ for positive
probabilities. Normalization implies $\sum_x\dot
p_{\theta_0}(x)=0$. Local unbiasedness gives
\begin{equation}\label{eq:rld-score-identity}
1=\sum_xg_x\dot p_{\theta_0}(x)
 =\sum_{x:p_{\theta_0}(x)>0}
    p_{\theta_0}(x)g_xs_x.
\end{equation}
Cauchy--Schwarz for this weighted sum now yields
\begin{equation}\label{eq:rld-classical-cr-proof}
1\leq
\left(\sum_xp_{\theta_0}(x)g_x^2\right)
\left(\sum_{x:p_{\theta_0}(x)>0}
 p_{\theta_0}(x)s_x^2\right)
=\operatorname{MSE}_{\theta_0}(\widehat\theta)\,
 I_{\mathrm{cl}}(\theta_0;p_\theta).
\end{equation}
Together with \eqref{eq:rld-measurement-bound}, this proves
\eqref{eq:rld-cramer-rao}. If the total local RLD information
is zero, \eqref{eq:rld-measurement-bound} forces every score
to vanish, contradicting \eqref{eq:rld-score-identity}.
A zero measured Fisher information contradicts
\eqref{eq:rld-score-identity} for that POVM. If the channel-RLD
sum is infinite, only its reciprocal becomes trivial; the
classical measurement-dependent lower bound remains valid.
\end{proof}

The same proof covers regular estimators for a POVM $M$ on a
general outcome space. In finite dimension, set
$\mu(E)\coloneqq\Tr M(E)$. This is a finite measure. Positivity
implies that $\mu(E)=0$ forces $M(E)=0$, so every matrix-entry
measure is absolutely continuous with respect to $\mu$.
The scalar Radon--Nikodym theorem, applied to the finitely many
entries, gives an operator-valued density
$M(dx)=m(x)\mu(dx)$. Testing its quadratic forms on a countable
dense set of vectors shows that $m(x)\geq0$ almost everywhere.
The trace and normalization give
\begin{equation}\label{eq:rld-povm-density}
\Tr m(x)=1\quad\mu\text{-almost everywhere},\qquad
\int m(x)\,\mu(dx)=I.
\end{equation}
Consequently $p_\theta(x)=\Tr[m(x)\omega_\theta]$ is a
normalized probability density, and
$\dot p_{\theta_0}(x)=\Tr[m(x)\dot\omega]$ is its derivative. Local
unbiasedness means \eqref{eq:rld-local-unbiasedness} with sums
replaced by integrals against $\mu$. Assume that the estimator
has finite second moment and that the derivative of its mean
can be passed under the integral. Applying
\eqref{eq:rld-measurement-cs:line-1}--\eqref{eq:rld-measurement-cs} pointwise with $m(x)$ in place
of $M_x$, then integrating and using
\eqref{eq:rld-povm-density}, proves the measurement bound.
The integral form of Cauchy--Schwarz gives
\eqref{eq:rld-classical-cr-proof} under the stated assumptions.
Thus the same error bound holds for these continuous-outcome
protocols as well.

For $n$ uses of one channel family,
\eqref{eq:rld-cramer-rao} becomes
\begin{equation}\label{eq:rld-identical-cr}
\operatorname{MSE}_{\theta_0}(\widehat\theta)
\geq\frac{1}{n\IR(\theta_0;\cN_\theta)},
\qquad 0<\IR(\theta_0;\cN_\theta)<+\infty.
\end{equation}
In particular, the root-mean-square error cannot decrease faster
than order $n^{-1/2}$ within this locally unbiased setting when
the single-use RLD information is fixed and finite.
This is a sufficient condition for such a scaling bound, not
a statement that the RLD constant is always attainable.
For a unitary-channel family with a nonzero Choi derivative at
$\theta_0$, the RLD information is infinite. Indeed, the Choi
operator has rank one. A Hermitian derivative supported on it
would be proportional to that Choi operator, and the zero
output partial trace of the derivative forces the proportionality
constant to vanish. Thus a nonzero derivative violates the
finite-RLD support condition. In this case the reciprocal RLD
bound is zero. The finite-RLD hypothesis cannot be removed
from the scaling conclusion.

\subsection{A tight example: estimating an erasure probability}
The erasure family of Section~\ref{sec:examples} provides a
transparent example in which the estimation bound is attained.
Let $A$ have dimension $d$, let $B=A\oplus\operatorname{span}
\{\ket e\}$, and let the unknown parameter be the erasure
probability $p\in(0,1)$. The channel is
\begin{equation}\label{eq:rld-erasure-channel}
\mathcal E_p(X)
\coloneqq(1-p)X\oplus p\Tr(X)\ket e\bra e.
\end{equation}
For an arbitrary, possibly mixed, parameter-independent probe $\psi_{RA}$,
its output and derivative have orthogonal blocks:
\begin{align}
\omega_p
&=(1-p)\psi_{RA}\oplus p\psi_R\otimes\ket e\bra e,
\label{eq:rld-erasure-output}\\
\frac{d\omega_p}{dp}
&=-\psi_{RA}\oplus\psi_R\otimes\ket e\bra e.
\label{eq:rld-erasure-derivative}
\end{align}
Here the first block is embedded into the non-erased output
subspace. Evaluating the inverse separately on the two supports
gives, including for singular probes,
\begin{align}
\IR(p;\omega_p)
&=\frac{\Tr[\psi_{RA}\psi_{RA}^{-1}\psi_{RA}]}{1-p}
  +\frac{\Tr[\psi_R\psi_R^{-1}\psi_R]}{p}
\label{eq:rld-erasure-information:line-1}\\
&=\frac{1}{1-p}+\frac{1}{p}
 =\frac{1}{p(1-p)}.
\label{eq:rld-erasure-information}
\end{align}
Every probe gives the same value, so
$\IR(p;\mathcal E_p)=1/[p(1-p)]$ by
\eqref{eq:rld-channel}. Corollary~\ref{cor:rld-cramer-rao}
therefore implies
\begin{equation}\label{eq:rld-erasure-error}
\operatorname{MSE}_{p}(\widehat p)\geq\frac{p(1-p)}{n}
\end{equation}
for every locally unbiased estimator based on an allowed
$n$-slot process.

To attain this bound, use the channels independently in parallel
and measure whether each output lies in the erasure subspace.
The erasure indicators $X_1,\ldots,X_n$ are independent
Bernoulli random variables satisfying
$\mathbb E_p[X_i]=p$ and
$\operatorname{Var}_p(X_i)=p(1-p)$. The empirical frequency
$\widehat p=n^{-1}\sum_{i=1}^nX_i$ is unbiased for every
$p\in(0,1)$, and independence gives
\begin{equation}\label{eq:rld-erasure-attainment}
\operatorname{MSE}_p(\widehat p)
=\operatorname{Var}_p(\widehat p)
=\frac{1}{n^2}\sum_{i=1}^n\operatorname{Var}_p(X_i)
=\frac{p(1-p)}{n}.
\end{equation}
Thus, for estimating this noise parameter, a parallel protocol
already attains the lower bound valid for every deterministic
process in Definition~\ref{def:process}. The example also
illustrates why singular Choi operators should be included:
the erasure-channel Choi operator is generally singular, yet
its RLD information is finite at every $p\in(0,1)$.

\section{Conclusion}\label{sec:conclusion}

\subsection{Summary}
We proved a geometric R\'enyi divergence inequality for deterministic
completely positive processes with independently replaceable channel
slots. The bound holds for an arbitrary number of slots and for
$\alpha\in(0,1)\cup(1,2]$, without an assumed causal order.
The optimized output geometric divergence is exactly the sum
of the individual channel divergences, with parallel probes
approaching the supremum.

The principal technical ingredient is the normalization bound
for positive local Choi operators. Completion to a channel
controls the largest input-effect eigenvalues. Reciprocal
normalization, the complex affine span of channels, and
Cauchy--Schwarz control the smallest eigenvalues. The latter
step supplies the missing direction needed below order one.
The proof remains within finite-dimensional operator algebra
and handles singular Choi operators through the appropriate
perspective extensions.

Applications include finite-use discrimination bounds, a
strong-converse estimate, symmetric-error bounds, and a
Belavkin--Staszewski channel-radius bound for exclusion.
For communication and resource generation, the same inequality
combines with a comparison test. The resulting capacity bounds
apply to process classes that preserve the specified comparison
outputs. This preservation requirement is stated separately
from deterministic normalization.

We also established an additive RLD Fisher-information inequality
for differentiable channel families under every fixed process in
our class. The proof passes from the geometric inequality at
order two to a local quadratic coefficient. Auxiliary affine
channel paths preserve the relevant derivatives while making
the support conditions explicit, so the theorem includes
singular Choi operators and does not assume constant support
for the original families. The resulting Cram\'er--Rao bound
constrains locally unbiased estimation without a definite causal
order. Counting erasure flags attains the bound for the erasure
probability, while general families may have a gap between
RLD information and achievable measurement information.

\subsection{Future directions}
Theorem~\ref{thm:normalization} may be useful for other channel quantities
expressed through positive operator perspectives. Identifying
which perspectives yield an additive finite-use divergence
bound requires both their behavior under completely positive
transformations and a suitable channel formula. The present
argument identifies these two requirements explicitly.

For estimation, an important question is when the RLD converse
captures the optimal measurement precision and when a different
quantum Fisher information gives a smaller upper bound on
measurement information. The finite-use RLD inequality can be
compared with the process optimizations and bounds in
\cite{LiuHuYuanYang2023,WeiLiPang2026}. Extending the present
scalar-parameter analysis to multiparameter estimation would
also require accounting for the simultaneous attainability of
measurements; a scalar divergence Hessian alone does not
resolve that issue.

The capacity applications motivate a structural question:
which processes without definite order preserve a given
comparison class? The answer depends on the resource and
the physical ownership of the input, output, control, and
memory systems. A characterization for the Rains comparisons
in \eqref{eq:rains-preservation}, or for the no-communication
condition in \eqref{eq:replacer-preservation}, would turn the
conditional operational classes considered here into more
explicit families of protocols.

Quantum reading offers another application to investigate. Here
the stored message selects a sequence of channels from a memory
cell, and the reader's operations are fixed independently of
that message. Singh, Huang, and Wilde derive an adaptive
energy-constrained reading converse and an unconstrained
Belavkin--Staszewski channel-radius bound
\cite[Theorem 1 and Corollary 1]{SinghHuangWilde2026}.
Comparing every stored codeword with one message-independent
tuple of channels suggests a connection with the process
inequality proved here. The useful next question is how energy
constraints and restrictions on the reader's physical operations
should enter this comparison when a causal order is not fixed.

An infinite-dimensional extension would require additional
normalization arguments. For sequential discrimination, Huang,
Lami, Singh, and Wilde proved a chain rule for the
Belavkin--Staszewski divergence under an energy constraint
\cite[Appendix B.2, Lemma 2]{HuangLamiSinghWilde2026}.
This result supplies a concrete comparison for an
energy-constrained extension of our theorem. In the present
normalization proof, however, inverse input effects may become
unbounded, and the finite-dimensional affine-interior argument
does not apply verbatim. A truncation approach would need
approximation errors controlled uniformly over the allowed
processes. The cited sequential chain rule alone does not
provide these additional estimates.

\section*{Acknowledgements}

ZH is supported by an ARC DECRA Fellowship (DE230100144) ``Quantum-enabled super-resolution
imaging'' and an RMIT Vice Chancellor's Senior Research Fellowship.
MMW acknowledges support from the Cornell University School of Electrical and Computer Engineering.

\section*{Statement of AI assistance}
ChatGPT assisted with proof formulation and exploration, exposition, literature checks, figure preparation, and numerical tests during the development of this manuscript. The checks performed during this assistance were internal reviews, not independent external peer review or formal verification. Responsibility for any submitted version and its claims rests with the authors.

\begingroup
\raggedright
\interlinepenalty=10000
\renewcommand{\bibsection}{%
  \section*{\refname}%
  \addcontentsline{toc}{section}{\refname}%
}
\bibliographystyle{unsrtnat}
\bibliography{references}
\endgroup
\appendix
\section{The Choi formula and approximating probes}\label{app:choi}
This appendix explains the established formula
\eqref{eq:channel-choi-formula} in a way that also checks the
optimization convention used in Corollary~\ref{cor:collapse}.
The argument is due to the channel-formula proofs
in~\cite{FangFawzi2021,KatariyaWilde2021}.
Fix a channel $\cN\colon A\to B$, a completely positive
transformation $\cM\colon A\to B$, and
$\alpha\in(0,1)\cup(1,2]$, with finite-dimensional $A,B$.
When $\alpha>1$, first assume $J_{\cN}\ll J_{\cM}$; the
unsupported case is treated after \eqref{eq:probe-approximation}.
We use the pure-probe reduction \eqref{eq:channel-pure-reduction}.

An arbitrary pure probe with reference dimension $d_A$ can be
written
\begin{equation}\label{eq:probe-filter}
\ket{\psi}_{AA'}=(Z_A\otimes I_{A'})\ket{\Gamma}_{AA'},
\qquad \Tr(Z^\dagger Z)=1.
\end{equation}
The two output operators are
$(Z\otimes I_B)J_{\cN}(Z^\dagger\otimes I_B)$ and the
corresponding expression for $\cM$.
If $Z$ is invertible, the perspective is exactly covariant
under this congruence for the pairs permitted in
Lemma~\ref{lem:perspective}, in particular with $X\ll Y$
when $1<\alpha\leq2$:
\begin{equation}\label{eq:invertible-congruence}
P_\alpha(LXL^\dagger,LYL^\dagger)
=LP_\alpha(X,Y)L^\dagger,\qquad L=Z\otimes I_B.
\end{equation}
Indeed, apply Lemma~\ref{lem:perspective} to the completely
positive congruence $\Psi_L(Z)=LZL^\dagger$ for one direction,
and apply it to the inverse congruence for the opposite direction.
The singular mean extension below one is compatible with
this identity.

Let $K=\Tr_BP_\alpha(J_{\cN},J_{\cM})$.
For the probe in \eqref{eq:probe-filter} with invertible $Z$,
\begin{equation}\label{eq:probe-moment}
\Qh_\alpha(\psi)
\coloneqq
\Tr P_\alpha\!\left(
(\id\otimes\cN)(\psi),(\id\otimes\cM)(\psi)\right)
=\Tr(Z^\dagger ZK).
\end{equation}
For a singular $Z$, the same lemma applied to $\Psi_L$ implies
$\Qh_\alpha(\psi)\leq\Tr(Z^\dagger ZK)$ above one and
$\Qh_\alpha(\psi)\geq\Tr(Z^\dagger ZK)$ below one.
Since the logarithmic prefactor is respectively positive
and negative, both statements bound the divergence by the
corresponding line of \eqref{eq:channel-choi-formula}.

To approach the bound, let $\ket{v}$ be a unit eigenvector of $K$,
corresponding to its largest eigenvalue above one and its smallest
eigenvalue below one. Set
\begin{equation}\label{eq:probe-approximation}
\tau_\eta=(1-\eta)\ket{v}\bra{v}+\eta I_A/d_A,
\qquad Z_\eta=\tau_\eta^{1/2},
\qquad 0<\eta<1.
\end{equation}
Then $Z_\eta$ is invertible and
$\Tr(\tau_\eta K)$ converges to the relevant eigenvalue.
Here $\tau_\eta$ is a density operator on the Choi input leg;
the physical channel-input marginal of the probe is
$\tau_\eta^{\transpose}$ in the fixed basis.
Equation~\eqref{eq:probe-moment} proves the required supremum.
If the smallest eigenvalue is zero, these moments approach
zero and the below-one divergence becomes unbounded; a
maximally entangled probe need not itself have infinite
divergence in this case.

For $\alpha>1$, failure of $J_{\cN}\ll J_{\cM}$ is detected
by the normalized maximally entangled probe and gives infinite
divergence.  In every other case, the preceding proof uses positivity
of $J_{\cM}$ and does not require its partial trace to be the
identity.  Thus the Choi formula applies directly to an arbitrary
completely positive comparison $\cM$; trace-nonincreasing
comparisons are included.  The scalar-rescaling identity is
\begin{equation}\label{eq:comparison-scaling}
\Dh_\alpha(\cN\|c\cM)
=\Dh_\alpha(\cN\|\cM)-\log c,
\qquad c>0.
\end{equation}
It follows from
$P_\alpha(X,cY)=c^{1-\alpha}P_\alpha(X,Y)$, which holds
directly for supported pairs.  Below one, factoring out $c$ from
the regularized second argument reduces the input pair to
$(X+\delta I,Y+\delta I/c)$.  Adding
$\min\{1,1/c\}\delta I$ or $\max\{1,1/c\}\delta I$ to both
arguments bounds this pair from below and above; monotonicity
then gives the same limiting perspective.
No optimizer-attainment assumption is needed.

\paragraph{The order-one limit.}
The same Choi representation gives a direct proof of the channel
endpoint \eqref{eq:bs-limit}. This proof keeps the channel pair fixed
and avoids moving an optimization through a limit. Let
$\cN,\cM\colon A\to B$ be channels satisfying
$X\coloneqq J_{\cN}\ll Y\coloneqq J_{\cM}$, and set
$T=Y^{-1/2}XY^{-1/2}$ on $\supp Y$. Define
\begin{equation}\label{eq:bs-choi-derivative}
Q\coloneqq\Tr_B\!\left[Y^{1/2}T\log T\,Y^{1/2}\right],
\qquad 0\log0\coloneqq0.
\end{equation}
This is a Hermitian operator on the Choi input system $A$.
For $t>0$ tending to zero, apply the scalar Taylor expansion to each
nonzero eigenvalue of $T$; its zero eigenvalues contribute zero
throughout. Finite dimensionality then gives, in operator norm,
\begin{equation}\label{eq:bs-choi-expansion}
K(t)\coloneqq\Tr_BP_{1+t}(X,Y)
=I_A+t(\ln2)Q+O(t^2).
\end{equation}
The constant term is $I_A$ because $Y^{1/2}TY^{1/2}=X$ and
$\Tr_BX=I_A$ for the channel $\cN$. The factor $\ln2$ converts the base-two logarithm in
\eqref{eq:bs-choi-derivative} to the derivative of the power.
To justify the largest-eigenvalue step explicitly, set
$E(t)=K(t)-I_A-t(\ln2)Q$. The variational characterization gives
\begin{equation}\label{eq:bs-eigenvalue-remainder}
\left|\lambda_{\max}(K(t))-
\bigl(1+t(\ln2)\lambda_{\max}(Q)\bigr)\right|
\leq\norm{E(t)}=O(t^2).
\end{equation}
Since $K(t)\geq0$, its operator norm equals its largest eigenvalue.
This estimate does not require the largest eigenvalue of $Q$ to
be simple. Consequently, the Choi formula yields
\begin{equation}\label{eq:bs-choi-limit}
\lim_{t\downarrow0}\Dh_{1+t}(\cN\|\cM)
=\lim_{t\downarrow0}\frac1t\log\norm{K(t)}
=\lambda_{\max}(Q).
\end{equation}

It remains to identify this eigenvalue with the channel divergence
in \eqref{eq:bs-channel}. For every fixed probe, its geometric
divergence is bounded by $\Dh_{1+t}(\cN\|\cM)$. Taking the
fixed-state limit \eqref{eq:bs-state-limit} and then using
\eqref{eq:bs-choi-limit} bounds its Belavkin--Staszewski divergence
by $\lambda_{\max}(Q)$. The pure-probe reduction
\eqref{eq:bs-channel-pure-reduction} gives the same upper bound
on the channel supremum. Conversely, for an invertible probe filter
$Z$ and $\tau=Z^\dagger Z$, equations
\eqref{eq:probe-moment} and \eqref{eq:bs-choi-expansion} give the
fixed-probe moment $1+t(\ln2)\Tr(\tau Q)+O(t^2)$.
Its Belavkin--Staszewski divergence is therefore $\Tr(\tau Q)$.
Choose $\tau_\eta$ from \eqref{eq:probe-approximation}, now taking
$\ket v$ to be a unit eigenvector corresponding to the largest
eigenvalue of $Q$.
These full-rank probes approach $\lambda_{\max}(Q)$, proving
$\Dh(\cN\|\cM)=\lambda_{\max}(Q)$ and hence
\eqref{eq:bs-limit}. If Choi support inclusion fails, the normalized
maximally entangled probe gives infinite divergence for every
$\alpha>1$ and for the Belavkin--Staszewski divergence, completing
that case as well.

\end{document}